\documentclass{article}
\usepackage{fullpage}
\usepackage{amsthm,amssymb,amsmath}  
\usepackage{authblk}
\usepackage{tikz} 

\usepackage{mathtools, eqparbox}%

\usepackage[short,nocomma]{optidef}
\usepackage{cases}
\usepackage{nicefrac}
\usepackage{amsthm,amssymb,amsmath}  
\usepackage{xspace,enumerate}
\usepackage[utf8]{inputenc}
\usepackage{thmtools}
\usepackage{thm-restate}
\usepackage{todonotes}
\usepackage{hyperref}
\usepackage{verbatim}
\usepackage{multirow}
\usepackage{url}
\usepackage{mathtools}
\usepackage{amsthm}
\usepackage[utf8]{inputenc}
\usepackage{color}
\usepackage{caption}
\usepackage{scalerel}
\usepackage{subcaption}
\usepackage{todonotes}
\usepackage[capitalize]{cleveref}
\usepackage{algorithm}
\usepackage[noend]{algpseudocode}
\usepackage[normalem]{ulem}
\usetikzlibrary{calc,shapes,arrows.meta}
\usepackage{booktabs}

  \theoremstyle{plain}
  \newtheorem{theorem}{Theorem}
  \newtheorem{lemma}[theorem]{Lemma}  
  \newtheorem{corollary}[theorem]{Corollary}

  \theoremstyle{definition}
  \newtheorem{definition}{Definition}
  
  \newtheorem{example}{Example}

\newcommand{\cO}{\mathcal{O}}

\newcommand{\Oh}{\mathcal{O}}
\def\dd{\mathinner{.\,.}}
\newcommand{\Occ}{|\text{Occ}|}
\newcommand{\OccSet}{\text{Occ}}
\newcommand{\Count}{\text{Count}}
\newcommand{\Tpref}{\mathcal{T}_{\text{pref}}}
\newcommand{\Tsuff}{\mathcal{T}_{\text{suff}}}
\newcommand{\Ipref}{\mathcal{I}_{\text{pref}}}
\newcommand{\Isuff}{\mathcal{I}_{\text{suff}}}
\newcommand{\diff}{\text{Diff}}
\newcommand{\INDEXING}{\textsc{$\ell$-Weighted indexing}\xspace}
\newcommand{\EFM}{\textsc{EFM}\xspace}
\newcommand{\SARS}{\textsc{SARS}\xspace}

\newcommand{\RS}{\textsc{RSSI}\xspace}

\newcommand{\Z}{\textsc{HUMAN}\xspace}
\newcommand{\OUTPUT}{|\text{Occ}|\xspace}
\newcommand{\WST}{\textsf{WST}\xspace}
\newcommand{\MWST}{\textsf{MWST}\xspace}
\newcommand{\MWSTG}{\textsf{MWST-G}\xspace}
\newcommand{\WSA}{\textsf{WSA}\xspace}
\newcommand{\MWSA}{\textsf{MWSA}\xspace}
\newcommand{\MWSAG}{\textsf{MWSA-G}\xspace}
\newcommand{\MWSTSE}{\textsf{MWST-SE}\xspace}

\title{Space-Efficient Indexes for Uncertain Strings}

\author[1]{Esteban Gabory}
\author[2]{Chang Liu}
\author[3]{Grigorios Loukides}
\author[1,4]{Solon P.\ Pissis}
\author[1]{Wiktor Zuba}

\affil[1]{CWI, Amsterdam, The Netherlands}
\affil[2]{Zhejiang University, Zhejiang, China}
\affil[3]{King's College London, London, UK}
\affil[4]{Vrije Universiteit, Amsterdam, The Netherlands}

\date{\today}

\begin{document}

\maketitle

\begin{abstract}
Strings in the real world are often encoded with some level of uncertainty,
for example, due to: unreliable data measurements; flexible sequence modeling; or noise introduced for privacy protection. 
In the \emph{character-level uncertainty model}, an \emph{uncertain string} 
$X$ of length $n$ on an alphabet $\Sigma$ is a sequence of $n$ probability distributions over $\Sigma$. Given an uncertain string $X$ and a weight threshold $\frac{1}{z}\in(0,1]$, we say that pattern $P$ occurs in $X$ at position $i$, if the product of probabilities of the letters of $P$ at positions $i,\ldots,i+|P|-1$ is at least $\frac{1}{z}$. While indexing standard strings for online pattern searches can be performed in linear time and space, 
indexing uncertain strings is much more challenging. 
Specifically, the state-of-the-art index for uncertain strings has $\cO(nz)$ size, requires $\cO(nz)$ time and $\cO(nz)$ space to be constructed, and answers pattern matching queries in the optimal $\cO(m+\OUTPUT)$ time, where $m$ is the length of $P$ and $\OUTPUT$ is the total number of occurrences of $P$ in $X$. For large $n$ and (moderate) $z$ values, this index is completely impractical to construct, which outweighs the benefit of the supported optimal pattern matching queries. We were thus motivated to design a space-efficient index at the expense of slower yet competitive pattern matching queries. We show that when we have at hand a lower bound $\ell$ on the length of the supported pattern queries, as is often the case in real-world applications, we can slash the index size \emph{and} the construction space roughly by $\ell$. 
In particular, we propose an index of $\cO(\frac{nz}{\ell}\log z)$ expected size, which can be constructed using $\cO(\frac{nz}{\ell}\log z)$ expected space, and  supports very fast pattern matching queries in expectation, for patterns of length $m\geq \ell$. 
We have implemented and evaluated several versions of our index. 
The best-performing version of our index is up to \emph{two orders of magnitude smaller} than the state of the art in terms of \emph{both} index size and construction space, while offering faster or very competitive query and construction times.  
\end{abstract}

\section{Introduction}

%Strings and Indexing
A large portion of the data feeding real-world database systems, including bioinformatics systems~\cite{DBLP:journals/nar/Rodriguez-TomeSCF96}, Enterprise Resource Planning (ERP) systems~\cite{DBLP:conf/edbt/0002RF14}, or Business Intelligence (BI) systems~\cite{DBLP:conf/sigmod/VogelsgesangHFK18}, is textual; namely, 
\emph{sequences} of letters over some alphabet (also known as \emph{strings}).
This happens because strings can easily encode data arising from different sources such as: sequences of nucleotides read by DNA sequencers (e.g., short or long DNA reads)~\cite{NGS}; natural language text generated by humans (e.g., description or comment fields)~\cite{DBLP:books/lib/JurafskyM09}; id's generated by software (e.g., URLs, email addresses, or IP addresses)~\cite{DBLP:journals/pvldb/Boncz0L20}; or discretized measurements generated by sensors (e.g., RSSI, EEG or EMG data)~\cite{DBLP:books/sp/mmsd2013,error_RSSI}. 

Given the ever increasing size of string data, 
it is crucial to represent them in a concise form but also to \emph{simultaneously} allow efficient pattern searches. This goal is formalized by the classic \emph{text indexing} problem~\cite{DBLP:books/daglib/0020103}: preprocess a string $T$ of length $n$ over an alphabet $\Sigma$ 
of size $\sigma$, known as the \emph{text}, into a data structure that supports pattern matching queries; i.e., report the set $\OccSet$ of all $\OUTPUT$ positions in $T$ where an occurrence of a string $P$, known as the \emph{pattern}, begins.

%Indexing measures of efficiency
In text indexing we are interested in four efficiency measures~\cite{DBLP:journals/pvldb/AyadLP23,DBLP:conf/soda/KempaK23}: \textbf{(i)} How much space does the final index occupy for a string $T$ of length $n$ (the \emph{index size})? \textbf{(ii)} How fast can we answer a query $P$ of length $m$ (the \emph{query time})? \textbf{(iii)} How much working space do we need to construct the index (the \emph{construction space})? \textbf{(iv)} How fast can we construct the index (the \emph{construction time})? The classic \emph{suffix tree} index~\cite{DBLP:conf/focs/Weiner73} has size linear in $n$ and optimal query time  $\cO(m+\OUTPUT)$; the  construction time and space are also linear in $n$~\cite{DBLP:conf/focs/Farach97}. 

%Uncertainty in strings
In the real world, strings are often encoded with some level of uncertainty; for example, due to: (i) imprecise, incomplete or unreliable data measurements, such as sensor measurements, RFID measurements or trajectory measurements~\cite{DBLP:series/ads/Aggarwal09}; (ii) deliberate flexible sequence modeling, such as the representation of a \emph{pangenome}, that is, a collection of closely-related genomes to be analyzed together~\cite{PanGenomeConsortium18}; or (iii) the existence of confidential information in a dataset which has been distorted deliberately for privacy protection~\cite{DBLP:conf/icde/Aggarwal08}.

While there are many practical solutions for text indexing~\cite{DBLP:journals/siamcomp/ManberM93,DBLP:journals/jacm/FerraginaM05,DBLP:journals/spe/GrabowskiR17,DBLP:journals/jacm/GagieNP20,DBLP:journals/pvldb/AyadLP23} and also for answering different types of queries on various types of uncertain data (see Section~\ref{sec:related}), \emph{practical indexing schemes} for uncertain strings are rather undeveloped. In response, our work makes an important step towards developing such practical space-efficient indexes.

%Our model
\subsection{Our Data Model and Motivation}

We use the standard \emph{character-level uncertainty model}~\cite{DBLP:conf/sigmod/JestesLYY10}. An \emph{uncertain string} (or \emph{weighted string}) 
$X$ of length $n$ on an alphabet $\Sigma$ is a sequence of $n$ sets. Every $X[i]$, for all $i\in[1,n]$, is a set of $|\Sigma|$ ordered pairs $(\alpha,p_i(\alpha))$, where $\alpha\in \Sigma$ and $p_i(\alpha)$ is the probability that letter $\alpha$ occurs at position $i$.

\begin{example}\label{ex:weighted_string}
The table below shows a weighted string $X$, with $n=6$ and $\Sigma=\{\texttt{A},\texttt{B}\}$, represented as a $|\Sigma| \times n$ matrix.
$$\begin{array}{c|c|c|c|c|c|c}
     & 1 & 2 & 3 & 4 & 5 & 6 \\
\hline
    \texttt{A} & 1 & 1/2 & \textbf{3/4} & 4/5 & \textbf{1/2} & 1/4 \\
    \texttt{B} & 0 & 1/2 & 1/4 & \textbf{1/5} & 1/2 & 3/4
\end{array}$$    
\end{example}

The data model of~\cite{DBLP:conf/sigmod/JestesLYY10} has been employed by many works~\cite{DBLP:conf/cpm/AmirCIKZ06,DBLP:journals/bioinformatics/KorhonenMPRU09,DBLP:journals/tcbb/PizziRU11,DBLP:journals/pvldb/GeL11,DBLP:conf/sdm/LiBKP14,DBLP:conf/isaac/KociumakaPR16}.
In bioinformatics, for example, weighted strings are known as \emph{position weight matrices}~\cite{PWM}.
As in these works, we define the \emph{occurrence probability} of pattern $P=\texttt{ABA}$ at position $3$ in $X$ of \cref{ex:weighted_string} as $3/4\times 1/5\times 1/2 = 3/40$ (shown in bold). 

\noindent{\bf Weighted Indexing.}~The \emph{Weighted Indexing} problem is defined as follows~\cite{DBLP:journals/tcs/AmirCIKZ08}: 
Given a weighted string $X$ of length $n$ on an alphabet $\Sigma$ of size $\sigma$ and a weight threshold $\frac{1}{z} \in (0,1]$, preprocess $X$ into a \emph{compact} data structure (\emph{the index}) that supports \emph{efficient} pattern matching queries; i.e., report all positions in $X$ where $P$ occurs with probability at least $\frac{1}{z}$.

\noindent{\bf State of the Art.}~The indexing problem on weighted strings has attracted a lot of attention by the theory community~\cite{DBLP:journals/fuin/IliopoulosMPPTT06,DBLP:journals/tcs/AmirCIKZ08,DBLP:conf/edbt/BiswasPTS16,DBLP:conf/cpm/BartonKPR16,DBLP:conf/edbt/BiswasPTS16,DBLP:journals/iandc/BartonK0PR20,DBLP:journals/jea/Charalampopoulos20}, culminating in the following bounds~\cite{DBLP:conf/cpm/BartonKPR16,DBLP:journals/iandc/BartonK0PR20}:
there exists an index of $\cO(nz)$ size supporting optimal $\cO(m+\Occ)$-time queries for any pattern of length $m$; it can be constructed in $\cO(nz)$ time using $\cO(nz)$ space. We call this index the \emph{weighted suffix tree} (\WST).

%Motivation
\noindent{\bf Our Motivation.}~Although these bounds are very appealing from a theoretical perspective, from a practical perspective, $\cO(nz)$ size and construction space are \emph{prohibitive} for large-scale datasets. Say we have an input weighted string of length $n=10^9$ bytes, that $z=100$, and that the constant in $\cO(nz)$ is something small, like $20$ bytes. Then we need around 2TBs of RAM to store the index for an input of 1GB! \emph{We were thus motivated to seek space-query time trade-offs for indexing weighted strings.}
In particular, we seek a parameterized version of \emph{Weighted Indexing} for which we can have indexes of size smaller than $\cO(nz)$. Ideally, we would also like to construct these smaller indexes using less than $\cO(nz)$ space. A good such input parameter is a lower bound $\ell$ on the length $m$ of any queried pattern. It is arguably a reasonable assumption to know $\ell$ in advance.
For instance, in bioinformatics~\cite{Winnowmap2,Wenger2019,Logsdon2020}, the length of sequencing reads (patterns) ranges from a few hundreds to 30,000~\cite{Logsdon2020}. Even when at most $k$ errors must be accommodated for matching, at least one out of $k+1$ fragments must be matched exactly. In natural language processing, the queried patterns can also be long~\cite{TOIS17}. Examples of such patterns are queries in question answering systems~\cite{longqueries0}, \emph{description queries} in TREC datasets~\cite{longqueries1,longqueries2}, and  representative phrases in documents~\cite{representativephrases1}. Similarly, a query pattern can be long when it  encodes an entire document (e.g., a webpage in the context of deduplication~\cite{webpageduplicate}), or  machine-generated messages~\cite{machinegenerated}.

\subsection{Our Techniques and Results}

\begin{figure}[t]
     \centering
     \includegraphics[width=0.7\textwidth]{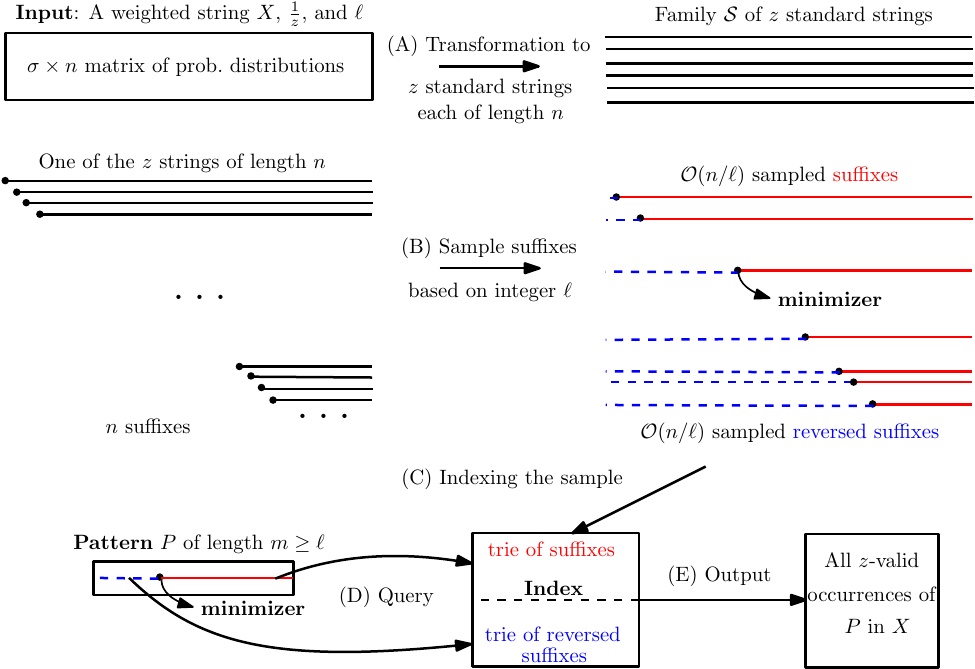}
     \caption{An informal overview of our techniques: Given a weighted string $X$ of length $n$ over an alphabet of size $\sigma$, a weight threshold $\frac{1}{z}$, and an integer $\ell$, we \textbf{(A)} use the algorithm from~\cite{DBLP:conf/cpm/BartonKPR16} to construct a family $\mathcal{S}$ of $z$  
     standard strings, each of length $n$. \textbf{(B)} For each such string, we consider all of its $n$ suffixes and sample them for the given integer $\ell$ using the minimizers mechanism~\cite{DBLP:journals/bioinformatics/RobertsHHMY04,DBLP:conf/sigmod/SchleimerWA03}. These suffixes imply a set of $\cO(n/\ell)$ suffixes and a set of $\cO(n/\ell)$ reversed suffixes in expectation. \textbf{(C)} We then index these suffixes in two suffix trees~\cite{DBLP:conf/focs/Weiner73}, which we link using a 2D grid, so as to answer pattern matching queries for patterns of length at least $\ell$. \textbf{(D)} When such a pattern $P$ of length $m$ is given, we find its leftmost minimizer, which implies a suffix and a reversed suffix of $P$, and query those using our index. Our index efficiently merges the partial results (i.e., occurrences of suffixes and reversed suffixes); and \textbf{(E)} outputs all $z$-valid occurrences of $P$ in $X$. The size of the resulting index is $\cO(\frac{nz}{\ell}\log z)$. The extra multiplicative $\log z$ factor comes from our representation of the edge labels in the suffix trees.}
     \label{fig:framework}
\end{figure}

Let us first recall $\WST$, the state-of-the-art index. In~\cite{DBLP:conf/cpm/BartonKPR16}, Barton et al. showed that, for any weight threshold $\frac{1}{z}$, a weighted string $X$ of length $n$ can be transformed into a family $\mathcal{S}$ of $z$ standard strings (each of length $n$), so that a pattern $P$ occurs in $X$ at position $i$ with probability $p$ only if $P$ occurs at position $i$ in $\lfloor p\cdot z \rfloor$ strings from $\mathcal{S}$. The authors have also shown how to index $\mathcal{S}$ using (a modified version of) suffix trees~\cite{DBLP:conf/focs/Weiner73} resulting in $\WST$: an index of total size $\cO(nz)$ supporting pattern matching queries in the optimal $\cO(m+\OUTPUT)$ time. A more space-efficient, array-based version was also presented in~\cite{DBLP:journals/jea/Charalampopoulos20}; it is known as the \emph{weighted suffix array} (\WSA).

Here we make the following three main contributions (see \cref{fig:framework} for an informal overview of our techniques):

\begin{enumerate}
    \item We show how to 
\underline{slash the size} of both \WST and \WSA roughly by $\ell$, while still supporting very fast queries in expectation for any pattern $P$ of length $m\geq \ell$, by 
combining the \emph{minimizers} sampling mechanism~\cite{DBLP:journals/bioinformatics/RobertsHHMY04,DBLP:conf/sigmod/SchleimerWA03}, \emph{suffix trees}~\cite{DBLP:conf/focs/Weiner73}, several combinatorial and probabilistic arguments, and a \emph{geometric data structure} (2D grid)~\cite{DBLP:conf/latin/MakinenN06}. The size of the resulting index is $\cO(\frac{nz}{\ell}\log z)$. The extra multiplicative $\log z$ factor comes from a new way to encode the edge labels in the suffix trees. The main novelty of our approach is the combination of minimizers and the new edge encoding allowing us to delete the $z$ strings after the construction, thus resulting in a significantly smaller index size for large $\ell$ values.

   \item Our technique still requires us to first construct the family $\mathcal{S}$ of the $z$ strings, which in any case gives an index with $\cO(nz)$ construction space. The challenge is how to construct the index \emph{without}
constructing the $z$ strings. To achieve this, we develop a fast, highly non-trivial, algorithm for constructing the \emph{same} index \emph{without generating $\mathcal{S}$ explicitly}. Our new algorithm samples an implicit representation of $\mathcal{S}$ using minimizers outputting the final index directly. The construction space of the resulting index \underline{matches its size}: $\cO(\frac{nz}{\ell}\log z)$. This is the most technically involved result of our paper.

    \item Following the different paradigms of suffix trees~\cite{DBLP:conf/focs/Weiner73} and suffix arrays~\cite{DBLP:journals/jacm/KarkkainenSB06} in the classic setting of standard strings, we have implemented tree and array-based versions of our index underlying Contributions 1 and 2. 
    The results show that our indexes are up to \emph{two orders of magnitude smaller} than the state of the art in terms of \emph{both} index size and construction space. For instance, for indexing $1,432$ bacterial samples, with $\ell=1024$ and $z=128$, which are reasonable in applications, our space-efficient index has size $640$MBs and needs only $772$MBs of memory to be constructed, while \WSA has size $7$GBs and needs $32$GBs of memory to be constructed and \WST has size $126$GBs and needs $241$GBs of memory to be constructed! Our results also show that our array-based indexes outperform the tree-based ones, offering \emph{very competitive query times and construction times} to those of the state of the art. 
    % They also show that our array-based indexes outperform the tree-based ones, offering \emph{very competitive query times and construction times} to those of the state-of-the-art indexes. For example, for indexing a collection of $1,432$ bacterial samples, with $\ell=256$ and $z=128$, which are reasonable in applications, our space-efficient index has size $640$MBs and needs only $772$MBs of memory to be constructed, while \WSA has size $7$GBs and needs $32$GBs of memory to be constructed! Furthermore, compared to \WSA, our space-efficient index takes $40\%$ less time to be constructed and its query time is $80$ microseconds on average (over about 1.9M queries), while that for \WSA is $81$ microseconds. 
\end{enumerate}

\subsection{Paper Organization}

In Section~\ref{sec:prel}, we provide the necessary background. 
In Section~\ref{sec:index}, we present our index.
In Section~\ref{sec:space-efficient-construction}, we present the space-efficient algorithm for constructing our index.
In Section~\ref{sec:fast_pm}, we present a more practical algorithm for querying our index. 
In Section~\ref{sec:related}, we discuss related work.
In Section~\ref{sec:exp}, we provide an extensive experimental evaluation of our algorithms.
We conclude in Section~\ref{sec:discussion} with a discussion on limitations and future work.

\section{Preliminaries and Problem Definition} \label{sec:prel}
\noindent{\bf Strings.}~An \emph{alphabet} $\Sigma$ of size $\sigma=|\Sigma|$ is a nonempty set of elements called \emph{letters}. By $\Sigma^k$ we denote the set of all length-$k$ strings over $\Sigma$. By $\varepsilon$ we denote the \emph{empty string} of length $0$.
For a \emph{string} $S=S[1] \cdots S[n]$ over $\Sigma$, by $n=|S|$ we denote its \emph{length}. 
The fragment $S[i\dd j]$ of $S$ is an \emph{occurrence} of the underlying \emph{substring} $P=S[i]\cdots S[j]$. 
We also say that $P$ occurs at \emph{position} $i$ in $S$. 
A {\em prefix} of $S$ is a substring of $S$ of the form $S[1\dd j]$ and a {\em suffix} of $S$ is a substring of $S$ of the form $S[i\dd n]$. Given a string $S=S[1]\cdots S[n]$, its \emph{reverse} is the string $S^r=S[n]\cdots S[1]$. For any two strings $S_1$ and $S_2$ of the same length, we define their \emph{Hamming distance} $d_H(S_1,S_2)$ as their total number of mismatching positions. 

\noindent{\bf Sampling.}~Given a fixed pair of positive integers $(\ell, k )$ s.t. $\ell \ge k$, we call a function 
$f:\Sigma^{\ell}\to [1, \ell-k+1]$ 
that selects the starting position of a length-$k$ fragment, for any string of length $\ell$, an $(\ell,k)$-\emph{local scheme}. 
We call the set $\mathcal{M}_{f}(S)=\{i+f(S[i\dd i+\ell-1])-1~|~1\le i \le |S|-\ell+1\}$, for an $(\ell,k)$-\emph{local scheme} $f$ on a string $S$, \emph{the set of selected indices}. An $(\ell,k)$-\emph{minimizer scheme} is an  $(\ell,k)$-local scheme that selects the position of the leftmost occurrence of the smallest length-$k$ substring, for a fixed $k$ and a fixed order on $\Sigma^k$. In that case, we call \emph{minimizers} the selected indices~\cite{DBLP:journals/bioinformatics/RobertsHHMY04,DBLP:conf/sigmod/SchleimerWA03}. 
The minimizer scheme can be based on a lexicographic order. 

\begin{example}
Let $S=\texttt{ABAABB}$, $\ell=4$, and $k=2$.
We obtain $\mathcal{M}_{f}(S)=\{3\}$ because $S[3\dd 4]=\texttt{AA}$ is the lexicographically smallest length-$2$ substring in all length-$4$ fragment of $S$.
\end{example}

The minimizer scheme can also be specified by a hash function, e.g., Karp-Rabin fingerprints~\cite{DBLP:journals/ibmrd/KarpR87}. 

\begin{definition}\label{def:density}
Let $f$ be an $(\ell,k)$-minimizer scheme. The \emph{specific density} of $f$ on $S$ is the value $|\mathcal{M}_{f}(S)|/|S|$. The \emph{density} of $f$ is the expected specific density on a sufficiently long random string (with letters chosen independently at random).
\end{definition}

\begin{lemma}[\cite{10.1093/bioinformatics/btaa472}]\label{lem:minimizer density}
The density of an $(\ell,k)$-minimizer scheme on alphabet $\Sigma$ with $k\ge \log_{|\Sigma|} \ell+c$ is $\cO(\frac{1}{\ell})$, for some $c=\cO(1)$.\end{lemma}

\noindent{\bf Weighted Strings.}~A \emph{weighted string} $X$ of length $n$ on an alphabet $\Sigma$ is a sequence of $n$ sets $X[1],\ldots,X[n]$. 
Every $X[i]$, for all $1\le i \le n$, is a set of $|\Sigma|$ ordered pairs $(\alpha,p_i(\alpha))$, where $\alpha\in \Sigma$ is a letter and $p_i(\alpha)$ is the probability of having $\alpha$ at position $i$ of $X$. 
Formally, $X[i]=\{(\alpha,p_i(\alpha)) \mid \alpha \in \Sigma\}$, where for every $\alpha\in\Sigma$ we have $p_i(\alpha)\in [0,1]$, and $\sum_{\alpha\in\Sigma}p_i(\alpha)=1$.  
A letter $\alpha$ \emph{occurs} at position $i$ of a weighted string $X$ if and only if $p_i(\alpha)$, the \emph{occurrence probability} 
of $\alpha$ at position $i$, $p_i(\alpha)$, is greater than $0$. A string $U$ of length $m$ is a \emph{factor} of a weighted string $X$ if and only if it occurs at some starting position $i$ with \emph{occurrence probability} $\mathbb{P}(X[i\dd i+m-1]=U)=\Pi_{j=1}^{m}p_{j+i-1}(U[j])>0$. 
Given a \emph{weight threshold} $1/z$ $\in (0,1]$, we say that factor $U$ is $z$-\emph{solid} (or $z$-\emph{valid}) or equivalently that factor $U$ has a $z$-solid occurrence in $X$ at some position $i$, if $\mathbb{P}(X[i\dd i+m-1]=U)\geq 1/z$. When the context is clear we may simply say \emph{solid} (or \emph{valid}). We say that factor $U$ is \emph{(right-)maximal} at position $i$ of $X$ if $U$ has a solid occurrence at position $i$ of $X$ and no string $U'=U\alpha$, for any $\alpha\in\Sigma$, has a solid occurrence at position $i$ of $X$.
For a weighted string $X$, a pattern $P$, and a weight threshold $1/z\in(0,1]$, $\OccSet_{1/z}(P,X)$ is the set of starting positions of valid occurrences of $P$ in $X$. A \emph{property} $\Pi$ of a string $S$ is a hereditary collection of integer intervals contained in $[1, n]$.\footnote{A collection that contains all the subintervals of its elements.} For simplicity, we represent every property $\Pi$ with an array $\pi[1\dd |S|]$ such that the longest interval $I\in \Pi$ starting at position $i$ is $[i, \pi[i]]$. Observe that $\pi$ can be an arbitrary array satisfying $\pi[i]\in[i-1, n]$, and $\pi[1]\le \pi[2]\le \cdots \le \pi[n]$ (where $\pi[i]=i-1$ means that $i$ is not contained in any interval $I\in\Pi$). For a string $P$, by $\OccSet_{\pi}(P,S)$ we denote the set of occurrences of $P$ in $S$ such that $i+|P|-1\le \pi[i]$. 

\begin{example}
Consider the string-property pair $(S_2,\pi_2)$ in Table~\ref{tab:sec2}. The pattern $P=\texttt{AAA}$ occurs at position $i=3$ because $i+|P|-1 = 3 + 3 - 1 \le \pi_2[3]=5$.  
\end{example}

Let us consider an indexed family $\mathcal{S}=(S_j,\pi_j)_{j=1}^k$ of strings $S_j$ with properties $\pi_j$. For a string $P$ and an index $i$, by $\Count_S(P,i)=|\{j\in[1,\dd,k]\mid i\in\OccSet_{\pi_j}(P,S_j)\}|$ we denote the total number of occurrences of $P$ at position $i$ in the strings $S_1,\ldots,S_k$ of $\mathcal{S}$ that respect the properties. 
We say that an indexed family $\mathcal{S}=(S_j,\pi_j)_{j=1}^{\lfloor z \rfloor}$ is a \emph{$z$-estimation} of a weighted string $X$ of length $n$ if and only if, for every string $P$ and position $i\in[1, n]$, $\Count_S(P,i)= \lfloor\mathbb{P}(X[i\dd i+|P|-1]=P)\cdot z\rfloor$. 
The following result has been shown by Barton et al.:

\begin{theorem}[\cite{DBLP:journals/iandc/BartonK0PR20}]\label{theorem:z-estimation}
For any weighted string $X$ of length $n$ and any weight threshold $\frac{1}{z}$, $X$ has a $z$-estimation of size $\cO(nz)$ constructible in $\cO(nz)$ time using $\cO(nz)$ space.
\end{theorem}

\begin{example}\label{ex:running}
For $\frac{1}{z}=\frac{1}{4}$, the weighted string $X$ in \cref{ex:weighted_string} admits the $4$-estimation $\mathcal{S}$ in Table~\ref{tab:sec2}, given by Theorem~\ref{theorem:z-estimation}.  
\begin{figure}[t]
    %\centering
    \begin{minipage}{.44\linewidth}\footnotesize
    \captionof{table}{A $4$-estimation $\mathcal{S}$ of $X$ from \cref{ex:weighted_string}.}\label{tab:sec2}
\begin{tabular}{c|cccccc} \footnotesize
$i$ & $1$ & $2$ & $3$ & $4$ & $5$ & $6$ \\
\hline
$S_1$ & $\texttt{A}$ & $\texttt{A}$ & $\texttt{A}$ & $\texttt{A}$ & $\texttt{A}$ & $\texttt{A}$ \\
$\pi_1$ & $2$ & $2$ & $3$ & $4$ & $5$ & $6$\\
\hline
$S_2$ & $\underline{\texttt{A}}$ & $\underline{\texttt{A}}$ & $\underline{\texttt{A}}$ & $\underline{\texttt{A}}$ & $\texttt{A}$ & $\texttt{B}$  \\
$\pi_2$ & 4 & 4 & 5 & 6 & 6 & 6\\
\hline
$S_3$ & $\underline{\texttt{A}}$ & $\texttt{B}$ & $\underline{\texttt{A}}$ & $\underline{\texttt{A}}$ & $\texttt{B}$ & $\texttt{B}$  \\
$\pi_3$ & 4 & 4 & 5 & 6 & 6 & 6\\
\hline
$S_4$ & $\texttt{A}$ & $\texttt{B}$ & $\texttt{B}$ & $\texttt{B}$ & $\texttt{B}$ & $\texttt{B}$  \\
$\pi_4$ & $2$ & $2$ & $3$ & $3$ & $5$ & $6$
\end{tabular}
 \end{minipage}%
    \hspace{+0.9cm}
    \begin{minipage}{.41\linewidth}
    \includegraphics[scale=0.53]{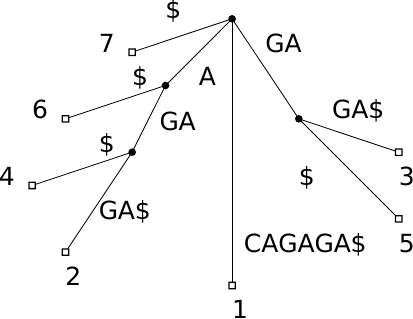}
    \captionof{figure}{Suffix tree of $S=\texttt{CAGAGA\$}$.}\label{fig:ST}
    \end{minipage}
\end{figure}

For pattern $P=\texttt{AB}$ and $S_3$, we have that $\OccSet_{\pi_3}(P,S_3)=\{1,4\}$ because $P$ occurs at position $1$, with $1+|P|-1\leq \pi_3[1]=4$, and at position $4$, with $4+|P|-1\leq \pi_3[4]=6$.

For pattern $P=\texttt{AB}$ and $i=1$, we have that
$\mathbb{P}(X[i\dd i+|P|-1]=P)=1\cdot 1/2=1/2$ and so $P$ occurs in
$\Count_S(P,i)=\lfloor\mathbb{P}(X[i\dd i+|P|-1]=P)\cdot z\rfloor=\lfloor (1/2)\cdot 4\rfloor =2$
strings of the $4$-estimation at position $1$ (namely, strings $S_3$ and $S_4$).

We construct the set of (lexicographic) minimizers that respect the property, for $\ell=3$ and $k=2$, for every $S_j$, $j\in [1,4]$, from $\mathcal{S}$; namely for $S_j\in \mathcal{S}$ we compute $f(S_j[i\dd i+2])$ if and only if $i+2 \le \pi_j[i]$. We underline the positions of the minimizers. Note that we have selected no minimizer in $S_1$ or $S_4$ as they have no solid factor of length $3$.
\end{example}

\noindent{\bf Problem Definition.}~In our \textsc{$\ell$-Weighted Indexing} problem, we are given a weighted string $X$ of length $n$ over an alphabet $\Sigma$, a weight threshold $\frac{1}{z}\in(0,1]$, and an integer $\ell>0$, and we are asked to preprocess them in a \emph{compact} data structure (the index) to support the following queries \emph{efficiently}: For any string $P$ of length $m\ge\ell$, report all elements of $\OccSet_{\frac{1}{z}}(P,X)$. Other than the \emph{index size} and the \emph{query time}, we seek to minimize the \emph{construction time} and the \emph{construction space}.

\noindent{\bf Suffix Tree.}~The classic indexing solution for standard strings is the suffix tree~\cite{DBLP:conf/focs/Weiner73}.
Given a set $\mathcal{F}$ of strings, the \emph{compacted trie}
of these strings is the trie obtained by compressing each path of nodes of degree one in the trie of the strings in $\mathcal{F}$, which takes $\cO(|\mathcal{F}|)$ space~\cite{DBLP:journals/jacm/Morrison68}. Each edge in the compacted trie has a label represented as a fragment of a string in $\mathcal{F}$. The \emph{suffix tree} $\textsf{ST}(S)$ is the compacted trie of the suffixes of string $S$. Assuming $S$ ends with a unique terminating letter $\$$, every suffix $S[i\dd |S|]$ is represented by a leaf decorated by index $i$; see Fig.~\ref{fig:ST} for an example. The suffix tree occupies $\cO(|S|)$ space and it can be constructed in $\cO(|S|)$ time~\cite{DBLP:conf/focs/Weiner73,DBLP:conf/focs/Farach97}. It supports queries in $\cO(m+\Occ)$ time. The \emph{suffix array} of $S$ is the list of leaf nodes read using a standard in-order DFS traversal on $\textsf{ST}(S)$. As this can be done in $\cO(|S|)$ time, the suffix array of $S$ can also be computed in $\cO(|S|)$ time.

We next summarize the crucial notation of our paper.

\vspace{+1mm}
\begin{center}
{\footnotesize
\begin{tabular}{|l|l|}\hline
       $\Sigma$ & Alphabet \\ 
    $\sigma$ & Size of the alphabet \\ 
    $X$ & Weighted string \\
    $n$ & Length of $X$\\
    $1/z$& Weight threshold\\
    $P$ & Pattern \\
    $m$ & Length of $P$\\ %queried pattern \\
    $\ell$& Lower bound on the length $m$ of $P$\\
    $p_i$ & Probability~distribution at $X[i]$\\
    $X[i]$ & $i$th set $\{(\alpha,p_i(\alpha))~|~\alpha\in\Sigma\}$ of $X$\\
    $\mathbb{P}(X[i\dd j]=P)$ & Probability $\Pi_{h=0}^{j-i}p_{i+h}(P[h+1])$\\
    $\mathcal{M}_X$& Minimizers for solid factors of $X$ \\
    $k$ &Length of minimizer substrings\\
    $H_X$& Heavy string of $X$\\
    $\mathcal{S}$& Property array $(S_j,\pi_j)_{j=1}^z$\\
    $\pi_i[j]=h$& $S_i[j\dd h]$ is a solid factor\\
    $\OUTPUT$ & Number of occurrences in the output\\
    $\OccSet_{1/z}(P,X)$ & The $z$-solid occurrences of $P$ in $X$\\
    $\Tsuff$ & Forward minimizer solid factor tree\\
    $\Tpref$ & Backward minimizer solid factor tree\\   
    $\Isuff(P)$& Leaf nodes in $\Tsuff$ after reading $P$\\
    $\Ipref(P)$& Leaf nodes in $\Tpref$ after reading $P$\\
    $\mathcal{N}(P)$ & Query rectangle for $P$ in the 2D grid\\
    \hline 
\end{tabular}
}
\end{center}

\section{The New Index: Minimizer-based \WST}\label{sec:index}

In this section, we describe our index for solving \INDEXING. We assume random access to $X$ (we can discard $X$ at the end of the construction). To simplify the analysis we assume $X$ is over an alphabet of size $\sigma=\cO(1)$.

\vspace{+1mm}
\noindent{\bf Main Idea.}~We start the index construction by building a $z$-estimation of $X$, whose total size is $\Theta(nz)$ (\cref{theorem:z-estimation}). We then use minimizers sampling to select $\cO(\frac{nz}{\ell})$ positions (\cref{lem:minimizer density}) of the $z$-estimation,
where $\ell$ is a predetermined lower bound on the length of the supported queries. Next, we construct two suffix trees, called \emph{minimizer solid factor trees}: (1) the compacted trie \emph{of all suffixes} of the solid factors in the $z$-estimation starting at the minimizer positions, and (2) the compacted trie \emph{of all the reversed prefixes} of the solid factors in the $z$-estimation ending at the minimizer positions.  To reduce the size of both trees, we discard the $z$-estimation using a combinatorial observation (\cref{lem:heavy-string,cor:extra info}) that allows us to store only $\cO(\log z)$ information to label a suffix tree edge.
This concludes the construction of the minimizer solid factor trees (\cref{lem:minimizer tree size}). After that, we pair up the leaf nodes corresponding to the same minimizer position, from one of these trees to the other, using  a 2D grid for \emph{range reporting}~\cite{DBLP:conf/compgeom/ChanLP11} (\cref{lem:checking rectangles,lem:2demptiness}). This results in an index of expected total size $\cO(n+\frac{nz}{\ell}\log z)$. Finally, we show how to query the index efficiently by using a probabilistic argument on the number of expected points returned by the 2D grid (\cref{lem: probabilistic bound}). We thus arrive to Contribution 1 (\cref{the:2d-structure}).

\vspace{+1mm}
\noindent{\bf Minimizer Solid Factor Trees.}~Let us fix a weighted string $X$ of length $n$ over an alphabet $\Sigma$ and a weight threshold $\frac{1}{z}$. We first define a \emph{forward solid factor tree} (resp.~\emph{backward solid factor tree}) for $X$ as the suffix tree for the set of maximal solid factors (resp.~the set of reversed solid factors) in $X$. By Theorem~\ref{theorem:z-estimation}, we know that each such solid factor appears in a $z$-estimation of size $\cO(nz)$, and therefore both solid factor trees have size $\cO(nz)$ as well. This argument also gives a method to construct the solid factor trees~\cite{DBLP:journals/iandc/BartonK0PR20}. 

We adapt the solid factor trees to make them more space-efficient for \INDEXING by employing minimizer schemes. Let us fix $\ell$, $k$ and an $(\ell,k)$-minimizer scheme $f$. We can assume throughout, from Lemma~\ref{lem:minimizer density}, that $\ell$ and $k$ are chosen so that the density of $f$ (see Def.~\ref{def:density}) is $\cO(\frac{1}{\ell})$. 
We then construct a $z$-estimation $\mathcal{S}=(S_j,\pi_{j})_{j=1}^{\lfloor z\rfloor}$ of $X$ using Theorem~\ref{theorem:z-estimation} and compute the set $\mathcal{M}_X$ of minimizers from $\mathcal{S}$ respecting the property; namely for $S_j\in \mathcal{S}$ we compute $f(S_j[i\dd i+\ell -1])$ if and only if $i+\ell-1 \le \pi_j[i]$. 

We represent each minimizer in $\mathcal{M}_X$ by a pair $(i,j)$, where $i$ is the minimizer position in the string $S_j\in \mathcal{S}$. In the following, we consider $\mathcal{M}_X$ fixed with $|\mathcal{M}_X|= \cO(\frac{nz}{\ell})$, as 
by Lemma~\ref{lem:minimizer density} there are in expectation $\cO(\frac{nz}{\ell})$ minimizers in $\mathcal{S}$.

Based on $\mathcal{M}_X$, we define a \emph{minimizer} forward (resp.~backward) solid factor tree as a compacted trie containing suffixes of solid factors (resp.~of reversed solid factors) \emph{starting} at position $i$ from a string $S_j\in\mathcal{S}$ with $(i,j)\in\mathcal{M}_X$. Each leaf has a minimizer label $(i,j)\in \mathcal{M}_X$ associated to the corresponding suffix. If one same suffix corresponds to several such labels (it occurs at several minimizers from $\mathcal{S}$), we add one copy of the leaf for each such label. Since $|\mathcal{M}_X|=\cO(\frac{nz}{\ell})$, the minimizer solid factor trees contain $\cO(\frac{nz}{\ell})$ leaves, and therefore nodes. 

Still the size of the $z$-estimation $\mathcal{S}$ is, by definition, always $\Theta(nz)$, which makes the total size of the index $\cO(\frac{nz}{\ell})+\Theta(nz)=\Theta(nz)$. We avoid this by employing the following crucial combinatorial observation on \emph{heavy strings}~\cite{DBLP:journals/mst/KociumakaPR19}:

\begin{definition}\label{def:heavy}
For any weighted string $X$, we call a \emph{heavy string} $H_X$ of $X$ a string such that $H_X[i]$ is the letter having a largest probability in $X[i]$ (ties are broken arbitrarily). 
\end{definition}
\begin{example}
Let $X$ be the weighted string of \cref{ex:weighted_string}; a heavy string of $X$ is $H_X=\texttt{ABAAAB}$ (the tie at position $2$ is broken for $\texttt{B}$ and the tie  at position $5$ is broken for $\texttt{A}$).
\end{example}

\begin{lemma}[\cite{DBLP:journals/mst/KociumakaPR19}]\label{lem:heavy-string}
Let $H_X$ be a heavy string of $X$. For a weight threshold $\frac{1}{z}$ and any $z$-solid factor $U$ starting at position $i$ and ending at position $j$ of $X$, 
$d_H(U,H_X[i\dd j])\le \log_2 z$ holds. 
\end{lemma}

\begin{example}
Let $X$ be the weighted string of \cref{ex:weighted_string}; with $H_X=\texttt{ABAAAB}$. If $z=4$, $\log_2 z=2$, so no solid factor has more than $2$ mismatches with $H_X$ at its occurrence position. The string $\texttt{AABB}$ is not valid at position $1$ as it has $3$ mismatches with $H_X$; indeed, its probability is $\frac{1}{40}<\frac{1}{4}$. The string $\texttt{AAAA}$ is valid at position $1$ with probability $\frac{3}{10}>\frac{1}{4}$, and has only one mismatch with $H_X$. The condition is however not sufficient: $\texttt{ABAB}$ is not valid at position $1$ (its probability is $\frac{3}{40}<\frac{1}{4}$), even if it has only one mismatch with $H_X$.
\end{example}

We directly get the following result, which allows us to avoid storing the $z$-estimation $\mathcal{S}$ explicitly.

\begin{corollary}\label{cor:extra info}
Every solid factor of a weighted string $X$ for a weight threshold $\frac{1}{z}$ can be characterized by an interval of the heavy string $H_X$ plus the information of at most $\log_2 z$ single mismatches. The minimizer solid factor tree can be implemented as a compacted trie whose edges store only that information, which takes $\cO(\log z)$ extra space per edge. 
\end{corollary}

We apply Corollary~\ref{cor:extra info} to obtain Lemma~\ref{lem:minimizer tree size}. Based on this lemma, we construct the minimizer solid factor trees for $X$.

\begin{lemma}\label{lem:minimizer tree size}
The (forward and backward) minimizer solid factor trees can be constructed in $\cO(nz)$ time using $\cO(nz)$ space. Each tree has $\cO(\frac{nz}{\ell})$ expected nodes and its expected total size is $\cO(\frac{nz}{\ell}\log z)$. \footnote{We claim $\cO(nz)$ time and space during our construction because if $\log z > \ell$, we can abort the construction and resort to $\cO(nz)$ size.}
\end{lemma}

\begin{proof}
We apply \cref{theorem:z-estimation} to construct a $z$-estimation for $X$ in $\cO(nz)$ time and space.
The set of minimizers of any string can be computed in linear time~\cite{DBLP:conf/esa/LoukidesP21}. Thus, computing $\mathcal{M}_X$ for the $z$-estimation takes $\cO(nz)$ time.
The compacted trie of any collection of substrings of a string can be constructed in linear time in the length of the string~\cite{DBLP:journals/jea/Charalampopoulos20,DBLP:conf/cpm/KasaiLAAP01}, and thus the minimizer solid factor trees can be constructed in $\cO(nz)$ time using $\cO(nz)$ space.
The number of nodes and the total size of the trees follow from \cref{lem:minimizer density} and \cref{cor:extra info}. 
\end{proof}

\noindent{\bf Exploiting 2D Range Reporting.}~We explain how to employ a geometric data structure to pair up the leaf nodes corresponding to the same minimizer position from one of the minimizer solid factor tree that we have constructed above to the other.

Let us write $\Tsuff$ (resp.~$\Tpref$) for the forward (resp.~backward) minimizer solid factor tree. We fix an order on the leaves of both $\Tpref$ and $\Tsuff$, such that for any node in one of the trees, the set of its descendant leaves forms an interval. This is possible, for example,  by sorting the strings corresponding to the leaves in  lexicographical order. Via this ordering, we can consider a pair of leaves from $\Tsuff$ and $\Tpref$ as a point of a 2D data structure, which we call \emph{the grid}.

We start by some definitions: (1) Given a string $P$, we denote by $\Isuff(P)$ (resp.~$\Ipref(P)$) the (possibly empty) interval of leaves in the subtree obtained by reading $P$ in $\Tsuff$ (resp.~$\Tpref$). (2) We denote by $\mathcal{N}$ the set of all those points corresponding to pairs of leaves from $\Tsuff$ and $\Tpref$ with identical minimizer labels. Each point in $\mathcal{N}$ corresponds to a given minimizer $(i,j)\in\mathcal{M}_X$, and a pair of maximal solid factors in $X$ that can be read from $i$, both right-wise and left-wise. 
(3) Given a pair $P_1$, $P_2$ of strings, we denote by $\mathcal{N}(P_1,P_2)$ the intersection of the set $\mathcal{N}$ with the rectangle $\Isuff(P_1)\times\Ipref(P_2)$. (4) Given a string $P$ of length $m\ge \ell$, such that $f(P[1\dd \ell])=\mu$, we denote $\mathcal{N}(P[\mu\dd m],(P[1\dd \mu])^r)$ by $\mathcal{N}(P)$. 

\begin{lemma}\label{lem:checking rectangles}
For any pattern $P$ of length $m$, with $n\ge m\ge \ell$, if $P$ is a solid factor in $X$, then $\mathcal{N}(P)$ is nonempty. In particular, if $P$ has a valid occurrence in $X$ starting at position $k$ then $\mathcal{N}(P)$ contains at least one point having label $(k-1+f(P[1\dd \ell]),j)$ for some $j\in [1, \lfloor z \rfloor]$.
\end{lemma}

\begin{proof}
Let $P$ be such a pattern, which is a solid factor in $X$ at position $k$. By definition of a $z$-estimation, we know that $P$ occurs at position $k$ in some $S_j\in \mathcal{S}$. The minimizer computed for position $k$ of $S$ is $i=k-1+\mu$ with $\mu=f(P[1\dd \ell])$, since $S_j[k\dd k+\ell -1]=P[1\dd \ell]$. Therefore, the tree $\Tsuff$ (resp.~$\Tpref$) contains a leaf $k_1$ (resp.~$k_2$) corresponding to the longest substring of $S_j$ starting at position $i$ respecting the property $\Pi$, which starts with $P[\mu\dd  m]$ (resp.~the longest reversed substring of $S_j$ ending at position $i$ respecting the property, which starts with $P[1\dd \mu]^r$). Those leaf nodes both have a label $(i,j)=(k-1+\mu,j)$, hence the corresponding point is in $\mathcal{N}(P)$, which proves the result.
\end{proof}

Our index (i.e., $\Tsuff$, $\Tpref$, and the grid)
solves \INDEXING by answering 2D range reporting queries~\cite{DBLP:conf/compgeom/ChanLP11}. In the \emph{2D range reporting} problem, we are given a set $\mathcal{N}$ of $N$ points to be preprocessed, so that when one gives an axis-aligned rectangle as a query, we report the subset $\mathcal{K}$ of $\mathcal{N}$ such that point $p\in \mathcal{K}$ if and only if the rectangle encloses $p$. 
We consider a special case of this problem which suffices for our purposes. In particular, we use the following result.

\begin{lemma}[Section~2 of~\cite{DBLP:conf/latin/MakinenN06}]\label{lem:2demptiness}
Let $\mathcal{N}$ be a set of $N$ points coming from
pairing two permutations of $[1, N]$.
With $\cO(N)$ construction time, we can answer 2D range reporting queries in $\cO((1+k)\log N)$ time using $\cO(N)$ space, where $k$ is the number of points from $\mathcal{N}$ enclosed by the query rectangle.
\end{lemma}

Note that, even if each occurrence of a pattern can be detected with 2D range reporting queries, the converse is not true: if a pattern $U$ has a minimizer at position $\mu$ and both $U[1\dd \mu]$ and $U[\mu\dd |U|]$ are solid factors occurring at respective positions $k$ and $k+\mu-1$, then a corresponding point will be detected with the 2D range reporting queries, \emph{even} if $U$ is not a solid factor itself. In that case, $U$ is by definition a substring of some $S_j\in \mathcal{S}$, but does \emph{not}  respect the property. We can simply compute all the points by 2D range reporting, and check naively for false positives by comparing the pattern with $X$ at the positions corresponding to these points. Conversely, one can have several points corresponding to a single occurrence, if the pattern appears in multiple strings in $\mathcal{S}$ at a same position, which could also increase the running time. To control the number of such additional checks (both for false positives and duplicate ones), we give a bound on the expected number of occurrences of a given pattern in the $z$-estimation $\mathcal{S}$:
\begin{figure}[t]
    \centering
    \includegraphics[width=\textwidth]{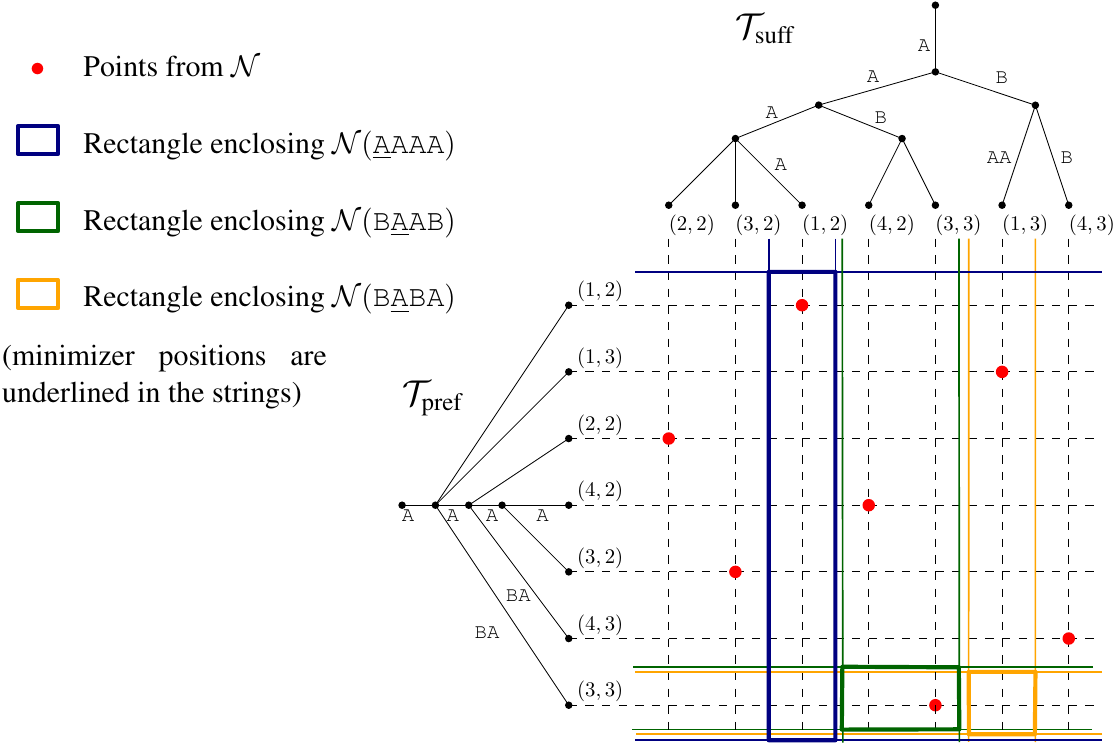}
    \caption{Our minimizer-based index for the weighted string from \cref{ex:weighted_string}, $\frac{1}{z}=\frac{1}{4}$, and the minimizers from Example~\ref{ex:2d-structure}. $\Tsuff$ is the forward minimizer solid factor tree and $\Tpref$ is the backward one.     Edges without labels are constructed for readability and mean that the parent and the children nodes correspond to the same string. Each leaf node representing the minimizer position $i$ in string $S_j$ is decorated with $(i,j)$.
    }\label{fig:2d-structure}   
\end{figure}

\begin{lemma}\label{lem: probabilistic bound}
For any string $P$ chosen uniformly at random from $\Sigma^m$, there are $\cO(nz/\sigma^m)$ points expected in $\mathcal{N}(P)$.
\end{lemma}

\begin{proof}
$\Tsuff$ and $\Tpref$ are constructed from a $z$-estimation $\mathcal{S}$, therefore each point in $\mathcal{N}(P)$ corresponds to an occurrence of $P$ in $\mathcal{S}$ (it might not respect property $\Pi$ however).
Since $\mathcal{S}$ has $(n-m+1)\lfloor z \rfloor\le nz$ substrings of length $m$, we have $\sum_{P\in\Sigma^m}|\mathcal{N}(P)|\le nz$, and hence if $P$ is chosen uniformly at random we obtain no more than $\frac{nz}{\sigma^m}$ points in expectation.
\end{proof}

\noindent{\bf Main Result.}~We arrive at the main result of the section:

\begin{theorem}\label{the:2d-structure}
Let $X$ be a weighted string of length $n$ over an alphabet of size $\sigma$, $\frac{1}{z}$ be a weight threshold, 
and $\ell>0$ be an integer. After $\cO(nz)$ construction time and using $\cO(nz)$ construction space, we can construct an index of $\cO(n+\frac{nz}{\ell}\log z)$ expected size answering \INDEXING queries of length $m\geq \ell$ in $\cO(m+(1+\frac{nz}{\sigma^m})\log \frac{nz}{\ell})$ expected time.
\end{theorem}

\begin{proof}
We first construct the minimizer solid factor trees of $X$ in $\cO(nz)$ time and space; the trees have size $\cO(\frac{nz}{\ell}\log z)$ (\cref{lem:minimizer tree size}). We preprocess the pairs of leaves for 2D range reporting queries in $\cO(\frac{nz}{\ell})$ time (\cref{lem:2demptiness}). When a pattern $P$ of length $m\ge \ell$ is given, we compute its leftmost minimizer $\mu$ in $\cO(\ell)$ time~\cite{DBLP:conf/esa/LoukidesP21}, compute the sides of the rectangle in $\cO(m)$ time by spelling $P[\mu\dd m]$ in  $\Tsuff$ and $(P[1\dd \mu])^r$ in $\Tpref$, and answer a 2D range reporting query in $\cO(\log \frac{nz}{\ell}(1+|\mathcal{N}(P)|))$ time (\cref{lem:2demptiness}). 
Finally, we must check, for every output point $(i,j)$, for a valid occurrence around the $i$th minimizer of the $j$th string of the $z$-estimation. To do this efficiently (i.e., in $\cO(\log z)$ time per point) \emph{without storing the $z$-estimation of $X$},
we store only the $\log_2 z$ closest mismatching positions to the left and to the right of every minimizer in $\mathcal{M}_X$ (\cref{lem:heavy-string}).
The total verification time is thus $\cO((\log z + \log(nz/\ell)) (|\mathcal{N}(P)|)+1)=\cO(\log \frac{nz}{\ell} (|\mathcal{N}(P)|+1))$. By \cref{lem: probabilistic bound}, we know that in expectation we have $|\mathcal{N}(P)|= \frac{nz}{|\Sigma|^m}$. We obtain an expected query time of $\cO(m+(1+\frac{nz}{|\Sigma|^m})\log \frac{nz}{\ell})$. The total size is $\cO(n+\frac{nz}{\ell}\log z)$, to store $H_X$ plus the index. 
\end{proof}

\begin{example}\label{ex:2d-structure}
Let $X$ be the weighted string from 
\cref{ex:weighted_string} and $\frac{1}{z}=\frac{1}{4}$.
The construction of our index is detailed in \cref{fig:2d-structure}, and query rectangles $\mathcal{N}(P)$ (resp.~$\mathcal{N}(P')$ and $\mathcal{N}(P'')$) are constructed for patterns $P=\underline{\texttt{A}}\texttt{AAA}$ (resp.~$P'=\texttt{B\underline{A}AB}$ and $P''=\texttt{B\underline{A}BA}$) whose minimizer positions are underlined. The blue rectangle $\mathcal{N}(P)$ contains exactly one point which corresponds to a substring \texttt{AAAA} in $S_2$ that respects the property. This substring corresponds to an occurrence of $P$ at position $1$ with probability $1\cdot \frac{1}{2}\cdot\frac{3}{4}\cdot\frac{4}{5}=0.3 >\frac{1}{4}$ in $X$. The green rectangle $\mathcal{N}(P')$ contains one point, which corresponds to an occurrence of $P'$ in $S_3$. However, this occurrence does not respect the property $\Pi$ (because $i+|P'|-1=2+4-1=5 > \pi_3[2]=4$) and therefore is a false positive in $X$ (it occurs with probability $0.15<\frac{1}{4}$). Finally, the orange rectangle $\mathcal{N}(P'')$ does not contain any point, because the pattern does not occur in the $z$-estimation. 
 \end{example}

\section{Space-efficient Construction of the Index}\label{sec:space-efficient-construction}

Recall that to construct the index in Section~\ref{sec:index}, we first construct a $z$-estimation, which temporarily takes $\Theta(nz)$ space during construction. 
In this section, we improve the space required for the construction of the index by designing a space-efficient algorithm for constructing a minimizer solid factor tree with only a moderate increase in the construction time.

\vspace{+1mm}
\noindent{\bf Main Idea.}~We start the index construction by \emph{simulating} the construction of an \emph{extended solid factor tree}. 
This is a trie of the solid factors of $X$ extended by $H_X$, the heavy string of $X$.
In particular, we maintain the subtree induced by the solid factors that \emph{start at minimizer positions} but discard the nodes that we do not need upon returning to their parents. We achieve this via traversing the full tree in a DFS order. Thus, even though the full tree size is $\cO(nz)$ (\cref{lem:classical tree size}), at any moment, we store \emph{only} the current leaf-to-root path plus the actual output. Therefore,  we use only $\cO(n + \frac{nz}{\ell}\log z)$ expected space at a cost of $\cO(nz\log \ell)$ time (\cref{lem:main}). We next reverse this tree (the solid factors are read from leaf to root there, while in the minimizer solid factor tree those are read from root to leaf), in $\cO(\frac{nz}{\ell}\log \frac{nz}{\ell}\log z)$ expected time, using $\cO(\frac{nz}{\ell}\log z)$ expected space (\cref{the:space_efficient}). \emph{Our approach thus trades construction time for less space}: in total, it adds up to $\cO(nz\log\ell +\frac{nz}{\ell}\log \frac{nz}{\ell}\log z)$ expected construction time (instead of $\cO(nz)$) but $\cO(n+\frac{nz}{\ell}\log z)$ expected construction space (instead of $\cO(nz)$). We thus arrive to Contribution 2.

\vspace{+1mm}
\noindent{\bf Key Concepts.}~The string $U\cdot H_X[j+1\dd n]$ (resp.~$(H_X[1\dd i-1]\cdot U)^r$) is called the \emph{right extension of the solid factor $U$} (resp.~\emph{left extension of the solid factor $U$}), if $U$ is a solid factor of $X$ starting at position $i$ and ending at position $j\ge i-1$. 

For such a $U$, we define a \emph{forward} \emph{extended solid factor tree}  of $X$ as a trie of all the reversals of $U\cdot H_X[j+1\dd n]$, and a \emph{backward} \emph{extended solid factor tree} of $X$ as a trie 
of all $H_X[1\dd i-1]\cdot U$. These trees can be constructed by looking only at the extensions of \emph{maximal} solid factors, namely those that cannot be extended into a longer solid factor~\cite{DBLP:conf/cpm/BartonKPR16}. We make use of the following lemma to bound the size of the trees:

\begin{lemma}[\cite{DBLP:conf/cpm/BartonKPR16}]\label{lem:classical tree size}
The extended solid factor trees have $\cO(nz)$ nodes.
\end{lemma}

We next describe our algorithm (see \cref{fig:pseudocode} for pseudocode).

\begin{figure}[ht]
\begin{minipage}{0.46\textwidth}
\vspace{-0.9cm}
\begin{algorithm}[H]\small 
\caption{Construct-$\mathcal{T}$
}\label{alg:minimizer tree}
\begin{algorithmic}[1]
\State{{\bf Global variables:} $j=n$, $p=1.0$, string $S=\varepsilon$, set $\diff=\emptyset$, set $\text{Minimizers}=\emptyset$.}
\State{create a node $root$}
\State{run Augment-$\mathcal{T}(n+1,root)$}
\State{{\bf return} $root$}\Comment{The minimizer extended solid factor tree} 
\end{algorithmic}
\end{algorithm}
\end{minipage}\hfill{}
\begin{minipage}{0.48\textwidth}
\begin{algorithm}[H]\small 
\caption{DOWN$(i,u,\alpha)$}\label{alg:down}
\begin{algorithmic}[1]
\State {$S\gets \alpha S$}
\If{$\alpha\neq H_X[i]$}{~add $(i,\alpha)$ to $\diff$}
\EndIf
\If{$|S|\ge\ell$ and $p\cdot PP_H[i-1+\ell]/PP_H[j]\ge \frac{1}{z}$}
    \Comment{$S[1\dd \ell]$ is solid}
    \State $\text{Minimizers}\gets\text{Minimizers}\cup\{i+f(S[1\dd\ell])-1\}$
\EndIf
\State{add a node $v$ as a child of $u$}
\State{run Augment-$\mathcal{T}(i,v)$}
\end{algorithmic}
\end{algorithm}
\end{minipage}\hfill{}
\begin{minipage}[b]{0.46\textwidth}
\vspace{-0.9cm}
\begin{algorithm}[H]\small 
\caption{Augment-$\mathcal{T}(i,u)$}\label{alg:augument}
\begin{algorithmic}[1]
\For{$\alpha\in\Sigma$}
\If{$p=1$ and $\alpha=H_X[i-1]$}\Comment{If $U$ is empty} 
    \State{$j\gets j-1$}
    \State{run DOWN$(i-1,u,\alpha)$}
    \State{$j\gets j+1$}
\ElsIf{$p\cdot p_{i-1}[\alpha]\ge\frac{1}{z}$}
    \State{$p\gets p\cdot p_{i-1}[\alpha]$}
    \State{run DOWN$(i-1,u,\alpha)$}
\EndIf
\EndFor
\If{$i<n+1$}{~run UP$(i,u)$}\EndIf
\end{algorithmic}
\end{algorithm}
\end{minipage}\hfill{}
\begin{minipage}[b]{0.48\textwidth}
\begin{algorithm}[H]\small 
\caption{UP$(i,u)$:}\label{alg:up}
\begin{algorithmic}[1]
\If{$i\in\text{Minimizers}$}
\State{remove $i$ from $\text{Minimizers}$}
\State{set label of $u$ to $(i,\diff)$}
\ElsIf{$u$ has at most one child}{~merge $u$ with $\textsc{parent}(u)$}
\EndIf
\State{$p\gets \min(1,p\cdot p_{i}[S[1]]^{-1}$)}
\If{$(i,S[1])\in \diff$}{~remove $(i,S[1])$ from $\diff$}
\EndIf    
\State{remove the first letter from $S$} 
\end{algorithmic}
\end{algorithm}
\end{minipage}
\caption{The space-efficient algorithm for constructing the minimizer extended solid factor tree of a weighted string $X$.}\label{fig:pseudocode}
\end{figure}

\noindent{\bf Construction.}~We start by constructing the minimizer versions of the extended solid factor trees -- that is for the solid factors trimmed to their parts starting (resp.~ending) at the position of their minimizer. In particular, we show how to construct the forward extended solid factor tree (see Fig.~\ref{fig:MWST-SE-example}) -- the backward one can be constructed by simply doing the same operations on the reversed string, except that the minimizers will be computed on the reversed substrings.

\noindent\textbf{Initialization.}~We construct the tree with a DFS traversal of the full (non-minimizer) extended solid factor tree, starting from the root, which corresponds to $\varepsilon$, the empty string. Each node corresponds to the right extension of a solid factor $U$ of $X$ starting at a position $i$ (recall that $U$ can be empty, in which case its right extension is $H_X[i\dd n]$).
We assume that $H_X$ and $PP_H$ are computed before running the main algorithm and can be accessed just like $X$ and $n=|X|$.

\noindent\textbf{First Visit to a Node.}~When a node $u$ that  corresponds to a solid factor $U$ starting at position $i$ of $X$ and ending at position $j$ is created, we keep a pair of labels ($i$, $\diff$), where $\diff$  is the sequence (list) of mismatches between $U$ and $H_X[i\dd j]$.
By Lemma~\ref{lem:heavy-string}, the label of a given node has size $\cO(\log z)$. Note, also, that for any ancestor of a node $u$ its list of mismatches will be a suffix of $\diff$.

A single node and equivalently such a pair of labels can still represent multiple solid factors (for different values of $j$ -- if the suffix of the solid factor matches the heavy string); henceforth, by $U$ we mean the shortest such solid factor: $j$ is the largest element of $\diff$ or $j=i-1$ if $\diff$ is empty.

Additionally, given a node $u$ representing $S=U\cdot H_X[j+1\dd n]$, we check if the longest solid prefix of $S$ has length at least $\ell$; we check this in $\cO(1)$ time using value $p$ -- a global variable that denotes the probability of the current node -- that is the weight of $U$ and the precomputed array $PP_H$ of prefix products of $H_X$ for the heavy part. If this is the case, we ask for the minimizer $\mu$ of this solid factor, and mark the $(\mu-1)$th ancestor of $u$ as a minimizer node. Such minimizer can be found in $\cO(1)$ time  using a heap data structure~\cite{DBLP:books/daglib/0023376}, which stores information about the length-$k$ substrings of the length-$\ell$ prefix of $S=U\cdot H_X[j+1\dd n]$ and is updated in $\cO(\log \ell)$ time in each step of the traversal.

\noindent\textbf{Stepping Down to a Child Node.}~If $U$ is empty, then the node $u$ corresponds to a string $H_X[i\dd n]$. In this case, $p$ is not updated when creating its child $v$ corresponding to $H_X[i-1\dd n]$. This way, we ensure by induction that $p=1$ at the creation of each such node, and only for such a node, so that this can be checked in $\cO(1)$ time. %We now assume that $U$ is nonempty ($\diff$ is nonempty). 
We now assume that the created child does not correspond to $H_X[i-1\dd n]$. 
To create a child $v$ of the node $u$, corresponding to the right extension of the string $\alpha\cdot U$ for some letter $\alpha\in \Sigma$, one needs to check if $\alpha\cdot U$ is valid by computing its probability ($U$ is nonempty, or the letter $\alpha$ is different from $H_X[i-1]$). This is done using $p$, which we multiply by $p_{i-1}[\alpha]$. In any case, the labels of $v$ are computed from the labels of $u$ by decreasing the starting position and potentially adding a new label to $\diff$.

\noindent\textbf{Returning to the Parent Node.}~In the DFS we traverse the full extended solid factor tree, but we are only interested in the strings that start in those minimizer nodes. Thus, we remove the nodes which correspond to letters of the solid factors that appear before the position of the first minimizer and recursively compactify the tree during the traversal.

After all descendants of $u$ are created, we keep $u$ explicit if it is a minimizer node or if it has more than one (not removed) child. Otherwise, the node $u$ is made implicit by merging it with its parent. Finally, upon returning to the parent of $u$ we update $p$ by dividing it by $p_{i}[U[1]]$ (if $p<1$) and the list of differences by removing position $i$ if $i\in\diff$.

\begin{figure}[t]
     \begin{subfigure}[b]{0.22\textwidth}
     \centering
          \includegraphics[scale=0.35,trim={5mm 0 0 0.1cm},clip]{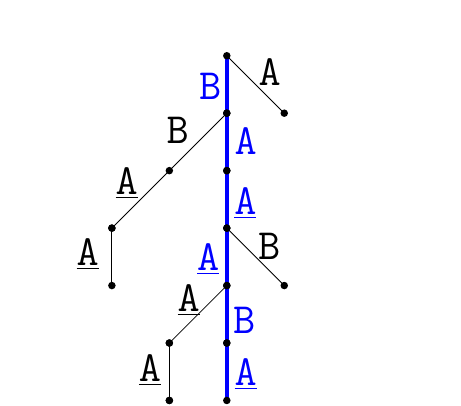}
             \caption{}
             \end{subfigure}\hspace{+2mm}
     \begin{subfigure}[b]{0.22\textwidth}
     \centering
          \includegraphics[scale=0.35,trim={5mm 0 0 0.1cm},clip]{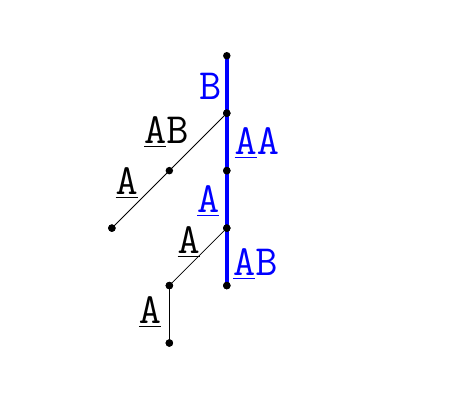}
          \caption{}
         \end{subfigure}\hspace{+2mm}
     \begin{subfigure}[b]{0.22\textwidth}
     \centering
    \includegraphics[scale=0.4,trim={2cm 0.6cm 1cm 0.1cm},clip]{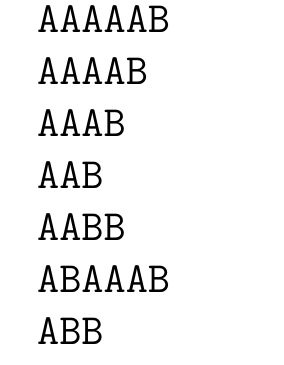}
        \caption{}
         \end{subfigure}
     \hspace{+5mm}
     \begin{subfigure}[b]{0.22\textwidth}
         \centering
          \includegraphics[scale=0.4,trim={0cm 1.7cm 0cm 0.1cm},clip]{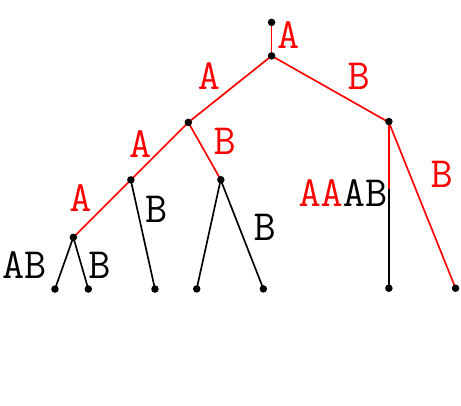}
          \caption{}
         \end{subfigure}
        \caption{(a) The forward extended solid factor tree with $X$ from \cref{ex:weighted_string}. 
        The blue edges correspond to the heavy string $H_X=\texttt{ABAAAB}$ (reversed). The minimizer positions are underlined. (b) The minimizer extended solid factor tree. The edges without any minimizer descendant nodes are pruned and the non-minimizer nodes are made implicit. (c) The lexicographically sorted strings corresponding to each path from a minimizer node to a root. (d) The  minimizer solid factor tree constructed by Theorem~\ref{the:space_efficient}. It contains the forward (top) tree from \cref{fig:2d-structure} as a red subtree. The edge with no  label is added in the figure to stress that $\texttt{AAB}$ also has a corresponding leaf. In the algorithm, we simply make the corresponding internal node explicit and treat it as a leaf node.}
        \label{fig:MWST-SE-example}
\end{figure}

\noindent\textbf{Complexity Analysis.}~By Lemma~\ref{lem:minimizer tree size} the final tree has $\cO(\frac{nz}{\ell})$ nodes in expectation. As for the construction space, observe that while a node $u$ is being processed, only the path between $u$ and the root is uncompacted, and contains at most $n$ nodes. All the other global variables also have size $\cO(n)$, therefore the total expected work space needed is $\Oh(\frac{nz}{\ell}\log z +n)$.
As for the construction time, note that the set of created nodes is exactly the set of nodes from the original extended solid factor tree (namely, without minimizers), which has size $\cO(nz)$ by Lemma~\ref{lem:classical tree size}. 
Reading $X$ and creating $H_X$ costs $n\sigma=\cO(n)$ time~\cite{DBLP:journals/iandc/Charalampopoulos19}; after that the total number of probabilities read is bounded by the number of tree nodes as for a single position we can read those in the order of non-increasing probabilities.

During the construction, all the operations cost $\cO(1)$ time with the exception of updating the minimizer heap which takes $\cO(\log \ell)$ time and storing a copy of the list of differences for each minimizer node in $\cO(\log z)$ time. 
Note that the last type of the operation does not influence the 
worst-case running time as we can abandon the computation upon learning that the total size of those lists reaches $nz$ -- in which case the classic (non-minimizer) data structure is more efficient. We have thus proved the following lemma.

\begin{lemma}\label{lem:main}
For any weighted string of length $n$, any weight threshold $\frac{1}{z}$, and any integer $\ell>0$, we can construct a representation of the minimizer extended solid factor trees in $\Oh(nz\log \ell)$ time using $\Oh(n+\frac{nz\log z}{\ell})$ expected space.
\end{lemma}

\noindent{\bf Main Result.}~The minimizer solid factors tree can be constructed from the minimizer extended solid factor tree in $\cO(\frac{nz}{\ell}\log \frac{nz}{\ell} \log z)$ expected time and $\cO(\frac{nz}{\ell}\log z)$ expected space (as proven below), hence we arrive at the main result of the section:

\begin{theorem}\label{the:space_efficient}
For any weighted string $X$ of length $n$, any weight threshold $\frac{1}{z}$, and any integer $\ell>0$, we can construct the minimizer solid factor tree in expected time $\cO(nz\log \ell+\frac{nz}{\ell}\log \frac{nz}{\ell} \log z)$ and space $\cO(n+\frac{nz}{\ell} \log z)$.
\end{theorem}

\begin{proof}
As stated above it remains to show how to construct the minimizer solid factor tree from the minimizer extended solid factor tree.
We achieve that by reversing the tree, that is, by creating the trie of all the strings from the minimizer extended solid factor tree read from leaf to root (corresponding to strings $U\cdot H_X[j+1\dd n]$).
Note that for two such strings we can find their longest common prefix (LCP), and hence also compare them in $\cO(\log z)$ time with a use of an LCP data structure for $H_X$~\cite{DBLP:journals/tcs/LandauV86} (comparison of $\cO(\log z)$ intervals of $H_X$ and $\cO(\log z)$ differences).

We first sort those strings in lexicographic order. Since there are in expectation $\cO(\frac{nz}{\ell})$ of them, and a single comparison takes $\cO(\log z)$ time, this takes $\cO(\frac{nz}{\ell}\log \frac{nz}{\ell} \log z)$ time in total using any optimal comparison-based sorting algorithm~\cite{DBLP:books/daglib/0023376}. Now we construct the compacted trie of those strings node by node in the order of a DFS.
Each edge will be labeled with an interval of $H_X$ and a list of at most $\log_2 z$ differences.

We start from creating a single edge from root to a leaf representing the first string.
Now we iterate over all remaining strings in lexicographic order -- we first compute the length of the LCP of this string, and the previous one, next starting from the leaf representing the previous string we move up the tree node by node to find its ancestor at depth equal to the length of this LCP. If this node turns out to be an implicit one, then we make it explicit by dividing the edge (and hence also the interval of $H_X$ and the list of differences). We finish by creating a new child of the reached node -- this child becomes the leaf representing the new string.

Unlike the construction from~\cite{DBLP:conf/cpm/BartonKPR16} we do not need to trim the $H_X$ parts after constructing the tree, as in our query algorithm we must verify the weight for each match anyway.  

\end{proof}

\section{Practically Fast Querying Without a Grid}\label{sec:fast_pm}

In this section, we describe a simple and fast querying algorithm
that does not require the grid to be constructed on top of the trees. This querying algorithm has worse guarantee than that in Theorem~\ref{the:2d-structure} but performs much better in practice, due to its simplicity, as we show %later 
in the experimental evaluation.

Like in the previous construction let $\mu=f(P[1\dd \ell])$.
Without loss of generality we assume that $\mu\le \frac{m}{2}$ (otherwise we swap the roles of the parts of $P$ and of trees $\Tsuff$ and $\Tpref$). Let $u$ be the node reached by reading $P[\mu\dd m]$ in $\Tsuff$. We can separately check each leaf in the subtree of $u$ as a potential candidate in $\cO(m)$ time: 
we can do this assuming we have random access to $X$.
This time we cannot use the $\frac{nz}{\sigma^m}$ bound on the expected number of candidates from Lemma~\ref{lem: probabilistic bound}. However, $P[\mu\dd m]$ has length at least $m/2$, and hence the expected number of candidates can still be bounded by $\sum_{k=\lceil m/2\rceil}^{m}\frac{nz}{\sigma^k}\le 2\frac{nz}{\sigma^{m/2}}$ using a similar argument. 
Thus we can answer a query in $\cO(m\cdot (1+\frac{nz}{\sigma^{m/2}}))$ expected time. 

\begin{example}\label{ex:fast_querying}
For the weighted string from \cref{ex:weighted_string} and pattern $P'=\texttt{B\underline{A}AB}$ the grid based construction finds a single candidate $(3,3)$ to check (see \cref{fig:2d-structure} and \cref{ex:2d-structure}). In case of the simpler querying algorithm we only locate one  part of the pattern in one of the trees, in this case $|\texttt{B}\underline{\texttt{A}}|<|\underline{\texttt{A}}\texttt{AB}|$, hence we locate $\texttt{AAB}$ in $\Tsuff$. We find two leaves with labels $(3,3)$ and $(4,2)$,  that is candidate occurrences starting at positions $2=3-|\texttt{B}\underline{\texttt{A}}|+1$ and  $3=4-|\texttt{B}\underline{\texttt{A}}|+1$.  Each of those can be checked naively to get the occurrence probabilities  $3/20$ and $3/40$ respectively;  each probability is smaller than $1/4$.
\end{example}

\section{Related Work}\label{sec:related}

The \emph{Weighted Indexing} problem was introduced by the work of Iliopoulos et al.~\cite{DBLP:journals/fuin/IliopoulosMPPTT06}; the authors gave a tree-based index supporting $\cO(m+\OUTPUT)$-time queries. The construction time, space, and size of the index was, however, $\cO(n\sigma^{z\log z})$. Their index is essentially a compacted trie of all the solid factors of $X$. Amir et al.~\cite{DBLP:conf/cpm/AmirCIKZ06} then reduced the Weighted Indexing problem to the so-called \emph{Property Indexing} problem in a standard string of length $\cO(nz^2\log z)$. For the latter problem, Amir et al. proposed a super-linear-time construction and $\cO(m+\OUTPUT)$-time queries. Later, it was shown that Property Indexing can be performed in linear time~\cite{DBLP:journals/iandc/BartonK0PR20}. This directly gives a solution to the Weighted Indexing problem with construction time, space, and size $\cO(nz^2\log z)$, preserving the optimal query time. 
The state-of-the-art indexes for Weighted Indexing are the weighted suffix tree (\WST)~\cite{DBLP:journals/iandc/BartonK0PR20,DBLP:conf/cpm/BartonKPR16} and the weighted suffix array (\WSA)~\cite{DBLP:journals/jea/Charalampopoulos20}; see the Introduction for more details. To conclude, \emph{there are currently no practical indexing schemes for Weighted Indexing mainly due to their prohibitive size and space requirements}.

There is substantial work on practical indexing schemes for probabilistic/uncertain data; e.g., for range queries~\cite{tao1,DBLP:journals/vldb/ChenGZJCZ17,tao2,rangereport,range3}, top-$k$ queries~\cite{topk1,topk2,topk3,topk4}, nearest neighbor queries~\cite{talg,tkde_nn,edbt_nn}, sql-like queries~\cite{DBLP:conf/sigmod/QiJSP10}, inference queries~\cite{DBLP:conf/sigmod/KanagalD09}, and probabilistic equality threshold queries~\cite{petq}. These indexes were developed for different uncertainty data models, such as tuple uncertainty and attribute uncertainty~\cite{DBLP:conf/icde/SinghMSPHNC08}. Under tuple uncertainty, the presence of a tuple in a relation is probabilistic, while under attribute uncertainty a tuple is certainly present in a database but one or more of its attributes are not known with certainty. Several indexes are built on R-trees or inverted indexes~(e.g.,~\cite{petq,DBLP:conf/ssd/DaiYMTV05}), while others are built on R$^*$-trees~\cite{tao1}. There are also specialized indexes, e.g., for  
probabilistic XML queries~\cite{xml_index} or uncertain graphs~\cite{graph1,graph2}. 
Our work differs substantially from the above in the supported data model and query type. 

\section{Experimental Evaluation}\label{sec:exp}

\subsection{Data and Setup}

\noindent {\bf Data.}~We used three real weighted strings which model variations found in the DNA ($\sigma=4$) of different samples of the same species.  The chromosomal or genomic location of a gene or any other genetic element is called a \emph{locus} and alternative DNA sequences at a locus are called \emph{alleles}. Allele frequency, or gene frequency, is the relative frequency of an allele at a particular locus in a population, expressed as a fraction or percentage.
Thus, alleles have a natural representation as  weighted strings: we model the probability $p_i(\alpha)$ in these strings as the relative frequency of letter $\alpha$ at position $i$ among the different samples. 

We next describe the datasets we used (see also Table~\ref{table:data}):
\begin{table}[t]
\caption{Characteristics of the real datasets we used.}
\label{table:data}
\centering
\resizebox{0.68\textwidth}{!}{%
\begin{tabular}{|c|| c |c | c| c|}
\hline \multirow{2}{*}{\bf{Dataset}}   & {\bf \# of}~   & {\bf Length} & {\bf $\Delta$}~ &  {\bf Size of $z$-estimation} \\ 
& {\bf samples} & $n$ & {\bf as percentage of $n$} & {\bf for the default $z$ (MBs)} \\
\hline \hline
\bf \SARS & $1,181$ & $29,903$ & $3.6\%$ & $31$ \\ \hline
\bf \EFM & $1,432$ & $2,955,294$ & $6\%$ & $378$ \\ \hline
\bf \Z & $2,504$ & $35,194,566$ & $3.2\%$ & $282$ \\ \hline  %n*z
\bf \RS & N/A & $6,053,462$ & $100\%$ & $96.9$\\\hline 
\end{tabular}}
\vspace{+2mm}
\end{table}

\begin{itemize}
\item \SARS: The full genome of \emph{SARS-CoV-2} (isolate Wuhan-Hu-1)~\cite{footnote5} 
combined with a set of single nucleotide polymorphisms (SNPs)~\cite{footnote6} 
taken from $1,181$ samples~\cite{SARS}.
\item \EFM: The full chromosome of \emph{Enterococcus faecium} Aus0004 strain (CP003351)~\cite{footnote3} 
combined with a set of SNPs~\cite{footnote4} 
taken from $1,432$ samples~\cite{EFM}.
\item \Z: The full chromosome 22 of the \emph{Homo sapiens} genome (v.~GRCh37)~\cite{footnote7} 
combined with a set of SNPs~\cite{footnote8} 
taken from the final phase of the 1000 Genomes Project (phase 3) 
representing $2,504$ samples~\cite{1000genomes}.  
\end{itemize}

The percentage of positions where more than one letter has a probability of occurrence larger than $0$ is denoted by $\Delta$. 

We also used another real weighted string comprised of Received Signal Strength Indicator (RSSI) (i.e., signal strength) values of sensors; $p_i(\alpha)$, $i\in[1,n]$, corresponds to the ratio of IEEE 802.15.4 channels that received an RSSI value $\alpha$ at time $i$, and $\sigma=91$~\cite{rssi_data}. Its characteristics are in Table~\ref{table:data}. Furthermore, we used a family of weighted strings, generated from \RS. Each such string is denoted by $\RS_{n,\sigma}$, where $n\in\{2,4,6,8\}$ is the number of times the length of this string is larger than that of $\RS$ and $\sigma\in \{16,32,64\}$ is its alphabet size. To increase $n$, we appended $\RS$ to itself. To reduce $\sigma$, we replaced each RSSI value (integer) $v$ with $v\mod y$, where $y\in \{16,32,64\}$ is the desired alphabet size. For all these datasets, $\Delta=100\%$.

\begin{figure*}
    \centering    
\subfloat[][\SARS]
{
    \includegraphics[width=.205\textwidth]{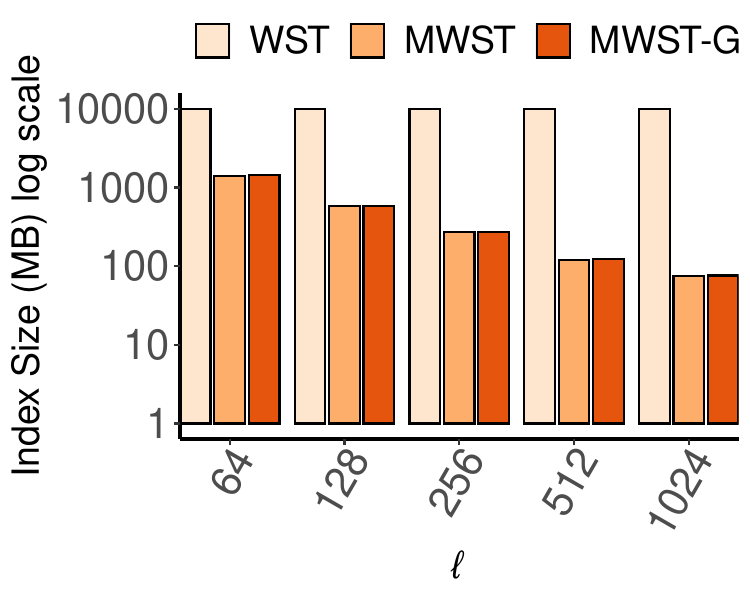}\label{fig:index_size_vs_ell_tree2}
}
\subfloat[][\SARS] % z = 1024
{
    \includegraphics[width=.205\textwidth]{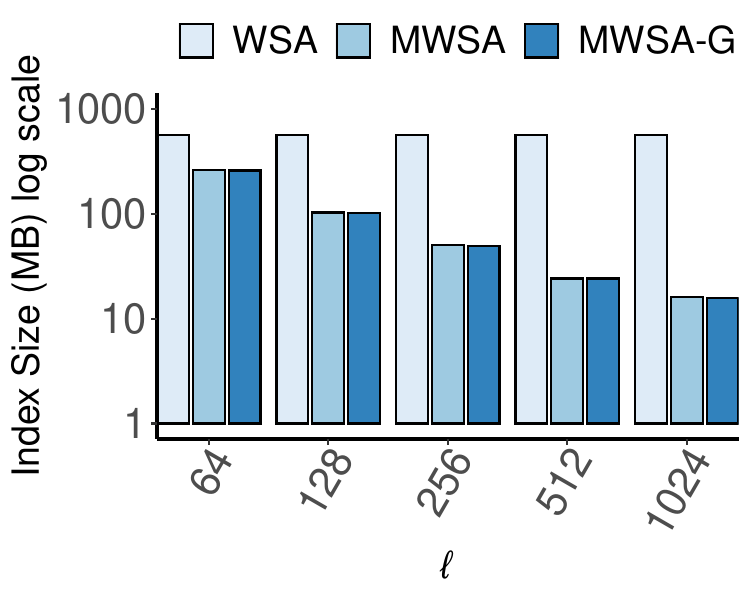}
    \label{fig:index_size_vs_ell_array2}
}
\subfloat[][\EFM]
{
    \includegraphics[width=.205\textwidth]{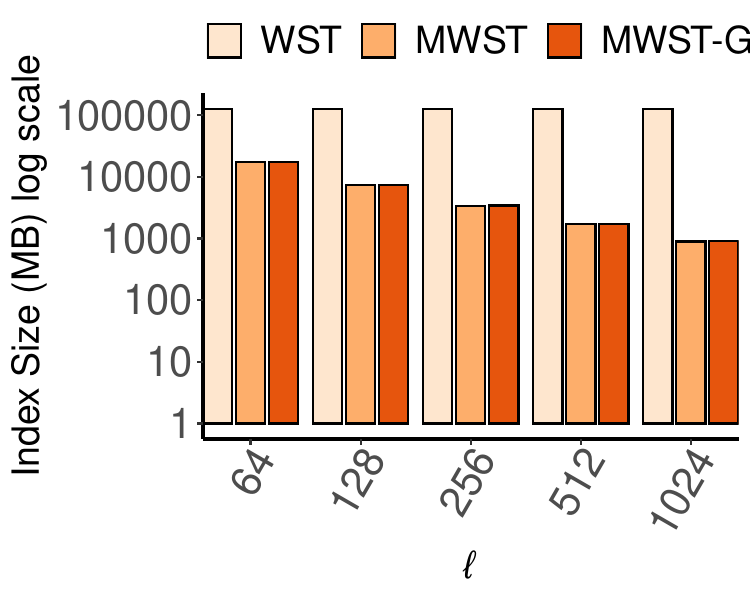}\label{fig:index_size_vs_ell_tree1}
}
\subfloat[][\EFM] %z = 128
{
    \includegraphics[width=.205\textwidth]{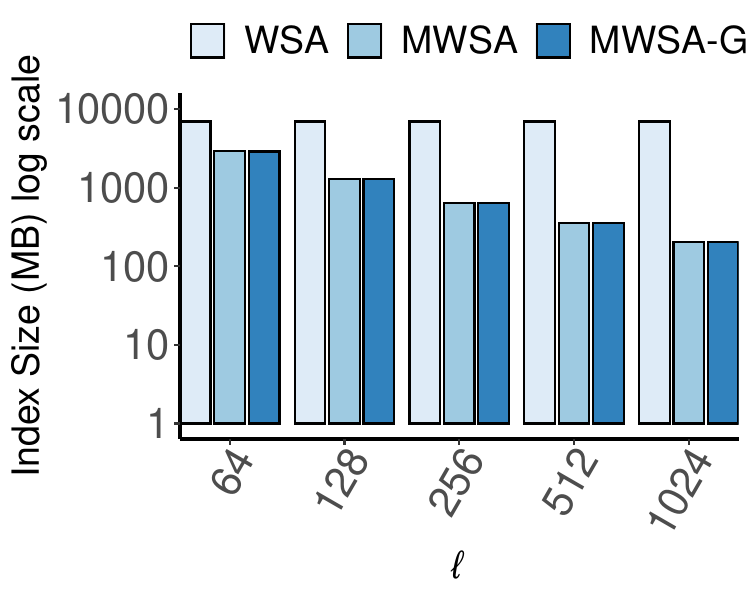}
    \label{fig:index_size_vs_ell_array1}
}
\\
\subfloat[][\Z]
{
    \includegraphics[width=.205\textwidth]{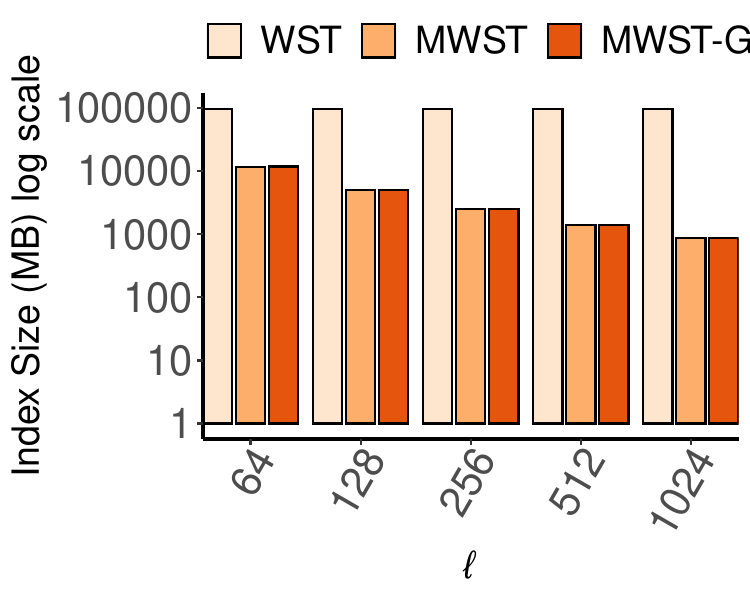}\label{fig:index_size_vs_ell_tree3}
}
\subfloat[][\Z] % z=8
{
    \includegraphics[width=.205\textwidth]{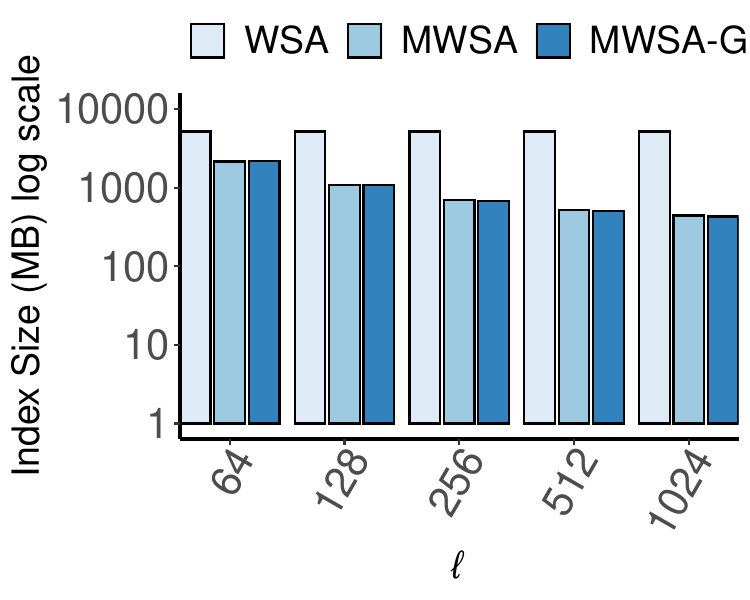}
    \label{fig:index_size_vs_ell_array3}
}
\caption{Index size (log scale, MB) vs. $\ell$.}\label{fig:index_size1}
\vspace{+2mm}
\subfloat[][\SARS]
{
    \includegraphics[width=.205\textwidth]{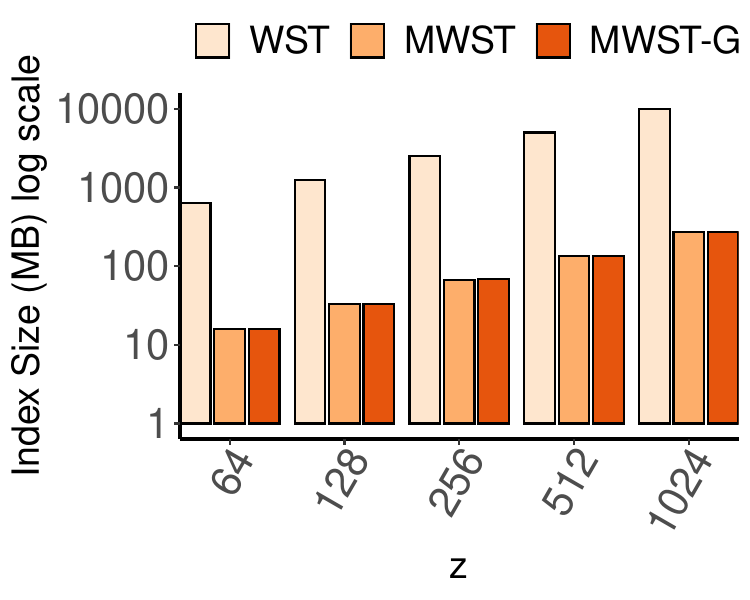}
    \label{fig:index_size_vs_z_tree2}
}
  \subfloat[][\SARS]
{
    \includegraphics[width=.205\textwidth]{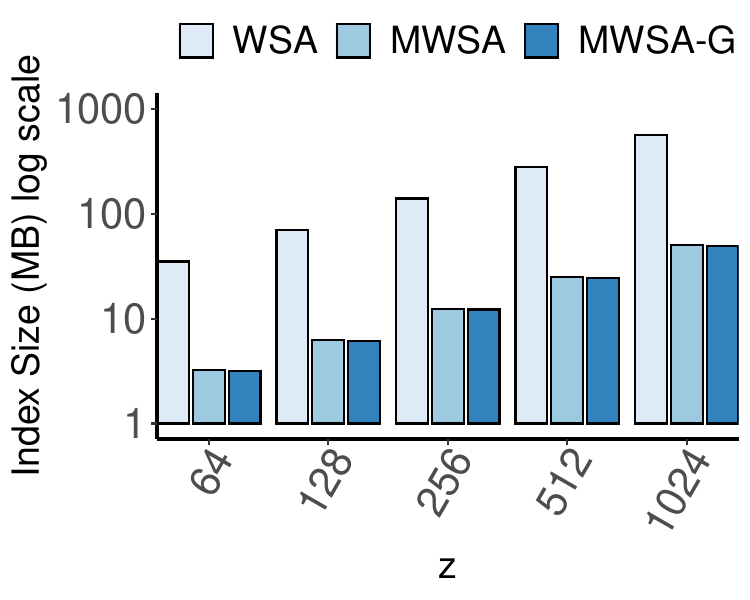}
    \label{fig:index_size_vs_z_array2}
}
\subfloat[][\EFM] % z=8
{
    \includegraphics[width=.205\textwidth]{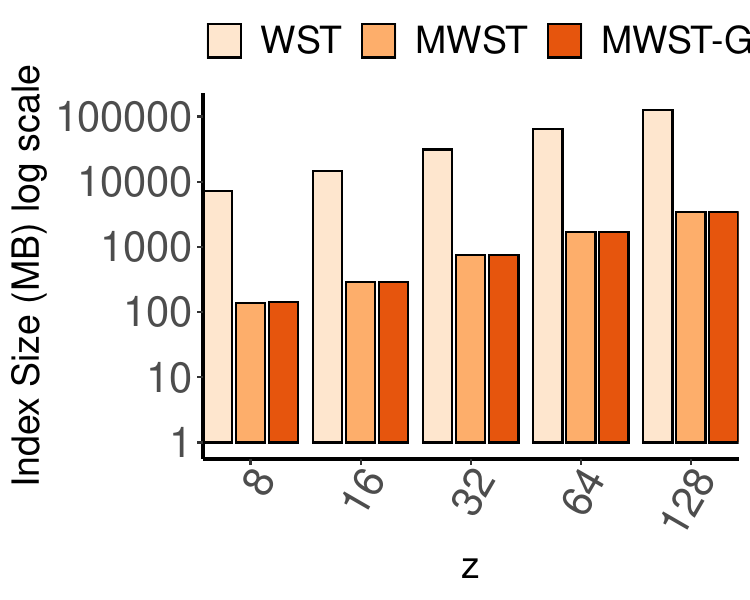}
    \label{fig:index_size_vs_z_tree1}
}
 \subfloat[][\EFM]
{
\includegraphics[width=.205\textwidth]
{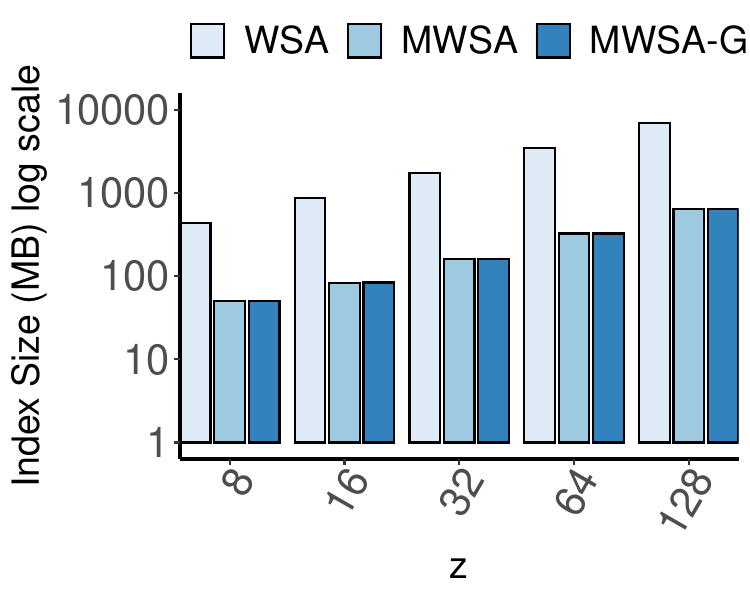}
\label{fig:index_size_vs_z_array1}
}\\
\subfloat[][\Z]
{
   \includegraphics[width=.205\textwidth]{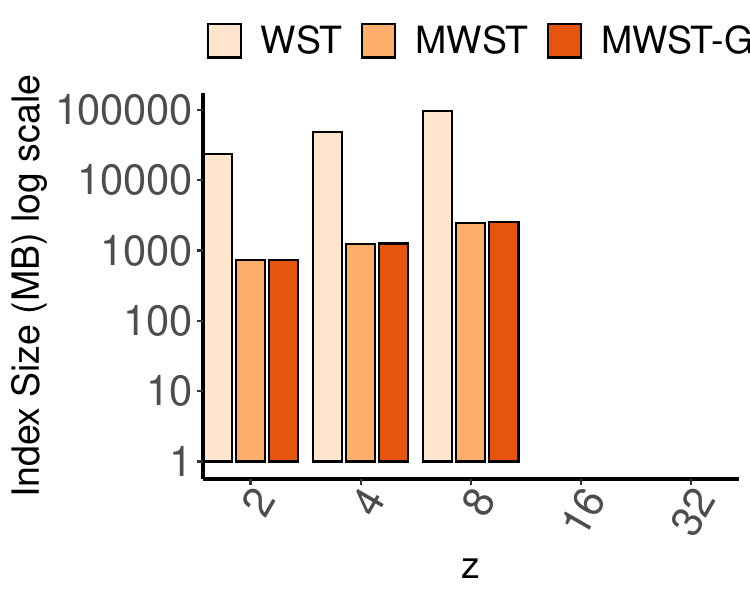}
   \label{fig:index_size_vs_z_tree3}
}
 \subfloat[][\Z]
{
    \includegraphics[width=.205\textwidth]{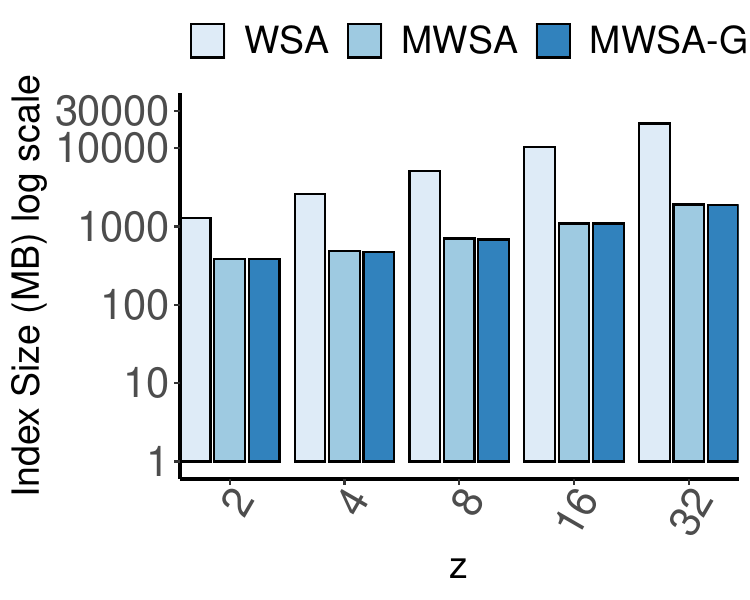}
    \label{fig:index_size_vs_z_array3}
}
    \caption{Index size (log scale, MB) vs. $z$. The tree-based indexes for \Z (Fig.~\ref{fig:index_size_vs_z_tree3}) needed $>$ 252GB of space when $z\geq 16$ and hence could not be constructed.}
    \label{fig:index_size2}
\end{figure*}

\begin{figure*}
    \centering
%    \subfloat[][\SARS]
% {
%     \includegraphics[width=.205\textwidth]{figures/sars_vs_ell_construction_space_tree.pdf}
%     \label{fig:construction_space_vs_ell_tree2}
% }
% \subfloat[][\SARS]
% {
%     \includegraphics[width=.205\textwidth]{figures/sars_vs_ell_construction_space.pdf}
%     \label{fig:construction_space_vs_ell_array2}
% }
% \\
 \subfloat[][\EFM]
{
    \includegraphics[width=.205\textwidth]
    %{figures/efm_vs_ell_construction_space_tree.pdf}
    {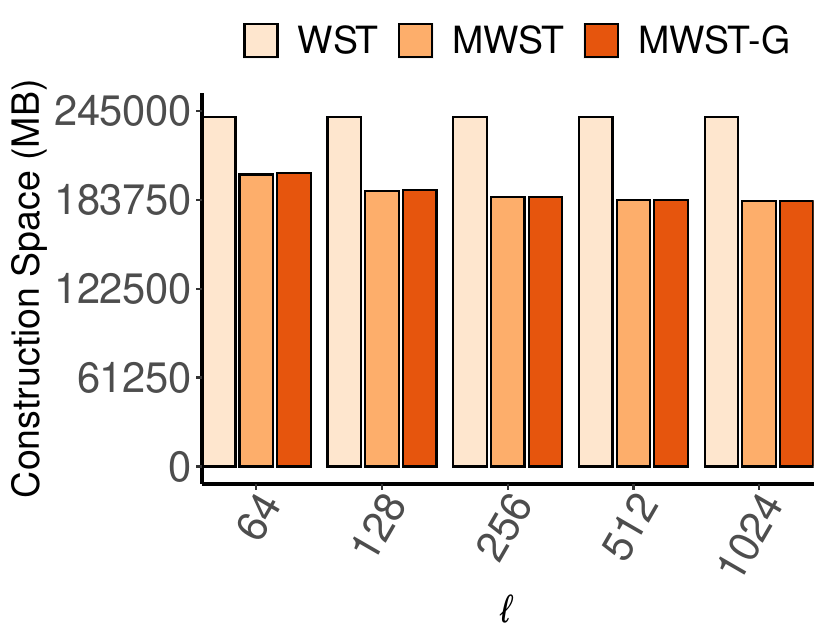}
    \label{fig:construction_space_vs_ell_tree1}
}
\subfloat[][\EFM]
{
    \includegraphics[width=.205\textwidth]{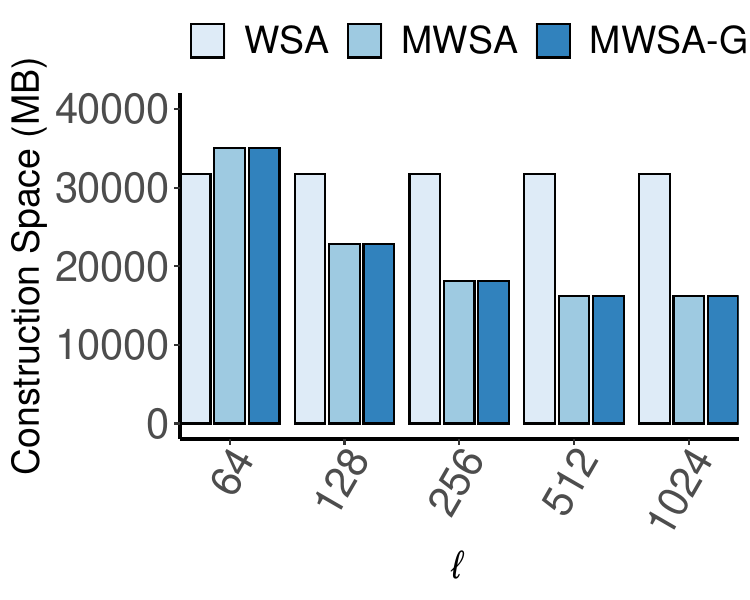}
    \label{fig:construction_space_vs_ell_array1}
}
   \subfloat[][\Z]
{
    \includegraphics[width=.205\textwidth]{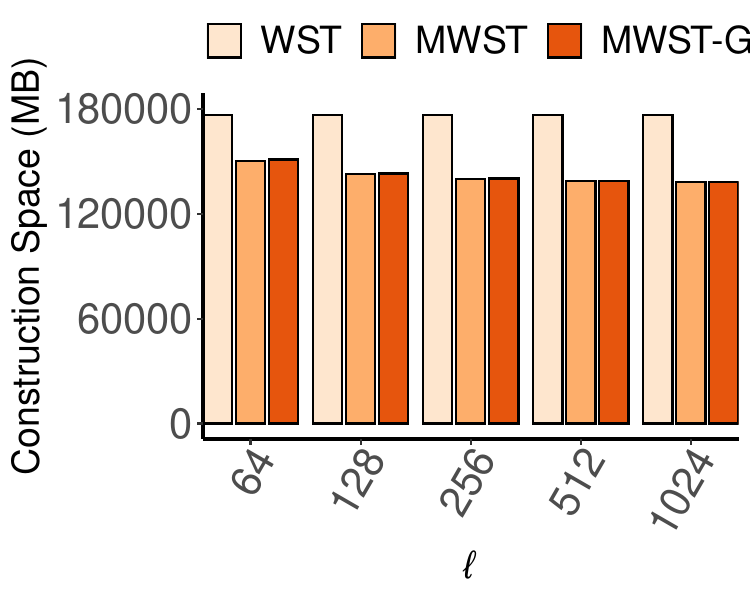}
    \label{fig:construction_space_vs_ell_tree3}
}
 \subfloat[][\Z]
{
    \includegraphics[width=.205\textwidth]{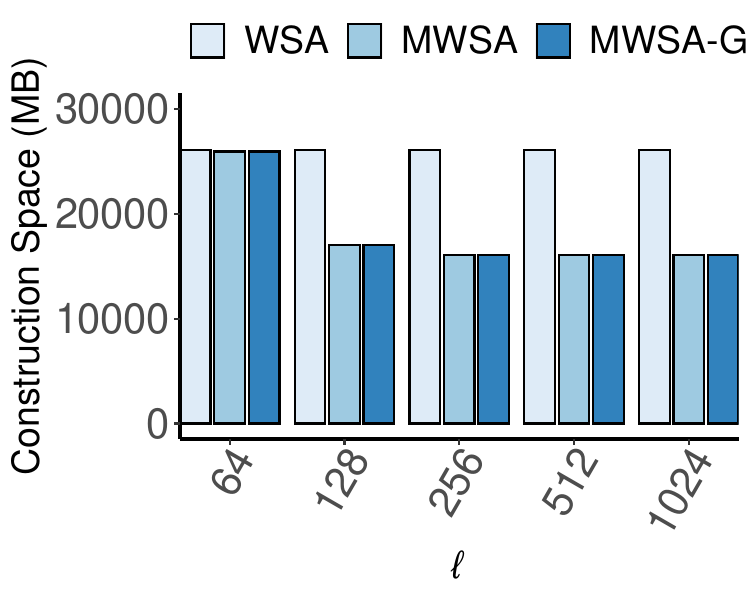}
    \label{fig:construction_space_vs_ell_array3}
}
\caption{Construction space (MB) vs. $\ell$.}\label{fig:construction_space1}
%\vspace{+5mm}
% \subfloat[][\SARS]
% {
%     \includegraphics[width=.205\textwidth]{figures/sars_vs_z_construction_space_tree.pdf}
%     \label{fig:construction_space_vs_z_tree2}
% }
%    \subfloat[][\SARS]
% {
%     \includegraphics[width=.205\textwidth]{figures/sars_vs_z_construction_space.pdf}
%     \label{fig:construction_space_vs_z_array2}
% }
% \\
\subfloat[][\EFM]
{
    \includegraphics[width=.205\textwidth]{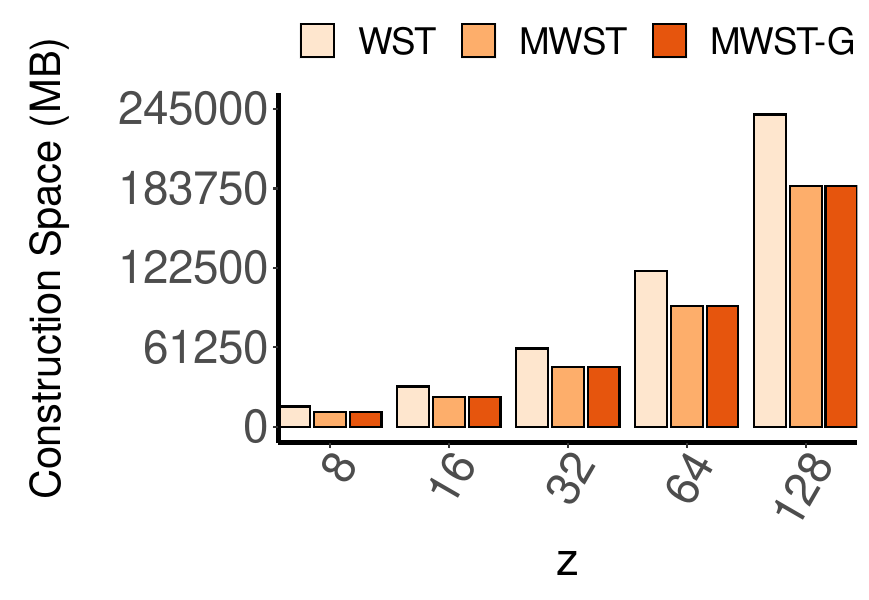}
    \label{fig:construction_space_vs_z_tree1}
}
\subfloat[][\EFM]
{
    \includegraphics[width=.205\textwidth]{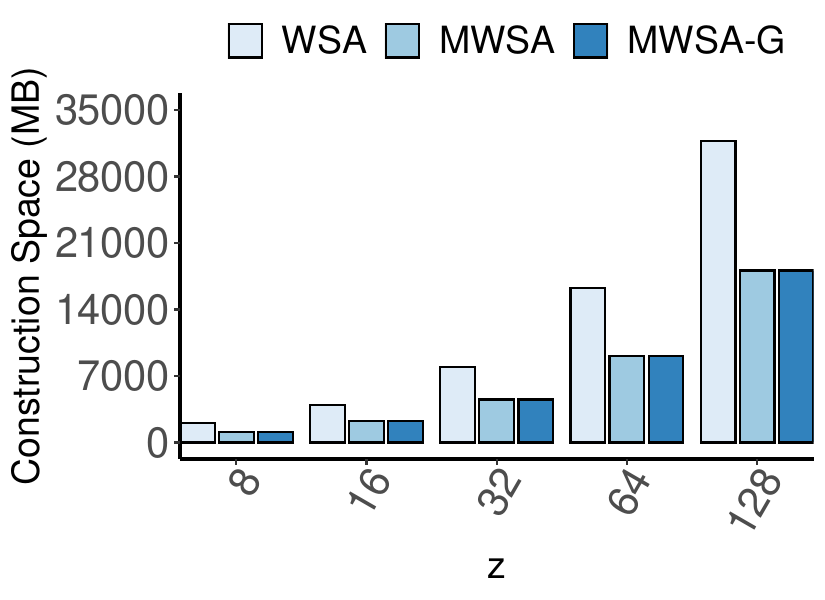}
    \label{fig:construction_space_vs_z_array1}
}
\subfloat[][\Z]
{
    \includegraphics[width=.205\textwidth]{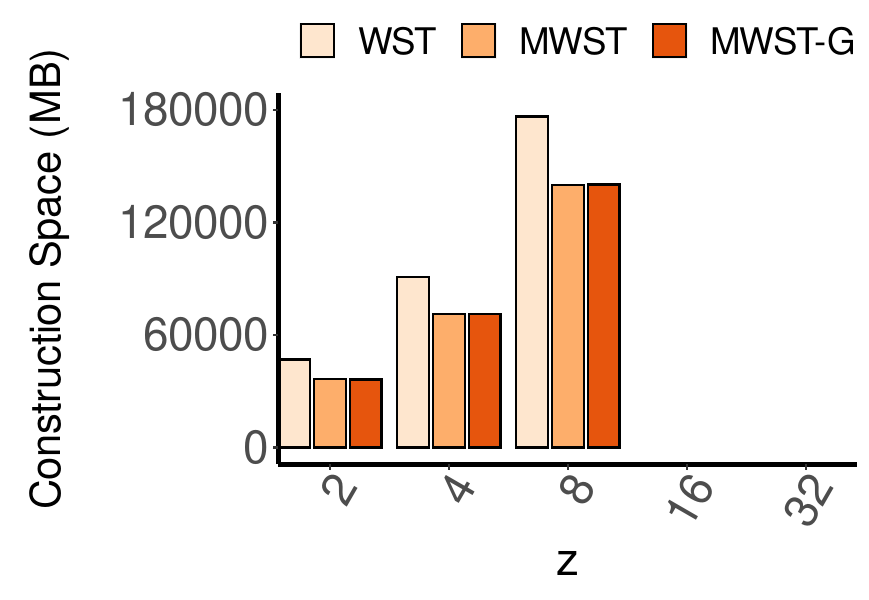}
    \label{fig:construction_space_vs_z_tree3}
} 
 \subfloat[][\Z]
{
     \includegraphics[width=.205\textwidth]{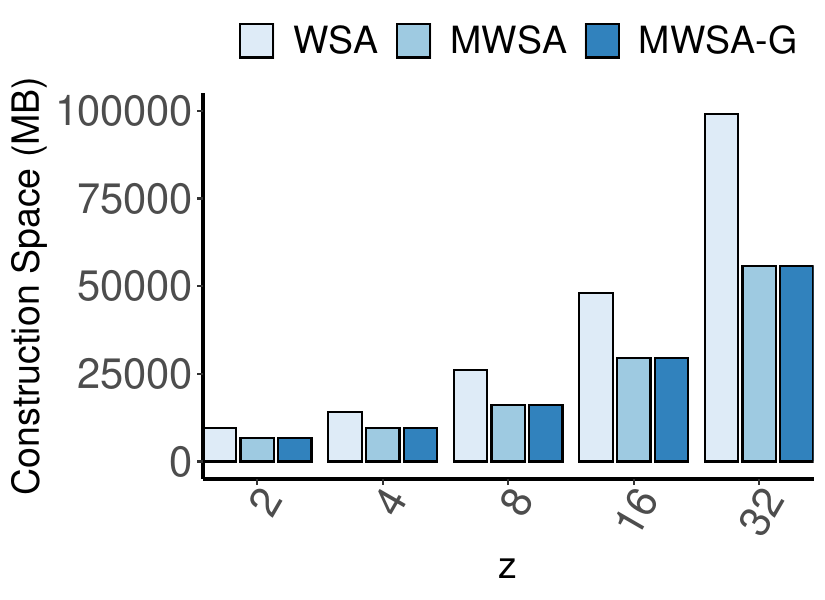}
     \label{fig:construction_space_vs_z_array3}
}
  \caption{Construction space (MB) vs. $z$. The tree-based indexes for \Z (Fig.~\ref{fig:construction_space_vs_z_tree3}) needed $>$ 252GB when $z\geq 16$ and hence could not be constructed.}
    \label{fig:construction_space2}    
\end{figure*}

\begin{figure*}
\centering   
    \subfloat[][\SARS]{
    \hspace{-2mm}
    \includegraphics[width=.24\textwidth]{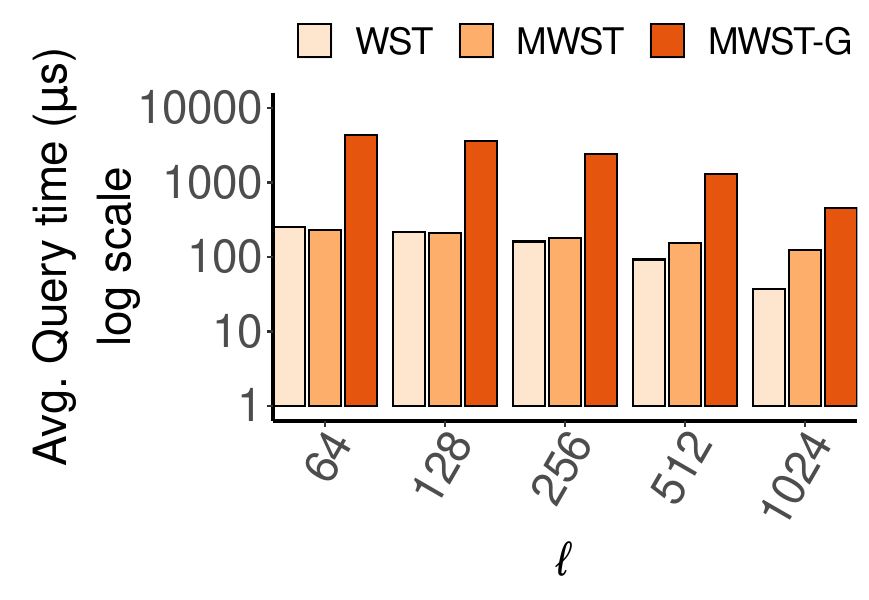}
    }
    \subfloat[][\SARS]{
    \hspace{-2mm}
    \includegraphics[width=.24\textwidth]{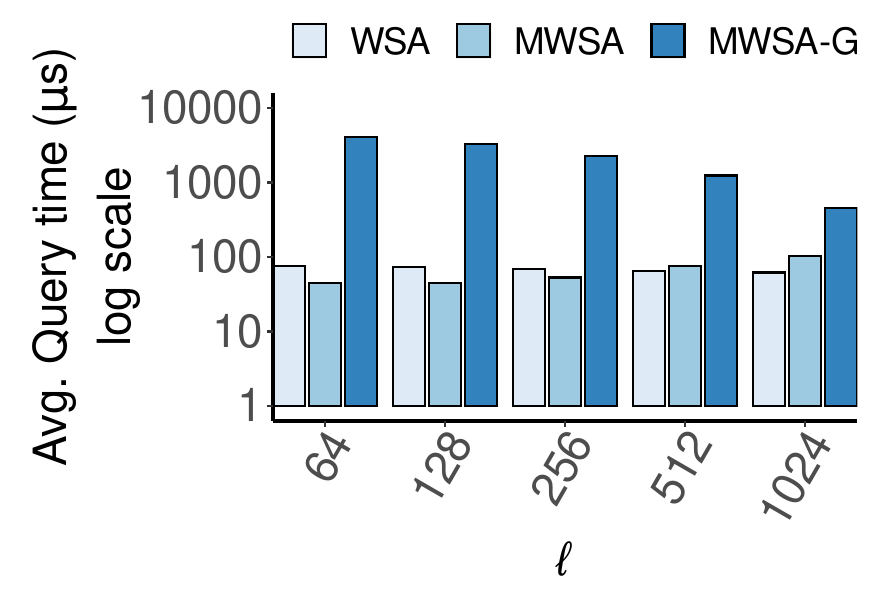}
    } 
    \subfloat[][\EFM]{
    \hspace{-2mm}
    \includegraphics[width=.24\textwidth]{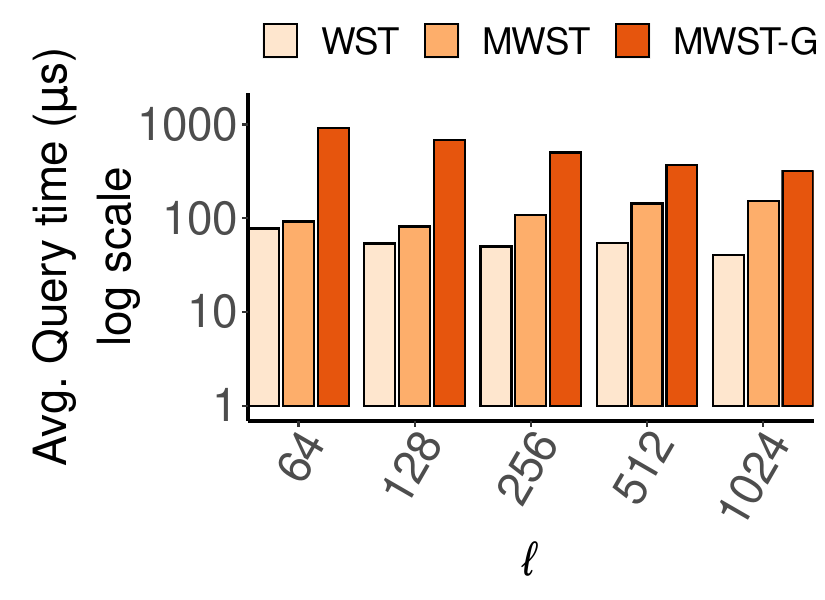}
    }
    \subfloat[][\EFM]{
    \hspace{-2mm}
    \includegraphics[width=.24\textwidth]{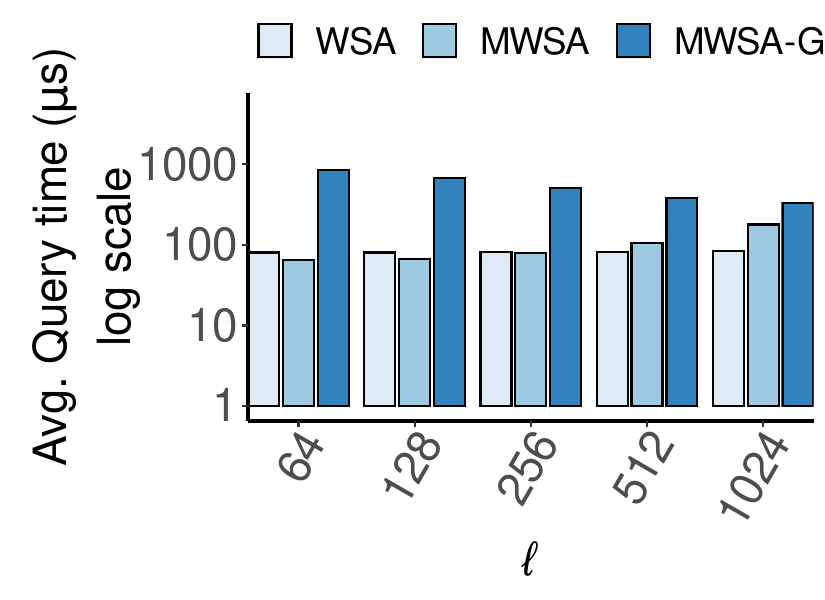}
    }
    \\    
    \subfloat[][\Z]{
    \hspace{-2mm}
    \includegraphics[width=.24\textwidth]{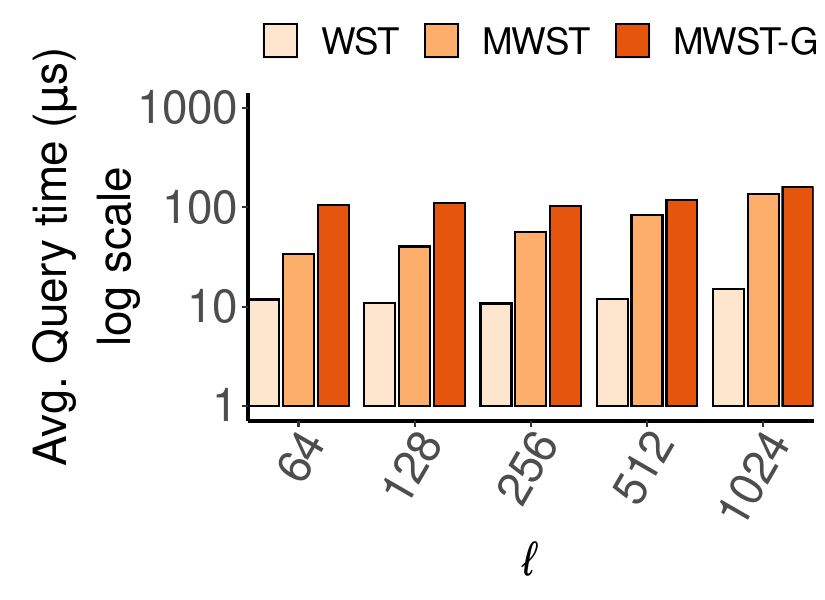}
    }
    \subfloat[][\Z]{
    \hspace{-2mm}
    \includegraphics[width=.24\textwidth]{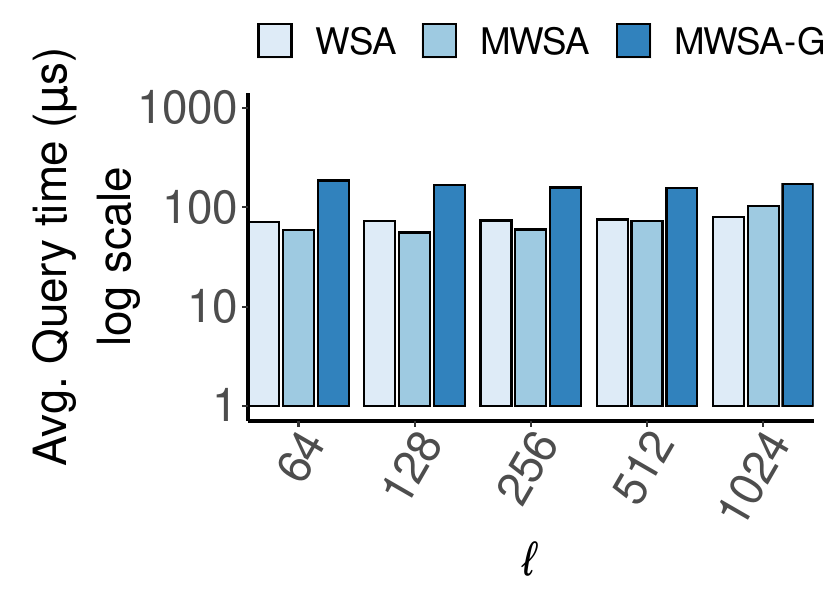}
    }
    \caption{Average query time (log scale, $\mu s$) vs. $
    \ell$.}\label{fig:query_time1}
    
    \vspace{+2mm}       
    \subfloat[][\SARS]{
    \hspace{-2mm}
    \includegraphics[width=.24\textwidth]{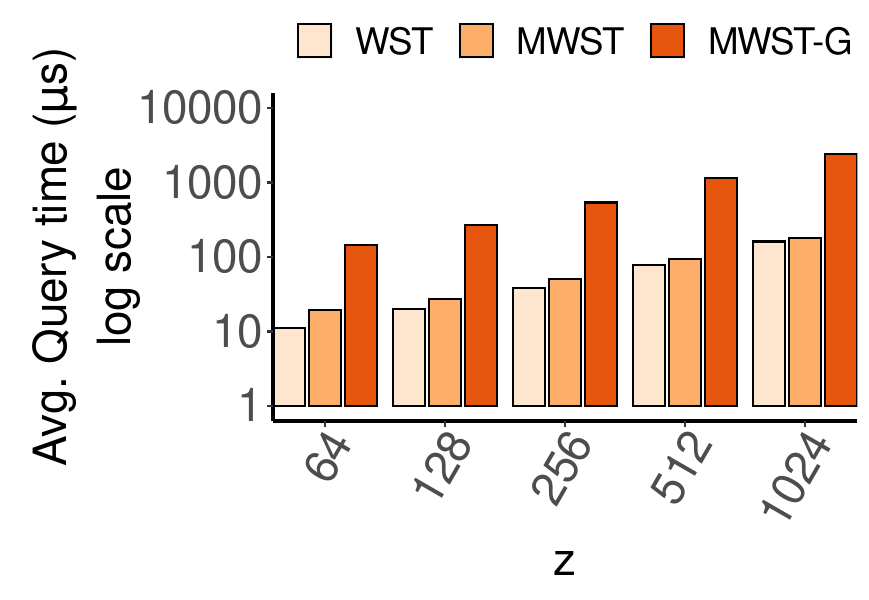}
    }
      \subfloat[][\SARS]
{
    \includegraphics[width=.24\textwidth]{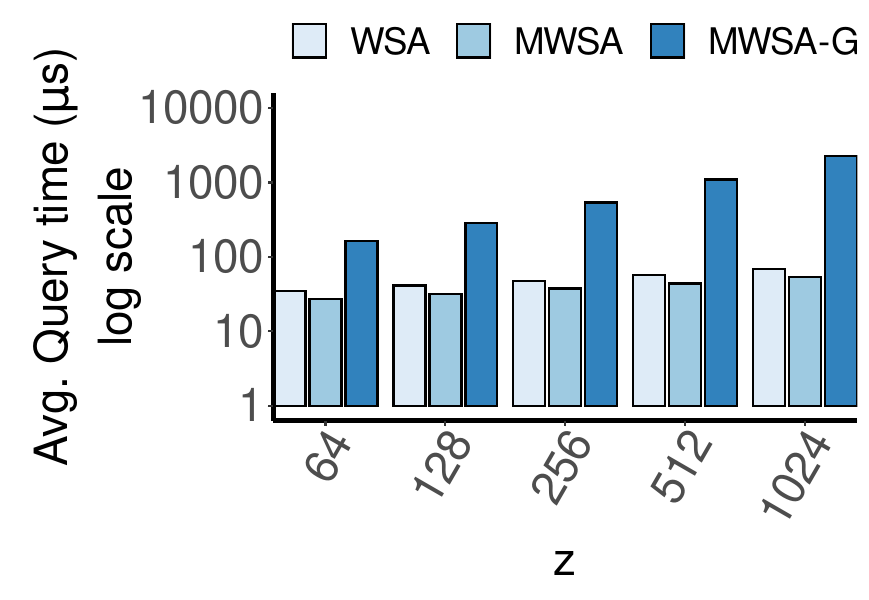}
}
     \subfloat[][\EFM]{
    \hspace{-2mm}
    \includegraphics[width=.24\textwidth]{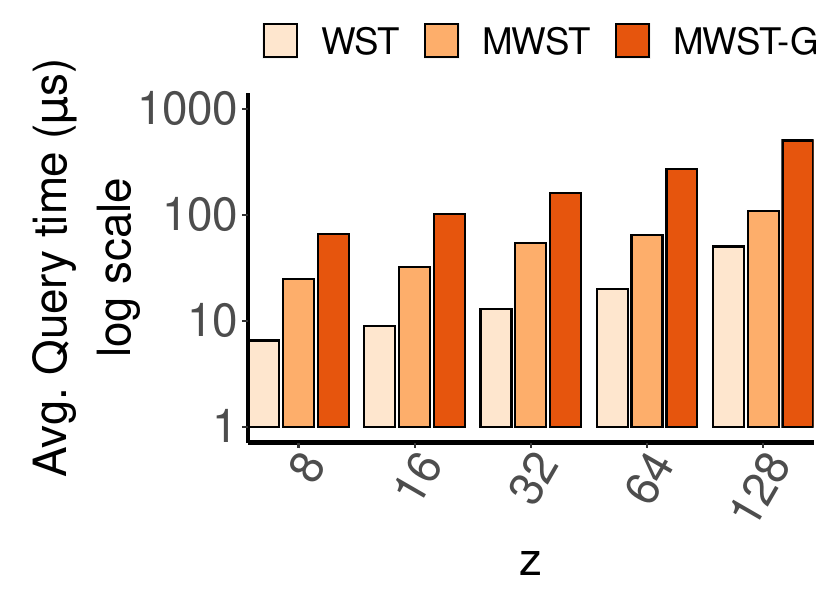}
    }
    \subfloat[][\EFM]{
\hspace{-2mm}
    \includegraphics[width=.24\textwidth]{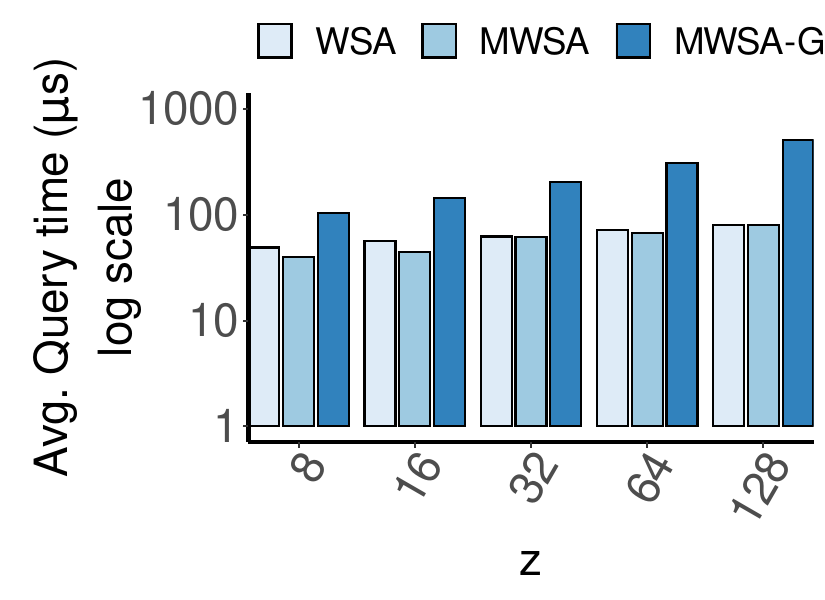}
    }
    \\
    \subfloat[][\Z]{
    \hspace{-2mm}
    \includegraphics[width=.24\textwidth]{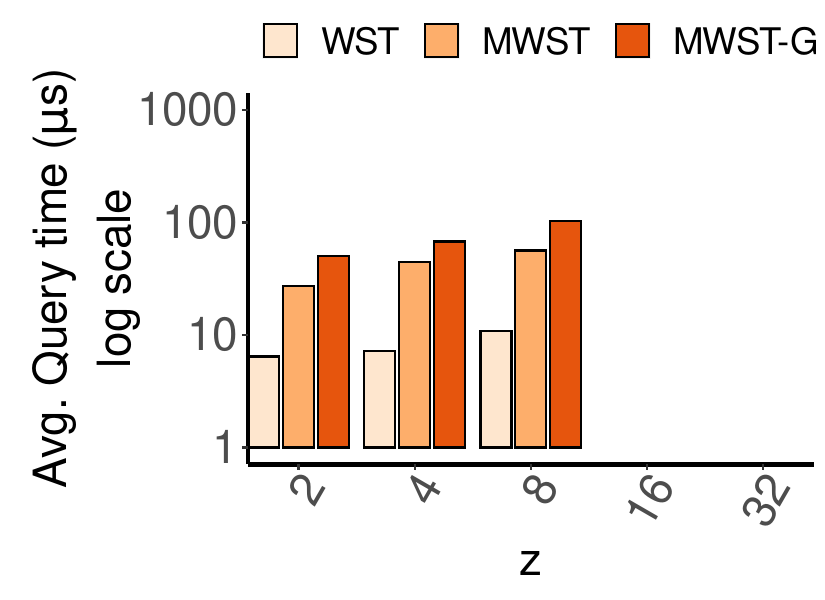}\label{fig:querytime2_tree3}
    } 
 \subfloat[][\Z]
{
    \includegraphics[width=.24\textwidth]{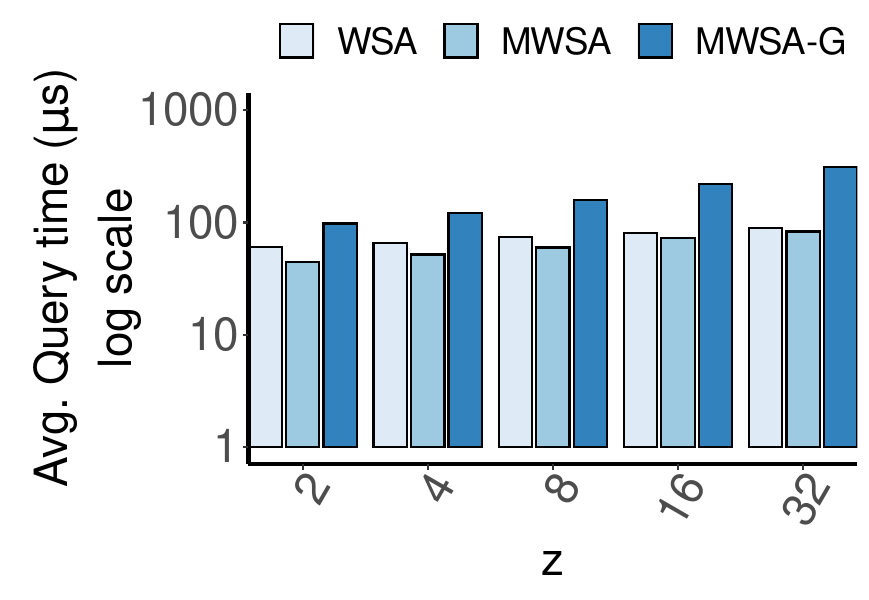}
}
    \caption{Average query time (log scale, $\mu s$) vs. $z$. The tree-based indexes for \Z (Fig.~\ref{fig:querytime2_tree3})  needed $>$ 252GB when $z\geq 16$ and hence could not be constructed.}\label{fig:query_time2}
\end{figure*}
\begin{figure*}
\centering
    \subfloat[][\EFM]{
    \hspace{-2mm}
    \includegraphics[width=.22\textwidth]
    %{figures/efm_vs_ell_construction_time_tree.pdf}
    {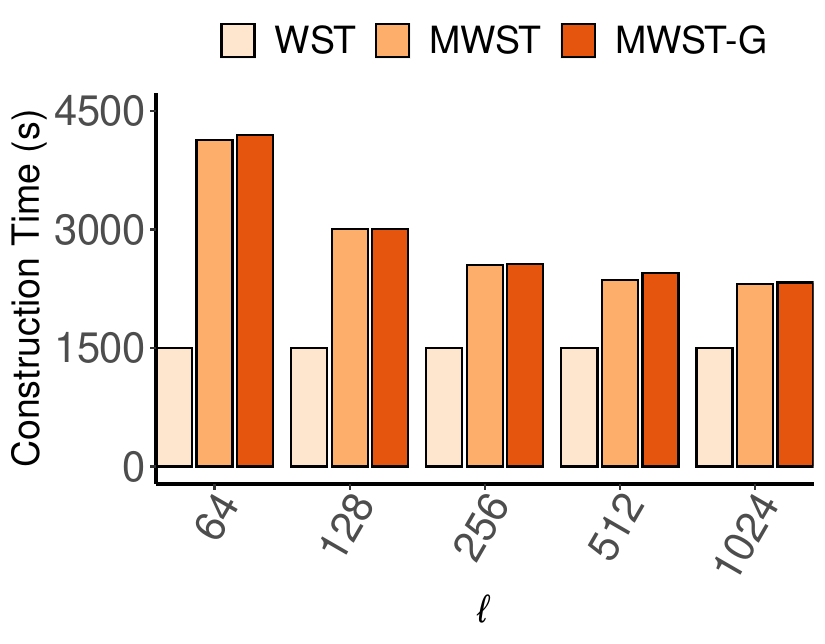}
    }
    \subfloat[][\EFM]{
    \hspace{-2mm}
    \includegraphics[width=.22\textwidth]
    %{figures/efm_vs_ell_construction_time.pdf}
    {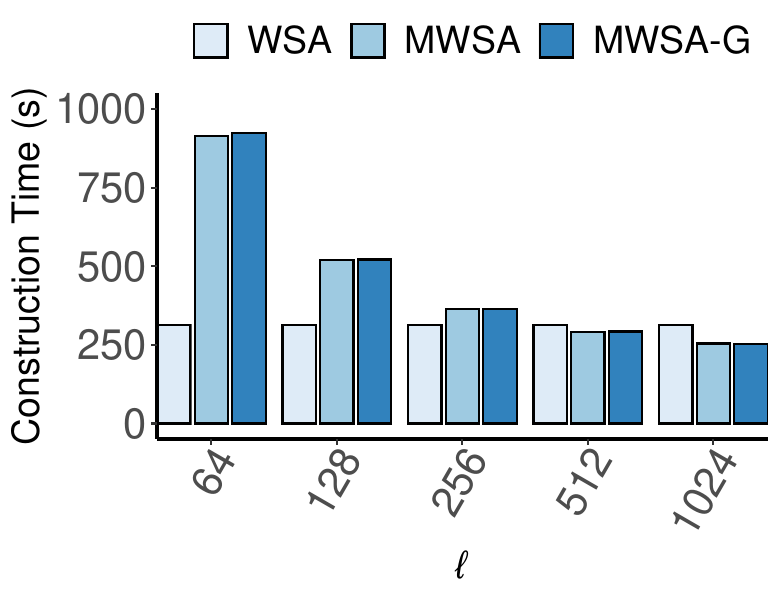}
    }
    \subfloat[][\EFM]{
    \hspace{-2mm}
    \includegraphics[width=.22\textwidth]{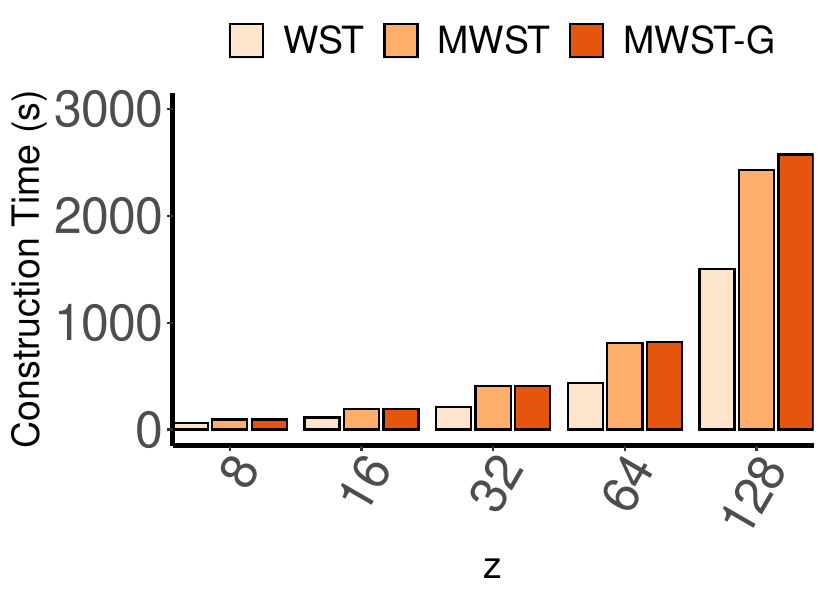}
    } %Rerunning; MWST: 2428665 and MWST-G: 2575968 [Greg: Done]
     \subfloat[][\EFM]{
    \hspace{-2mm}
    \includegraphics[width=.22\textwidth]
    %{figures/efm_vs_z_construction_time.pdf}
    {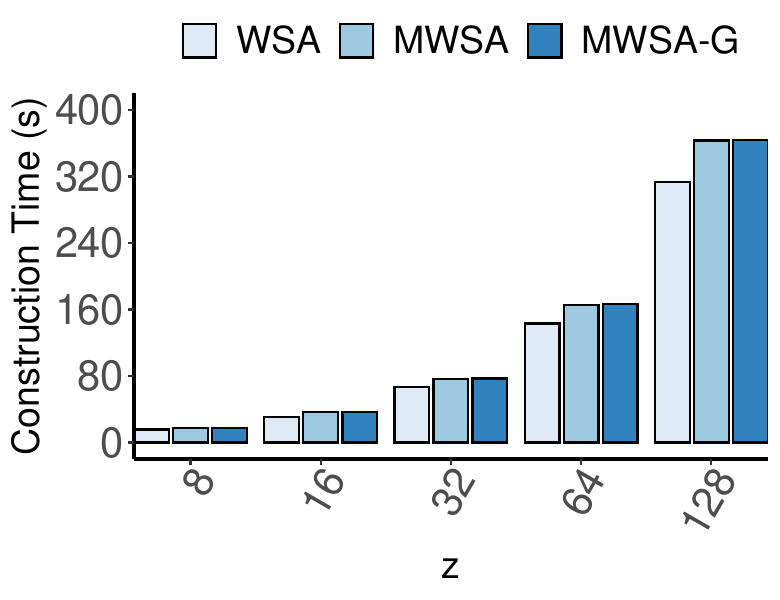}
    }
    \caption{(a, b) Construction time ($s$) vs. $\ell$ for \EFM. (c, d) Construction time ($s$) vs. $z$  for \EFM. The results for \SARS and \Z were analogous.}\label{fig:construction_time1}
\end{figure*}

\noindent{\bf Parameters.}~For every weighted string of length $n$, every pattern length $m\in\{64,128,256,512,1024\}$, and every $z$ we used, we selected $\lfloor nz/200 \rfloor$ patterns from the $z$-estimation of the weighted string, uniformly at random, to account for the different $n$ and $z$ values. For example, for \Z whose length is $n=35,194,566$, and for $z = 32$, we have selected $5,631,130$ patterns uniformly at random from its $32$-estimation. The default $z$ for \SARS, \EFM, \Z, \RS, and $\RS_{n,\sigma}$ was 
$1024$, $128$, $8$, $16$, and $16$,  and led to $z$-estimations with sizes of several MBs; see Table~\ref{table:data}. The parameter  
$\ell$ was set to $m$, and the default $m$ value was $256$. 

\noindent{\bf Implementations.}~We used the implementations of the state-of-the-art indexes \WST and \WSA  from~\cite{DBLP:journals/iandc/BartonK0PR20} and~\cite{DBLP:journals/jea/Charalampopoulos20}, respectively. We implemented: (1) \MWSTG, the algorithm underlying our Theorem~\ref{the:2d-structure}. (2) \MWST, a simplified version of \MWSTG that drops the 2D grid and performs pattern matching as described in Section~\ref{sec:fast_pm}. (3) \MWSA and \MWSAG, the \emph{array-based} versions of \MWST and \MWSTG, respectively. Note that an in-order DFS traversal of the tree gives the array.
(4) \MWSTSE, the space-efficient construction of \MWST underlying Theorem~\ref{the:space_efficient}. 
In all implementations, we used Karp-Rabin fingerprints~\cite{DBLP:journals/ibmrd/KarpR87} 
to compute the minimizers. 

\noindent{\bf Measures.}~We used all four relevant measures of efficiency (see Introduction): index size; query time; construction space; and construction time. To measure the query and construction time, we used the \texttt{chrono C++} library. To measure the index size, we used the \texttt{malloc2 C++} function. To measure the construction space, we recorded the maximum resident set size  
using the \texttt{/usr/bin/time -v} command. 

\noindent {\bf Environment.}~All experiments ran using a single AMD EPYC 7282 CPU at 2.8GHz with 252GB RAM under GNU/Linux. All methods were implemented in \texttt{C++} and compiled with \texttt{g++} (v.~12.2.1) at optimization level \texttt{-O3}. 

\noindent{\bf Code and Datasets.}~The code and all datasets are available at \url{https://github.com/solonas13/ius} under GNU GPL v3.0.

\subsection{Evaluating our Minimizer-based Indexes}

% This section shows that: (1) our indexes are up to two orders of magnitude smaller than the state-of-the-art indexes and can be constructed in much less space; (2) our indexes have query and construction times that are competitive to that of the state of the art; and (3) our simplified indexes allow faster queries than the grid-based ones despite having weaker guarantees.  

\noindent{\bf Index Size.}~Figs.~\ref{fig:index_size1} and~\ref{fig:index_size2} show that our tree-based (resp.~array-based) indexes occupy \emph{up to two orders of magnitude less space} than \WST (resp.~\WSA). The size of our indexes decreases with $\ell$ and increases with $z$ (see Theorem~\ref{the:2d-structure}).  
Furthermore, the array-based indexes occupy several 
times less space than the tree-based ones, as it is widely known~\cite{DBLP:journals/jacm/KarkkainenSB06}. For example, note from Figs.~\ref{fig:index_size_vs_ell_tree1} and~\ref{fig:index_size_vs_ell_array1} that for $\ell=1024$, \WST occupied  126GB of space, whereas our \MWST 900MB and \MWSA only 204MB! As expected, our grid-based indexes \MWSTG and \MWSAG occupy a slight amount of extra space compared to \MWST and \MWSA, respectively.  

\begin{figure*}
    \hspace{-2mm}
    \subfloat[][\EFM]{\label{fig:construction_space_se_efm_ell}    
    \includegraphics[width=.24\textwidth]
    %{figures/efm_vs_ell_construction_space_se3.pdf}
    %{figures/efm_vs_ell_construction_space_se4.pdf}
    {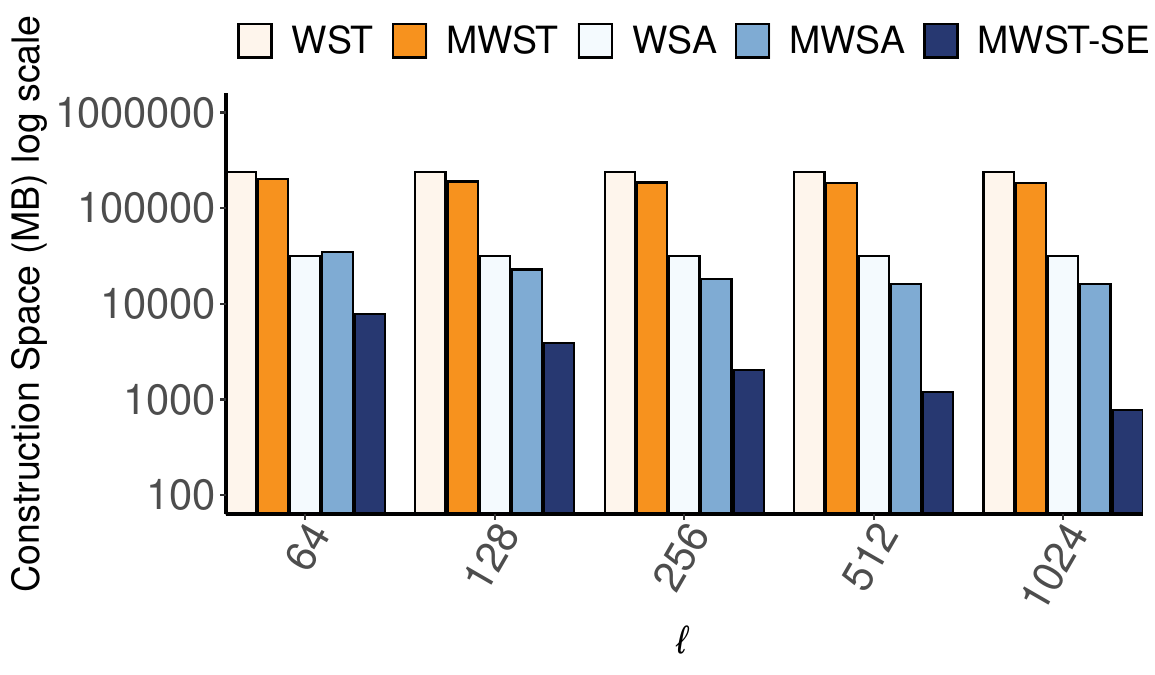}
}
    \subfloat[][\Z]{\label{fig:construction_space_se_z_ell}
    \hspace{-2mm}
    \includegraphics[width=.235\textwidth]{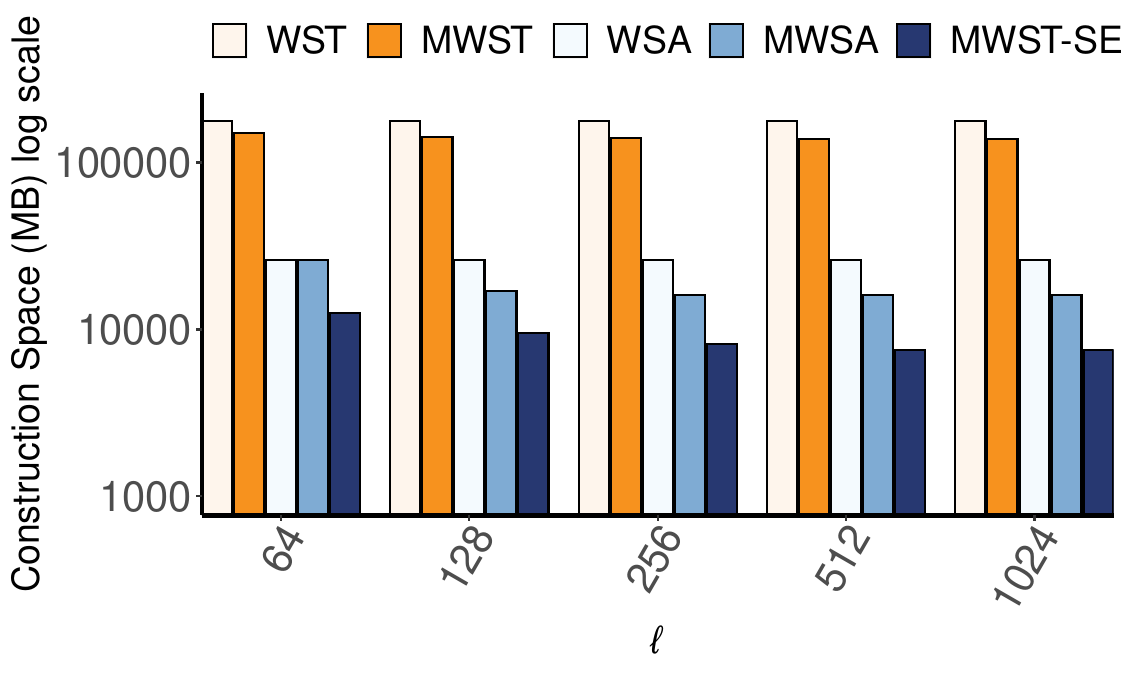}
}
\subfloat[][\EFM]{
    \hspace{-2mm}\label{fig:construction_space_se_efm_z} 
    \includegraphics[width=.24\textwidth]{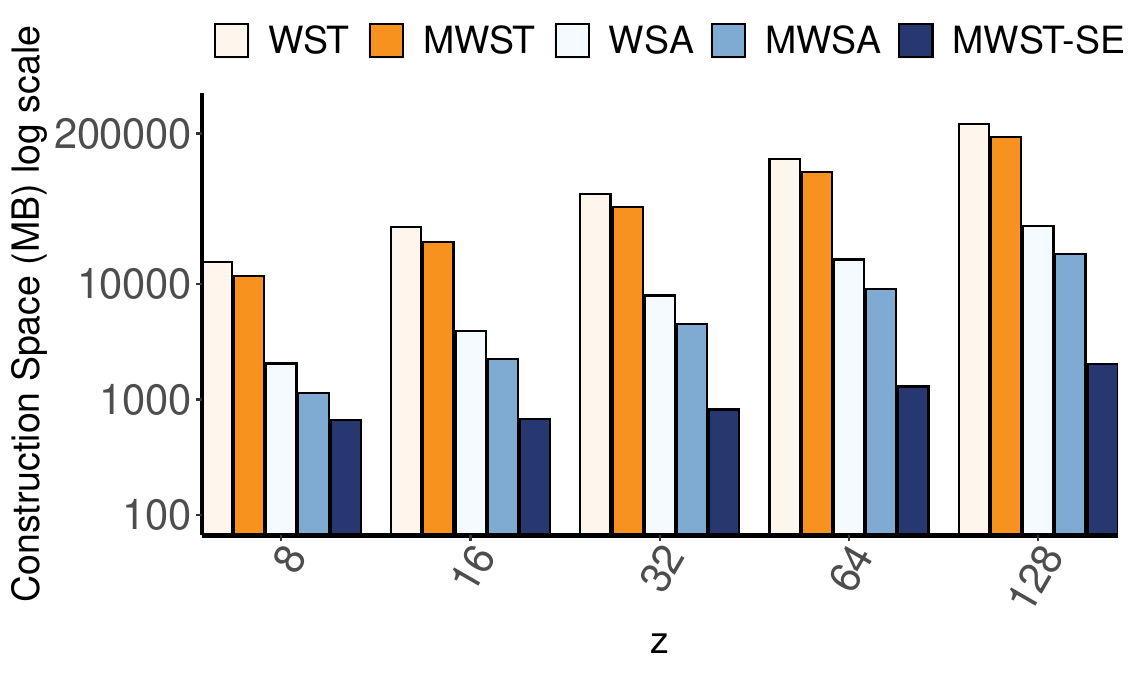}
    }
    \subfloat[][\Z]{\label{fig:construction_space_se_d}
    \hspace{-2mm}
    \includegraphics[width=.24\textwidth]{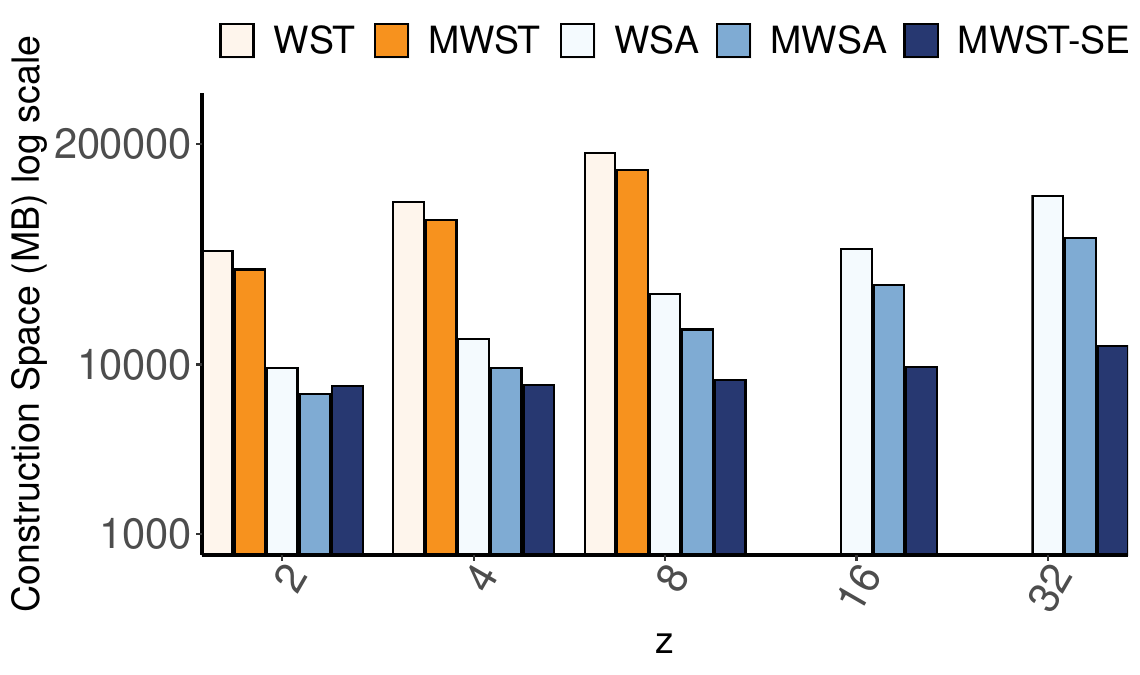}
    }
    \caption{Construction space (log scale, MB) vs: (a, b) $
    \ell$. (c, d) $z$. \WST and \MWST for \Z (Fig.~\ref{fig:construction_space_se_d})  needed $>$ 252GB when $z\geq 16$ and hence could not be constructed. The results for \SARS were analogous.}
   \label{fig:construction_space_se}
\centering
      \subfloat[][$\RS$]{
    \hspace{-2mm}\label{fig:sen_vs_ell_construction_space}
    \includegraphics[width=.24\textwidth]
    %{figures/sen_vs_ell_construction_space_se.pdf}
    {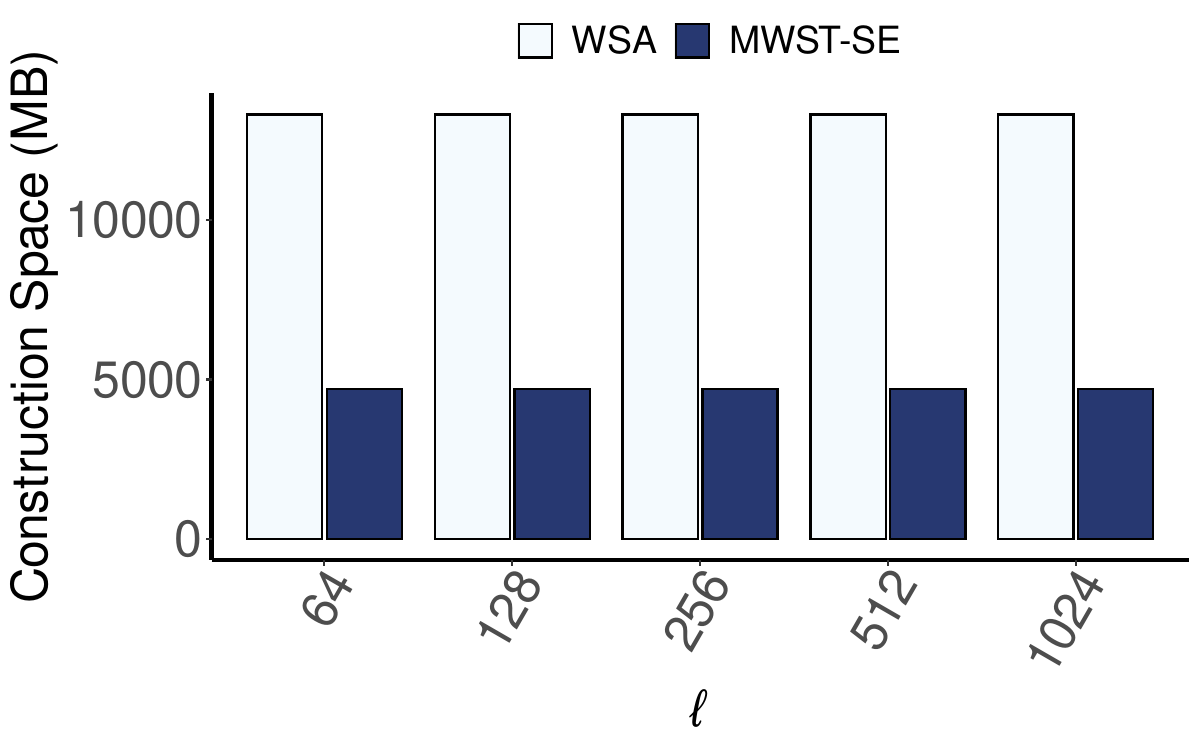}
    }
    \subfloat[][$\RS$]{
    \hspace{-2mm}\label{fig:sen_vs_z_construction_space}
    \includegraphics[width=.24\textwidth]
    %{figures/sen_vs_z_construction_space_se.pdf}
       {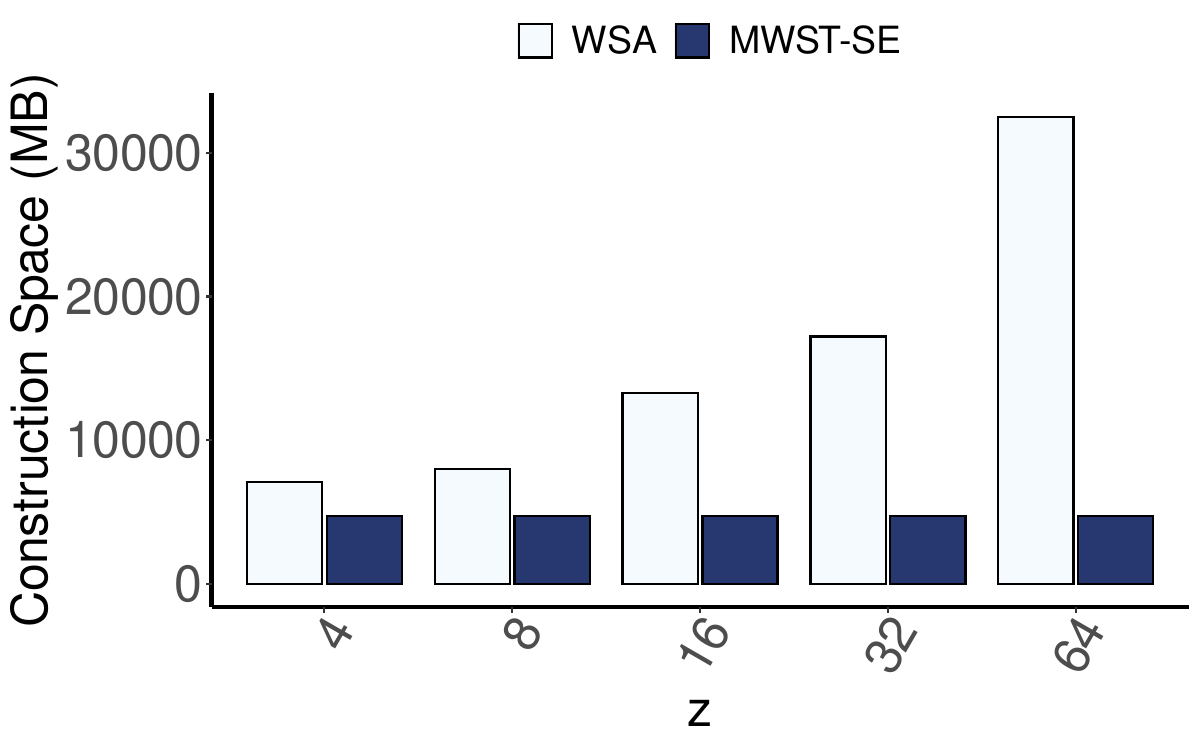}
    }      
       \subfloat[][$\RS_{1,16}, \ldots,\RS$]{
    \hspace{-2mm}\label{fig:sen_vs_sigma_construction_space}
    \includegraphics[width=.24\textwidth]{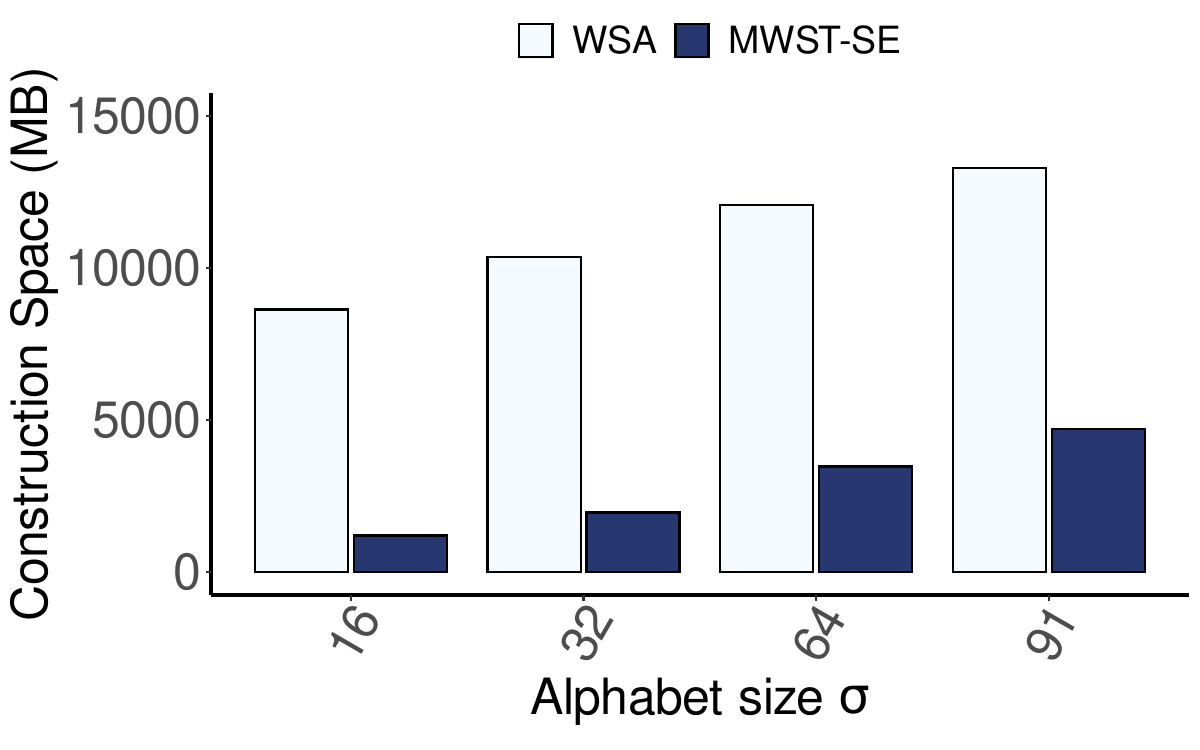}
    }
    \subfloat[][$\RS_{1,32}, \ldots,\RS_{8,32}$]{
    \hspace{-2mm}\label{fig:sen_vs_n_construction_space}
\includegraphics[width=.24\textwidth]{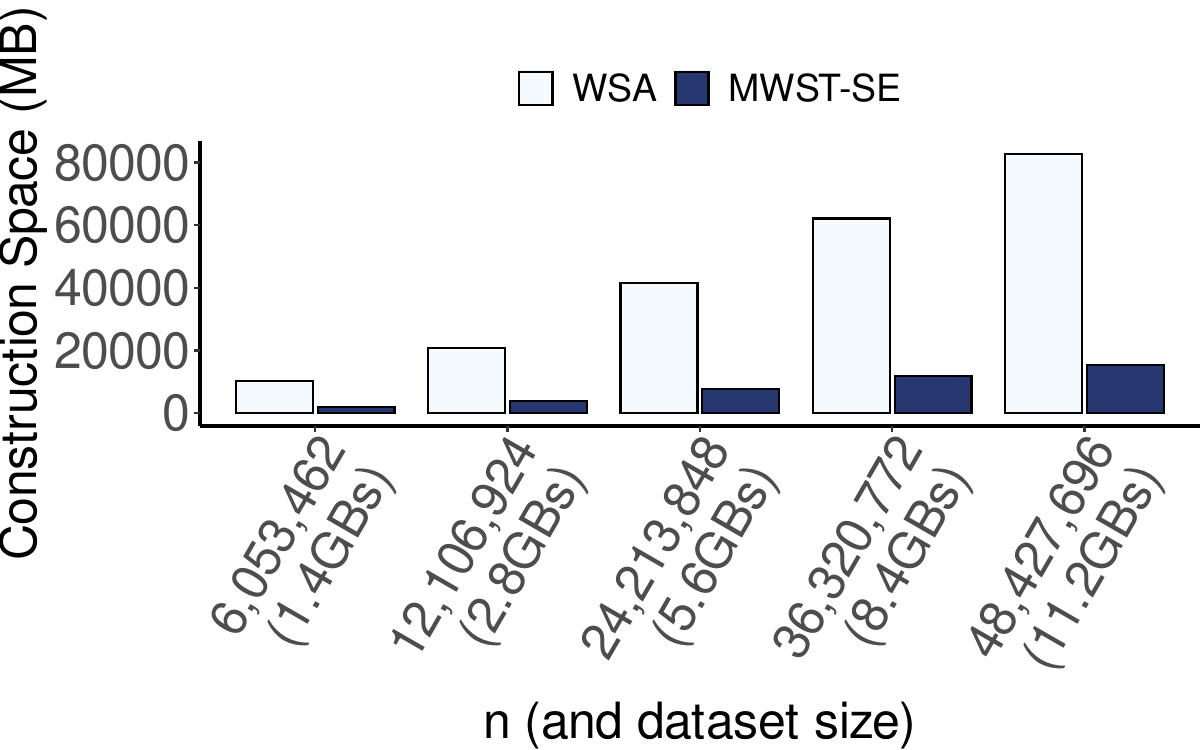}
    }
    \caption{Construction space (MBs) vs. 
    (a) $\ell$, (b) $z$, (c) $\sigma$, and (d) $n$ (and dataset size).}\label{fig:construction_space_rssi}    
    \subfloat[][\EFM]{
    \hspace{-4mm}\label{fig:construction_time_se_efm_ell} 
    \includegraphics[width=.24\textwidth]%%
    %{figures/efm_vs_ell_construction_time_se4.pdf}
    {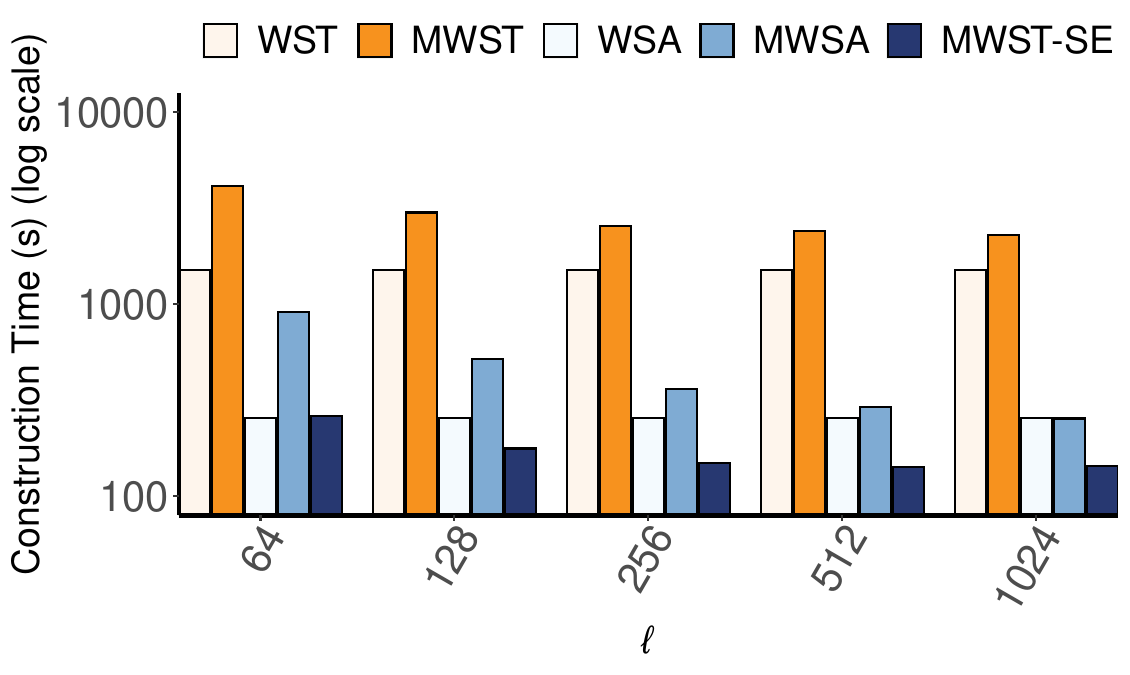}
    }
    \subfloat[][\Z]{\label{fig:construction_time_se_z_ell} 
    \hspace{-3mm}
    \includegraphics[width=.25\textwidth]
    %{figures/z_vs_ell_construction_time_se4.pdf}
    {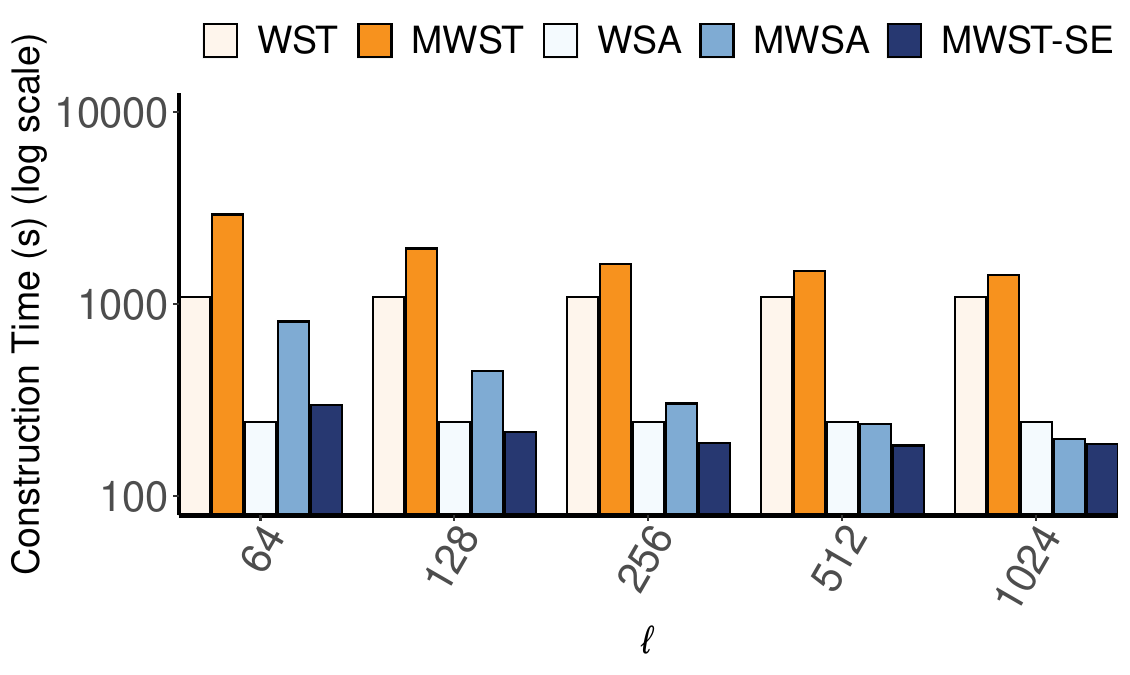}
    }
    \subfloat[][\EFM]{\label{fig:construction_time_se_efm_z} 
    \hspace{-4mm}
    \includegraphics[width=.24\textwidth]
    %{figures/efm_vs_z_construction_time_se4.pdf}
    {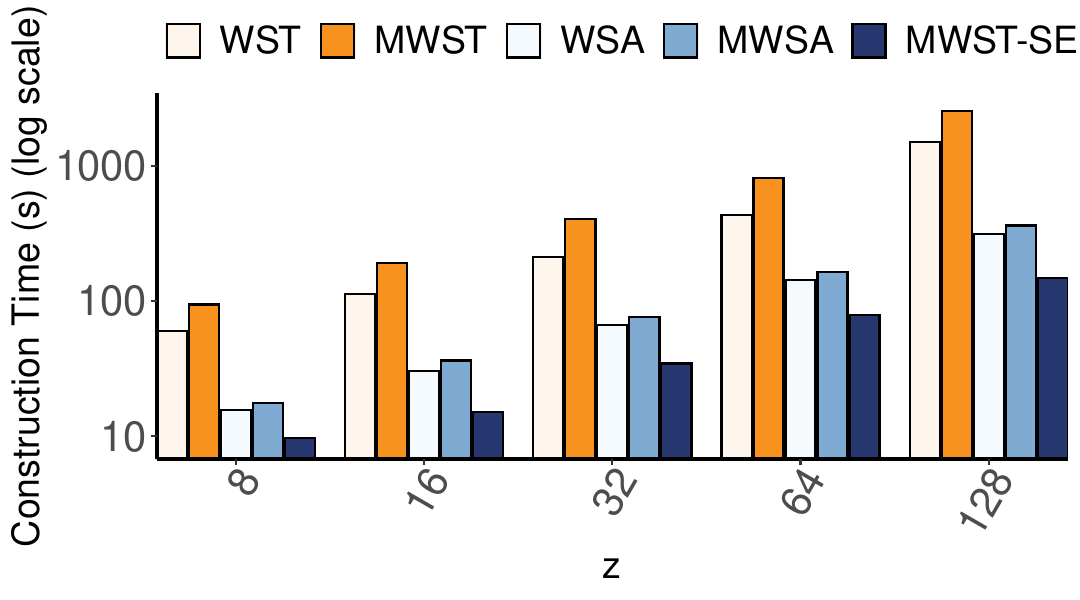}
    }
    \subfloat[][\Z]{\label{fig:construction_time_se_d}
    \hspace{-3mm}
    \includegraphics[width=.24\textwidth]
    %{figures/z_vs_z_construction_time_se5.pdf}
    {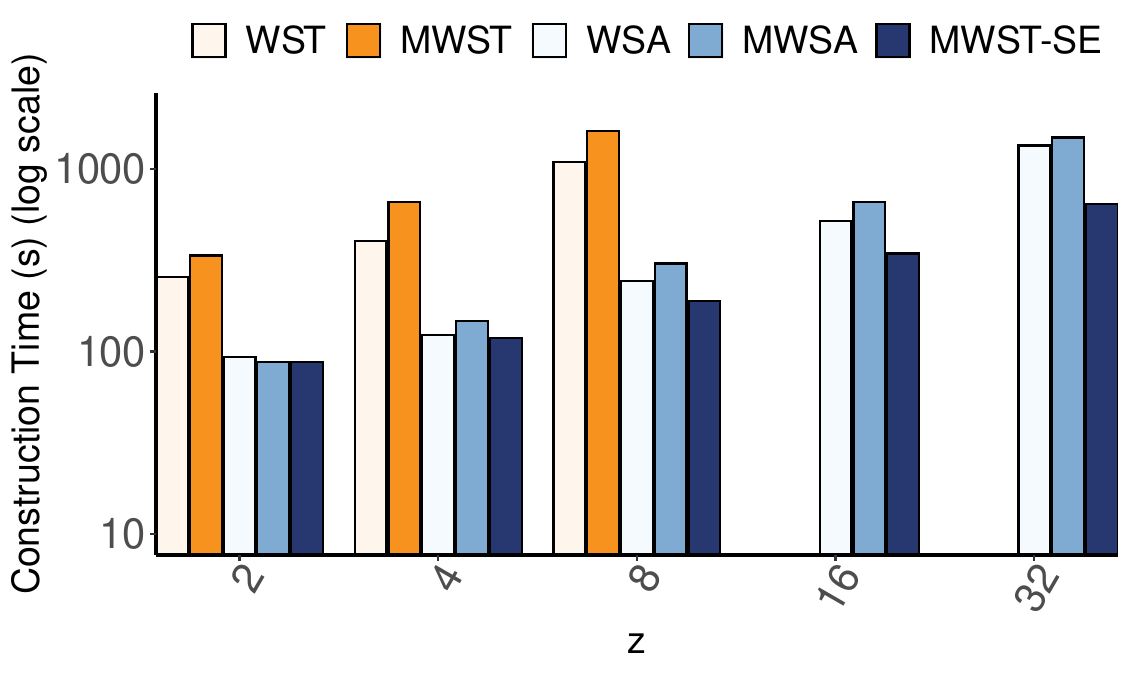}
    }
    \caption{Construction time (log scale, $s$) vs: (a, b) $
    \ell$. (c, d) $z$. \WST and \MWST for \Z (Fig.~\ref{fig:construction_time_se_d}) needed $>$ 252GB when $z\geq 16$ and hence could not be constructed. The results for \SARS were analogous.}
   \label{fig:construction_time_se}
   \end{figure*}
   \begin{figure*}[t]
\centering
\subfloat[][$\RS$]{\label{fig:sen_vs_ell_construction_time}
    \hspace{-2mm}
\includegraphics[width=.24\textwidth]
%{figures/sen_vs_ell_construction_time_se_rev}
{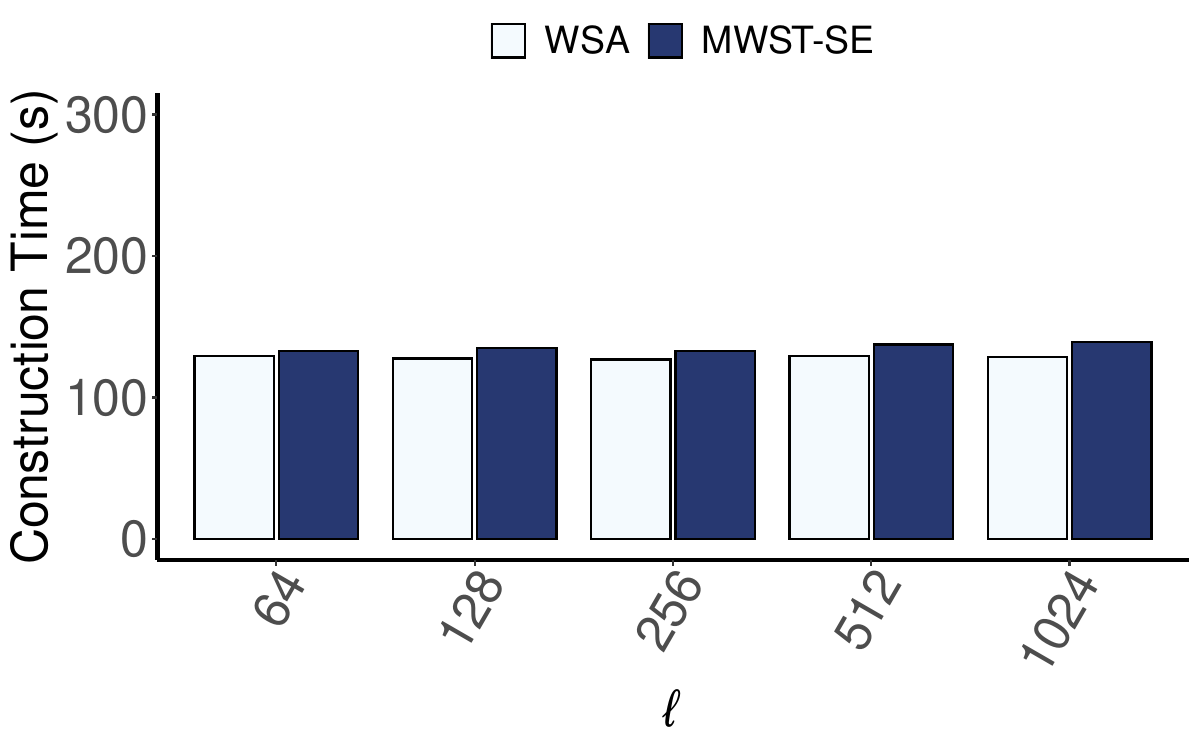}     
    }
 \subfloat[][$\RS$]{\label{fig:sen_vs_z_construction_time}
    \hspace{-2mm}
    \includegraphics[width=.24\textwidth]
    %{figures/sen_vs_z_construction_time_se.pdf}
    {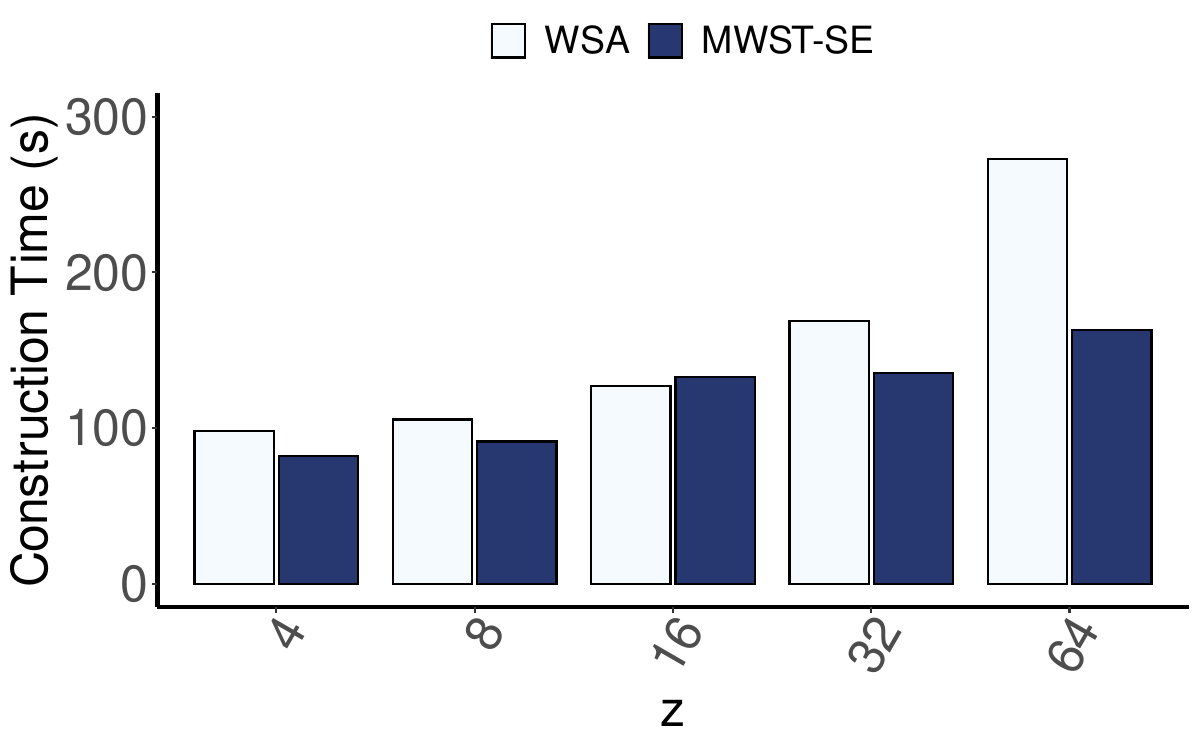}
    }   
    \subfloat[][$\RS_{1,16}, \ldots,\RS$]{
    \hspace{-2mm}\label{fig:sen_vs_sigma_construction_time}
    \includegraphics[width=.24\textwidth]{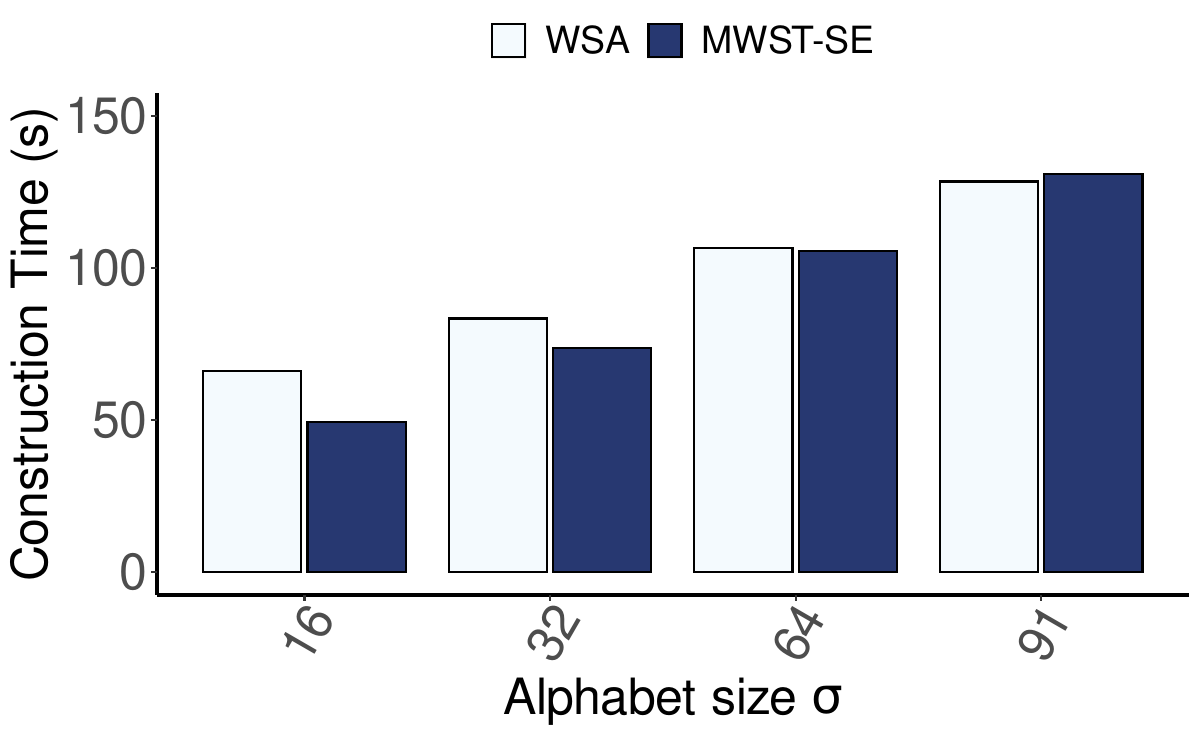}
    }
     \subfloat[][$\RS_{1,32}, \ldots,\RS_{8,32}$]{
    \hspace{-2mm}\label{fig:sen_vs_n_construction_time}
    \includegraphics[width=.24\textwidth]{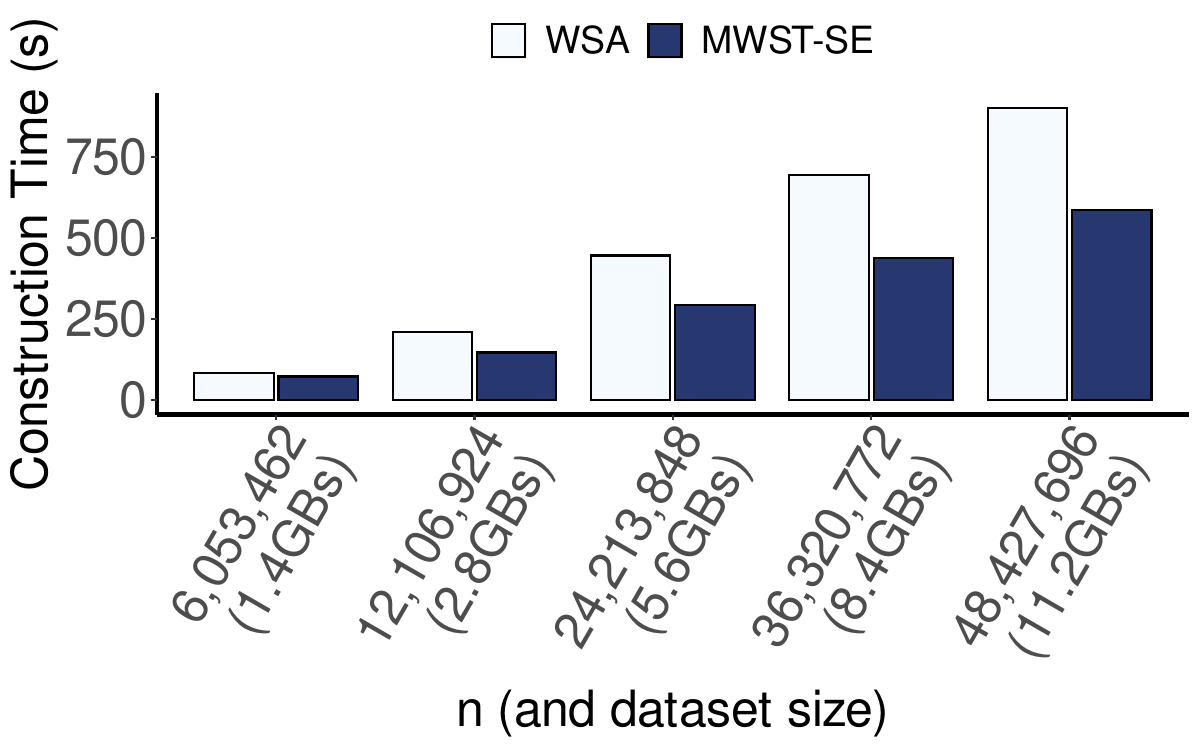}
    }
   %      \subfloat[][$\RS_{1,32}$]{
   %  \hspace{-2mm}
   %  \includegraphics[width=.49\textwidth]{figures/sen_vs_ell_query_time_se2.pdf}
   %  }
   %     \subfloat[][$\RS_{1,32}$]{
   %  \hspace{-2mm}
   %  \includegraphics[width=.49\textwidth]{figures/sen_vs_z_query_time_se.pdf}
   %  }\\
   % \subfloat[][$\RS_{1,32}$]{
   %  \hspace{-2mm}
   %  \includegraphics[width=.49\textwidth]{figures/sen_vs_sigma_query_time_se.pdf}
   %  }    
    \caption{Construction time ($s$) vs. (a) $\ell$, (b) $z$, (c) $\sigma$, and (d) $n$ (and dataset size). 
    %(g to i): Average query time ($\mu s$) vs. (g) $z$, (b) $\ell$, and (c) $\sigma$.}
    }\label{fig:construction_time_rssi}
\end{figure*}

\noindent{\bf Construction Space.}~Figs.~\ref{fig:construction_space1} and~\ref{fig:construction_space2} show that our tree-based (resp.~array-based) indexes outperform \WST (resp.~\WSA) \emph{by $27\%$ (resp.~$61\%$)} on average; the results for \SARS are analogous to those for \EFM. Although our construction algorithm (see Theorem~\ref{the:2d-structure}) takes $\Theta(nz)$ space in any case, it carries
lower constant factors than that of \WST. That is, in practice, the index construction space for our tree-based indexes decreases as $\ell$ increases and increases with $z$ -- see Lemmas~\ref{lem:minimizer tree size} and~\ref{lem:2demptiness}, which show a clear dependency on the number $\cO(\frac{nz}{\ell})$ of sampled minimizers. 
The same explanation holds for \WSA and our array-based indexes.  
Again, as it is widely known~\cite{DBLP:journals/jacm/KarkkainenSB06}, the array-based indexes outperform the tree-based ones in terms of space; and, as expected, \MWSTG and \MWSAG need a very slightly larger construction space than \MWST and \MWSA, respectively.  

\noindent{\bf Query Time.}~Figs.~\ref{fig:query_time1} and~\ref{fig:query_time2} show that 
\MWST is generally slower than \WST because its search operation is more costly than that  of \WST (see Theorem~\ref{the:2d-structure}). 
However, \MWSA is competitive to \WSA since, due to the  smaller size of the former, the binary search operation used in query answering (pattern matching)~\cite{DBLP:journals/siamcomp/ManberM93} becomes faster. This is \emph{very encouraging} given its substantially smaller  index size and index construction space across all $z$ and $\ell$ values.  
Furthermore, \MWST and \MWSA outperform \MWSTG and \MWSAG, respectively. This is in line with the findings of~\cite{DBLP:journals/pvldb/AyadLP23,DBLP:conf/esa/LoukidesP21}, which show that simple verification schemes like the one developed by us in Section~\ref{sec:fast_pm}, are faster than  grid approaches, even if the theoretical guarantees provided by the former are weaker. Note in Fig.~\ref{fig:query_time1} that the query time of the grid-based indexes is not negatively affected by increasing $\ell$, unlike \MWST and \MWSA. In particular, Fig.~\ref{fig:query_time1} shows that, although the grid-based indexes are slower, the difference in performance \emph{decreases} as $\ell$ grows. This is because, as $\ell$ grows, the grid becomes smaller and the simple verification schemes become more expensive, which highlights the benefit of Theorem~\ref{the:2d-structure}.
The query time of all indexes increases with $z$, as expected by their time complexities. The query time of \WST and \WSA does not depend on $\ell$, as expected by their time complexities. 

\noindent{\bf Construction Time.}~Fig.~\ref{fig:construction_time1} shows that \WST and \WSA can be constructed in less time than our tree-based and array-based indexes, respectively. This is expected as our construction is much more complex than that of \WST and \WSA~\cite{DBLP:journals/iandc/BartonK0PR20,DBLP:journals/jea/Charalampopoulos20}. In particular, although our construction algorithm (see Theorem~\ref{the:2d-structure}) takes $\Theta(nz)$ time in any case, it carries higher constant factors than those of \WST and \WSA. This is expected as, in some sense, our construction largely follows the one of \WST and \WSA but it does additional work implied by the sampling mechanism.
In practice, the construction time decreases as $\ell$ increases and increases with $z$ -- see Lemmas~\ref{lem:minimizer tree size} and~\ref{lem:2demptiness}, which show a clear dependency on the number $\cO(\frac{nz}{\ell})$ of sampled minimizers. On average, \MWST requires $70\%$ (resp.~\MWSA requires $32\%$) %$41\%$) 
more time to be constructed than \WST (resp.~\WSA). \MWSTG and \MWSAG have similar construction time to \MWST and \MWSA, respectively. 

\subsection{Evaluating our Space-efficient Index Construction} 

\vspace{+1mm}
\noindent{\bf Construction Space.} Fig.~\ref{fig:construction_space_se} shows that the construction space of \MWSTSE is up to \emph{one order of magnitude smaller} than that of \WSA and \emph{$52$ times smaller} than that of \MWST on average. The construction space of \MWSTSE decreases with $\ell$ and increases with $z$, as expected by Theorem~\ref{the:space_efficient}. For example, in Fig.~\ref{fig:construction_space_se_efm_ell} \MWSTSE needs only 772MB  of memory to be constructed when $\ell=1024$, while \WSA and \MWST need over 32GBs and 183GBs, respectively. Even for $\ell=64$, the construction space of \MWSTSE is $4$ times smaller than that of \WSA and more than $25$ times smaller than that of \MWST. On the \RS datasets, \MWSTSE substantially outperformed all other indexes. Fig.~\ref{fig:construction_space_rssi} shows the results against the best competitor, \WSA. \MWSTSE needed $4$ times less space than \WSA on average. Although $z$ and $\ell$ appear in the space bound of Theorem~\ref{the:space_efficient}, and $\sigma$ does not, the impact of the latter is larger on the construction space. The reason is that the maximum resident set size we measure includes the memory for reading the dataset file, which increases with $\sigma$ for all indexes. Note that the construction space for both indexes scaled linearly with $n$, as can be seen in Fig.~\ref{fig:sen_vs_n_construction_space}. 

\vspace{+1mm}
\noindent{\bf Construction Time.} Fig.~\ref{fig:construction_time_se} shows that the construction time of \MWSTSE is on average \emph{$53\%$} smaller than that of \WSA, the next fastest index. This is \emph{very encouraging}, as \MWSTSE is quite complex. The construction time of \MWSTSE decreases with $\ell$ and increases with $z$, as expected by Theorem~\ref{the:space_efficient}.
For example, for $\ell=1024$ and $z=128$ in  Fig.~\ref{fig:construction_time_se_efm_ell}, the construction time of \MWSTSE is smaller by $31\%$ (resp.~$16$ times smaller) compared to that of \WSA (resp.~\MWST). This faster construction is a consequence of \WSA and \WST being always of $\Theta(nz)$ size (producing copies of solid factors), while in the extended solid factor trees each solid factor is considered only once. The extra $\cO(\log \ell)$ cost for heap operations is very optimized and in practice comparable with the large constants of the other constructions for reasonable $\ell$ values. 

On the \RS datasets, \MWSTSE substantially outperformed all other indexes. Fig.~\ref{fig:construction_time_rssi} shows the results against the best competitor, \WSA. Now the term $n\sigma=\cO(n)$ prevails in the construction time bound of Theorem~\ref{the:space_efficient}. Thus, the construction time increased linearly with $\sigma$ (Fig.~\ref{fig:sen_vs_sigma_construction_time}), while the impact of $\ell$ and $z$ was smaller. \MWSTSE scaled linearly with $n$ (Fig.~\ref{fig:sen_vs_n_construction_time}). The input datasets we used are in the order of GBs, thus their reading takes much of the construction time.

\subsection{Conclusion of our Experimental Evaluation} 

The most practical solution to $\ell$-\textsc{Weighted Indexing} is to use the \MWSTSE algorithm, which requires the smallest construction space and time (see Figs.~\ref{fig:construction_space_se}, \ref{fig:construction_space_rssi},~\ref{fig:construction_time_se}, and~\ref{fig:construction_time_rssi}) to construct \MWST and then infer \MWSA, the array-based version of \MWST via a standard 
in-order DFS traversal on \MWST\cite{DBLP:journals/siamcomp/ManberM93}, as \MWSA has the smallest index size and a competitive query time to \WSA (see Figs.~\ref{fig:index_size1},~\ref{fig:index_size2},~\ref{fig:query_time1}, and~\ref{fig:query_time2}).

\section{Discussion: Limitations and Future Work}\label{sec:discussion}

The size of our indexes is $\cO(\frac{nz}{\ell}\log z)$ \emph{in expectation} because minimizers have no worst-case guarantees; e.g., in string $\texttt{abcdefg}\ldots$, every position is a minimizer. This could be avoided by using \emph{difference covers}~\cite{DBLP:journals/jacm/KarkkainenSB06} instead. Difference covers are practical to construct, but slash the set of positions by $\sqrt{\ell}$ instead of $\ell$ (the factor that minimizers slash by in real-world datasets). As our focus is on practical schemes for reducing the space, we have resorted to minimizers. There are other sampling schemes slashing the set of positions by $\ell$ in the worst case~\cite{DBLP:conf/stoc/KempaK19}, such as \emph{string synchronizing sets}; these are however impractical to construct (see~\cite{DBLP:conf/esa/Dinklage0HKK20}, for example) and not directly applicable to the uncertain setting. 

Our methods do not explicitly take advantage of the $n$ probability distributions. The following may be explored to improve our indexes: (i) use the percentage $\Delta$ of positions where more than one letter has a probability of occurrence larger than $0$ as a complexity parameter; (ii) use lower or upper bounds on the letter probabilities depending on the underlying application. For instance, in sequencing data, a quality score is assigned to a single nucleotide in every position $i$, which is then translated to a probability $p$. The other 3 nucleotides can be assigned to probability $(1-p)/3$. However, based on the vicinity of the $i$th position or on domain knowledge, we could rather use upper or lower bounds for the remaining letters.

Our indexes rely explicitly or implicitly on the $z$-estimation of the input weighted string. This is a family $\mathcal{S}$ of $\lfloor z\rfloor$ strings of length $n$. Although the construction of $\mathcal{S}$ is combinatorially correct~\cite{DBLP:journals/iandc/BartonK0PR20} it is oblivious to domain knowledge. It may thus generate solid factors that are completely impossible to occur in a real-world dataset. For instance, in biology, such implausible factors are termed \emph{absent} or \emph{avoided}~\cite{DBLP:journals/almob/AlmirantisCGIMP17}. If one had this knowledge, they could easily trim these factors from the index by first querying them during a post-processing step. It seems much more challenging though to amend our data model to include this knowledge while constructing the index.

\section*{Acknowledgments}
This work was partially supported by the PANGAIA and ALPACA projects that have received
funding from the European Union’s Horizon 2020 research and innovation programme under the
Marie Skłodowska-Curie grant agreements No 872539 and 956229, respectively.
Wiktor Zuba received funding from the European Union's Horizon 2020 research and innovation
programme under the Marie Skłodowska-Curie grant agreement Grant Agreement No 101034253.

\bibliographystyle{plain}
\bibliography{references.bib}

\begin{thebibliography}{10}

\bibitem{footnote5}
\url{https://www.ncbi.nlm.nih.gov/nuccore/MN908947.3}.

\bibitem{footnote6}
\url{https://www.ncbi.nlm.nih.gov/pmc/articles/PMC8363274/bin/elife-66857-supp2.txt}.

\bibitem{footnote3}
\url{https://www.ncbi.nlm.nih.gov/nuccore/CP003351}.

\bibitem{footnote4}
\url{https://github.com/francesccoll/powerbacgwas/blob/main/data/efm_clade_all.vcf.gz}.

\bibitem{footnote7}
\url{https://www.ncbi.nlm.nih.gov/datasets/genome/GCF_000001405.13/}.

\bibitem{footnote8}
\url{https://ftp.1000genomes.ebi.ac.uk/vol1/ftp/release/20130502/ALL.chr22.phase3_shapeit2\_mvncall\_integrated\_v5b.20130502.genotypes.vcf.gz}.

\bibitem{talg}
Pankaj~K. Agarwal, Boris Aronov, Sariel Har{-}Peled, Jeff~M. Phillips, Ke~Yi,
  and Wuzhou Zhang.
\newblock Nearest-neighbor searching under uncertainty {II}.
\newblock {\em {ACM} Trans. Algorithms}, 13(1):3:1--3:25, 2016.

\bibitem{range3}
Pankaj~K. Agarwal, Siu{-}Wing Cheng, Yufei Tao, and Ke~Yi.
\newblock Indexing uncertain data.
\newblock In Jan Paredaens and Jianwen Su, editors, {\em Proceedings of the
  Twenty-Eigth {ACM} {SIGMOD-SIGACT-SIGART} Symposium on Principles of Database
  Systems, {PODS} 2009, June 19 - July 1, 2009, Providence, Rhode Island,
  {USA}}, pages 137--146. {ACM}, 2009.

\bibitem{DBLP:conf/icde/Aggarwal08}
Charu~C. Aggarwal.
\newblock On unifying privacy and uncertain data models.
\newblock In {\em Proceedings of the 24th International Conference on Data
  Engineering ({ICDE})}, pages 386--395. {IEEE} Computer Society, 2008.

\bibitem{DBLP:series/ads/Aggarwal09}
Charu~C. Aggarwal.
\newblock {\em Managing and Mining Uncertain Data}, volume~35 of {\em Advances
  in Database Systems}.
\newblock Kluwer, 2009.

\bibitem{DBLP:books/sp/mmsd2013}
Charu~C. Aggarwal, editor.
\newblock {\em Managing and Mining Sensor Data}.
\newblock Springer, 2013.

\bibitem{DBLP:journals/almob/AlmirantisCGIMP17}
Yannis Almirantis, Panagiotis Charalampopoulos, Jia Gao, Costas~S. Iliopoulos,
  Manal Mohamed, Solon~P. Pissis, and Dimitris Polychronopoulos.
\newblock On avoided words, absent words, and their application to biological
  sequence analysis.
\newblock {\em Algorithms Mol. Biol.}, 12(1):5:1--5:12, 2017.

\bibitem{DBLP:conf/cpm/AmirCIKZ06}
Amihood Amir, Eran Chencinski, Costas~S. Iliopoulos, Tsvi Kopelowitz, and Hui
  Zhang.
\newblock Property matching and weighted matching.
\newblock In Moshe Lewenstein and Gabriel Valiente, editors, {\em Combinatorial
  Pattern Matching, 17th Annual Symposium, {CPM} 2006, Barcelona, Spain, July
  5-7, 2006, Proceedings}, volume 4009 of {\em Lecture Notes in Computer
  Science}, pages 188--199. Springer, 2006.

\bibitem{DBLP:journals/tcs/AmirCIKZ08}
Amihood Amir, Eran Chencinski, Costas~S. Iliopoulos, Tsvi Kopelowitz, and Hui
  Zhang.
\newblock Property matching and weighted matching.
\newblock {\em Theor. Comput. Sci.}, 395(2-3):298--310, 2008.

\bibitem{longqueries2}
Mozhdeh Ariannezhad, Ali Montazeralghaem, Hamed Zamani, and Azadeh Shakery.
\newblock Improving retrieval performance for verbose queries via axiomatic
  analysis of term discrimination heuristic.
\newblock In Noriko Kando, Tetsuya Sakai, Hideo Joho, Hang Li, Arjen~P.
  de~Vries, and Ryen~W. White, editors, {\em Proceedings of the 40th
  International {ACM} {SIGIR} Conference on Research and Development in
  Information Retrieval, Shinjuku, Tokyo, Japan, August 7-11, 2017}, pages
  1201--1204. {ACM}, 2017.

\bibitem{DBLP:journals/pvldb/AyadLP23}
Lorraine A.~K. Ayad, Grigorios Loukides, and Solon~P. Pissis.
\newblock Text indexing for long patterns: Anchors are all you need.
\newblock {\em Proc. {VLDB} Endow.}, 16(9):2117--2131, 2023.

\bibitem{DBLP:journals/iandc/BartonK0PR20}
Carl Barton, Tomasz Kociumaka, Chang Liu, Solon~P. Pissis, and Jakub
  Radoszewski.
\newblock Indexing weighted sequences: {Neat} and efficient.
\newblock {\em Inf. Comput.}, 270, 2020.

\bibitem{DBLP:conf/cpm/BartonKPR16}
Carl Barton, Tomasz Kociumaka, Solon~P. Pissis, and Jakub Radoszewski.
\newblock Efficient index for weighted sequences.
\newblock In Roberto Grossi and Moshe Lewenstein, editors, {\em 27th Annual
  Symposium on Combinatorial Pattern Matching, {CPM} 2016, June 27-29, 2016,
  Tel Aviv, Israel}, volume~54 of {\em LIPIcs}, pages 4:1--4:13. Schloss
  Dagstuhl - Leibniz-Zentrum f{\"{u}}r Informatik, 2016.

\bibitem{longqueries1}
Michael Bendersky and W.~Bruce Croft.
\newblock Discovering key concepts in verbose queries.
\newblock In Sung{-}Hyon Myaeng, Douglas~W. Oard, Fabrizio Sebastiani,
  Tat{-}Seng Chua, and Mun{-}Kew Leong, editors, {\em Proceedings of the 31st
  Annual International {ACM} {SIGIR} Conference on Research and Development in
  Information Retrieval, {SIGIR} 2008, Singapore, July 20-24, 2008}, pages
  491--498. {ACM}, 2008.

\bibitem{DBLP:conf/edbt/BiswasPTS16}
Sudip Biswas, Manish Patil, Sharma~V. Thankachan, and Rahul Shah.
\newblock Probabilistic threshold indexing for uncertain strings.
\newblock In Evaggelia Pitoura, Sofian Maabout, Georgia Koutrika, Am{\'{e}}lie
  Marian, Letizia Tanca, Ioana Manolescu, and Kostas Stefanidis, editors, {\em
  Proceedings of the 19th International Conference on Extending Database
  Technology, {EDBT} 2016, Bordeaux, France, March 15-16, 2016, Bordeaux,
  France, March 15-16, 2016}, pages 401--412. OpenProceedings.org, 2016.

\bibitem{DBLP:journals/pvldb/Boncz0L20}
Peter~A. Boncz, Thomas Neumann, and Viktor Leis.
\newblock {FSST:} fast random access string compression.
\newblock {\em Proc. {VLDB} Endow.}, 13(11):2649--2661, 2020.

\bibitem{DBLP:conf/compgeom/ChanLP11}
Timothy~M. Chan, Kasper~Green Larsen, and Mihai P{u{a}}tra{c{s}}cu.
\newblock Orthogonal range searching on the {RAM}, revisited.
\newblock In Ferran Hurtado and Marc~J. van Kreveld, editors, {\em Proceedings
  of the 27th {ACM} Symposium on Computational Geometry, Paris, France, June
  13-15, 2011}, pages 1--10. {ACM}, 2011.

\bibitem{DBLP:journals/jea/Charalampopoulos20}
Panagiotis Charalampopoulos, Costas~S. Iliopoulos, Chang Liu, and Solon~P.
  Pissis.
\newblock Property suffix array with applications in indexing weighted
  sequences.
\newblock {\em {ACM} J. Exp. Algorithmics}, 25:1--16, 2020.

\bibitem{DBLP:journals/iandc/Charalampopoulos19}
Panagiotis Charalampopoulos, Costas~S. Iliopoulos, Solon~P. Pissis, and Jakub
  Radoszewski.
\newblock On-line weighted pattern matching.
\newblock {\em Inf. Comput.}, 266:49--59, 2019.

\bibitem{DBLP:journals/vldb/ChenGZJCZ17}
Lu~Chen, Yunjun Gao, Aoxiao Zhong, Christian~S. Jensen, Gang Chen, and Baihua
  Zheng.
\newblock Indexing metric uncertain data for range queries and range joins.
\newblock {\em {VLDB} J.}, 26(4):585--610, 2017.

\bibitem{edbt_nn}
Reynold Cheng, Lei Chen, Jinchuan Chen, and Xike Xie.
\newblock Evaluating probability threshold k-nearest-neighbor queries over
  uncertain data.
\newblock In Martin~L. Kersten, Boris Novikov, Jens Teubner, Vladimir Polutin,
  and Stefan Manegold, editors, {\em {EDBT} 2009, 12th International Conference
  on Extending Database Technology, Saint Petersburg, Russia, March 24-26,
  2009, Proceedings}, volume 360 of {\em {ACM} International Conference
  Proceeding Series}, pages 672--683. {ACM}, 2009.

\bibitem{tkde_nn}
Reynold Cheng, Dmitri~V. Kalashnikov, and Sunil Prabhakar.
\newblock Querying imprecise data in moving object environments.
\newblock {\em {IEEE} Trans. Knowl. Data Eng.}, 16(9):1112--1127, 2004.

\bibitem{rangereport}
Reynold Cheng, Yuni Xia, Sunil Prabhakar, Rahul Shah, and Jeffrey~Scott Vitter.
\newblock Efficient indexing methods for probabilistic threshold queries over
  uncertain data.
\newblock In Mario~A. Nascimento, M.~Tamer {\"{O}}zsu, Donald Kossmann,
  Ren{\'{e}}e~J. Miller, Jos{\'{e}}~A. Blakeley, and K.~Bernhard Schiefer,
  editors, {\em (e)Proceedings of the Thirtieth International Conference on
  Very Large Data Bases, {VLDB} 2004, Toronto, Canada, August 31 - September 3
  2004}, pages 876--887. Morgan Kaufmann, 2004.

\bibitem{EFM}
Francesc Coll, Theodore Gouliouris, Sebastian Bruchmann, Jody Phelan, Kathy~E.
  Raven, Taane~G. Clark, Julian Parkhill, and Sharon~J. Peacock.
\newblock {PowerBacGWAS: a computational pipeline to perform power calculations
  for bacterial genome-wide association studies}.
\newblock {\em Communications Biology}, 5(266), 2022.

\bibitem{DBLP:books/daglib/0023376}
Thomas~H. Cormen, Charles~E. Leiserson, Ronald~L. Rivest, and Clifford Stein.
\newblock {\em Introduction to Algorithms, 3rd Edition}.
\newblock {MIT} Press, 2009.

\bibitem{DBLP:books/daglib/0020103}
Maxime Crochemore, Christophe Hancart, and Thierry Lecroq.
\newblock {\em Algorithms on strings}.
\newblock Cambridge University Press, 2007.

\bibitem{DBLP:conf/ssd/DaiYMTV05}
Xiangyuan Dai, Man~Lung Yiu, Nikos Mamoulis, Yufei Tao, and Michail Vaitis.
\newblock Probabilistic spatial queries on existentially uncertain data.
\newblock In Claudia~Bauzer Medeiros, Max~J. Egenhofer, and Elisa Bertino,
  editors, {\em Advances in Spatial and Temporal Databases, 9th International
  Symposium, {SSTD} 2005, Angra dos Reis, Brazil, August 22-24, 2005,
  Proceedings}, volume 3633 of {\em Lecture Notes in Computer Science}, pages
  400--417. Springer, 2005.

\bibitem{DBLP:conf/esa/Dinklage0HKK20}
Patrick Dinklage, Johannes Fischer, Alexander Herlez, Tomasz Kociumaka, and
  Florian Kurpicz.
\newblock Practical performance of space efficient data structures for longest
  common extensions.
\newblock In Fabrizio Grandoni, Grzegorz Herman, and Peter Sanders, editors,
  {\em 28th Annual European Symposium on Algorithms, {ESA} 2020, September 7-9,
  2020, Pisa, Italy (Virtual Conference)}, volume 173 of {\em LIPIcs}, pages
  39:1--39:20. Schloss Dagstuhl - Leibniz-Zentrum f{\"{u}}r Informatik, 2020.

\bibitem{1000genomes}
Susan Fairley, Ernesto Lowy-Gallego, Emily Perry, and Paul Flicek.
\newblock {The International Genome Sample Resource (IGSR) collection of open
  human genomic variation resources}.
\newblock {\em Nucleic Acids Research}, 48(D1):D941--D947, 10 2019.

\bibitem{DBLP:conf/focs/Farach97}
Martin Farach.
\newblock Optimal suffix tree construction with large alphabets.
\newblock In {\em 38th Annual Symposium on Foundations of Computer Science,
  {FOCS} '97, Miami Beach, Florida, USA, October 19-22, 1997}, pages 137--143,
  1997.

\bibitem{DBLP:journals/jacm/FerraginaM05}
Paolo Ferragina and Giovanni Manzini.
\newblock Indexing compressed text.
\newblock {\em J. {ACM}}, 52(4):552--581, 2005.

\bibitem{DBLP:journals/jacm/GagieNP20}
Travis Gagie, Gonzalo Navarro, and Nicola Prezza.
\newblock Fully functional suffix trees and optimal text searching in
  {BWT}-runs bounded space.
\newblock {\em J. {ACM}}, 67(1):2:1--2:54, 2020.

\bibitem{DBLP:journals/pvldb/GeL11}
Tingjian Ge and Zheng Li.
\newblock Approximate substring matching over uncertain strings.
\newblock {\em Proc. {VLDB} Endow.}, 4(11):772--782, 2011.

\bibitem{xml_index}
Jian Gong, Reynold Cheng, and David~W. Cheung.
\newblock Efficient management of uncertainty in {XML} schema matching.
\newblock {\em {VLDB} J.}, 21(3):385--409, 2012.

\bibitem{NGS}
Sara Goodwin, John~D. McPherson, and W.~Richard McCombie.
\newblock Coming of age: ten years of next-generation sequencing technologies.
\newblock {\em Nature Reviews Genetics}, 17:333--351, 2016.

\bibitem{DBLP:journals/spe/GrabowskiR17}
Szymon Grabowski and Marcin Raniszewski.
\newblock Sampled suffix array with minimizers.
\newblock {\em Softw. Pract. Exp.}, 47(11):1755--1771, 2017.

\bibitem{longqueries0}
Manish Gupta and Michael Bendersky.
\newblock Information retrieval with verbose queries.
\newblock In Ricardo Baeza{-}Yates, Mounia Lalmas, Alistair Moffat, and
  Berthier~A. Ribeiro{-}Neto, editors, {\em Proceedings of the 38th
  International {ACM} {SIGIR} Conference on Research and Development in
  Information Retrieval, Santiago, Chile, August 9-13, 2015}, pages 1121--1124.
  {ACM}, 2015.

\bibitem{webpageduplicate}
Monika~Rauch Henzinger.
\newblock Finding near-duplicate web pages: a large-scale evaluation of
  algorithms.
\newblock In Efthimis~N. Efthimiadis, Susan~T. Dumais, David Hawking, and
  Kalervo J{\"{a}}rvelin, editors, {\em {SIGIR} 2006: Proceedings of the 29th
  Annual International {ACM} {SIGIR} Conference on Research and Development in
  Information Retrieval, Seattle, Washington, USA, August 6-11, 2006}, pages
  284--291. {ACM}, 2006.

\bibitem{topk2}
Ming Hua, Jian Pei, and Xuemin Lin.
\newblock Ranking queries on uncertain data.
\newblock {\em {VLDB} J.}, 20(1):129--153, 2011.

\bibitem{DBLP:journals/fuin/IliopoulosMPPTT06}
Costas~S. Iliopoulos, Christos Makris, Yannis Panagis, Katerina Perdikuri,
  Evangelos Theodoridis, and Athanasios~K. Tsakalidis.
\newblock The weighted suffix tree: An efficient data structure for handling
  molecular weighted sequences and its applications.
\newblock {\em Fundam. Informaticae}, 71(2-3):259--277, 2006.

\bibitem{Winnowmap2}
Chirag Jain, Arang Rhie, Nancy Hansen, Sergey Koren, and Adam~M. Phillippy.
\newblock Long-read mapping to repetitive reference sequences using winnowmap2.
\newblock {\em Nat Methods}, 19:705--710, 2022.

\bibitem{DBLP:conf/sigmod/JestesLYY10}
Jeffrey Jestes, Feifei Li, Zhepeng Yan, and Ke~Yi.
\newblock Probabilistic string similarity joins.
\newblock In Ahmed~K. Elmagarmid and Divyakant Agrawal, editors, {\em
  Proceedings of the {ACM} {SIGMOD} International Conference on Management of
  Data, {SIGMOD} 2010, Indianapolis, Indiana, USA, June 6-10, 2010}, pages
  327--338. {ACM}, 2010.

\bibitem{machinegenerated}
Jiaojiao Jiang, Steve Versteeg, Jun Han, Md.~Arafat Hossain, Jean{-}Guy
  Schneider, Christopher Leckie, and Zeinab Farahmandpour.
\newblock P-gram: Positional n-gram for the clustering of machine-generated
  messages.
\newblock {\em {IEEE} Access}, 7:88504--88516, 2019.

\bibitem{DBLP:books/lib/JurafskyM09}
Dan Jurafsky and James~H. Martin.
\newblock {\em Speech and language processing: an introduction to natural
  language processing, computational linguistics, and speech recognition, 2nd
  Edition}.
\newblock Prentice Hall series in artificial intelligence. Prentice Hall,
  Pearson Education International, 2009.

\bibitem{DBLP:conf/sigmod/KanagalD09}
Bhargav Kanagal and Amol Deshpande.
\newblock Indexing correlated probabilistic databases.
\newblock In Ugur {\c{C}}etintemel, Stanley~B. Zdonik, Donald Kossmann, and
  Nesime Tatbul, editors, {\em Proceedings of the {ACM} {SIGMOD} International
  Conference on Management of Data, {SIGMOD} 2009, Providence, Rhode Island,
  USA, June 29 - July 2, 2009}, pages 455--468. {ACM}, 2009.

\bibitem{DBLP:journals/jacm/KarkkainenSB06}
Juha K{\"{a}}rkk{\"{a}}inen, Peter Sanders, and Stefan Burkhardt.
\newblock Linear work suffix array construction.
\newblock {\em J. {ACM}}, 53(6):918--936, 2006.

\bibitem{DBLP:journals/ibmrd/KarpR87}
Richard~M. Karp and Michael~O. Rabin.
\newblock Efficient randomized pattern-matching algorithms.
\newblock {\em {IBM} J. Res. Dev.}, 31(2):249--260, 1987.

\bibitem{DBLP:conf/cpm/KasaiLAAP01}
Toru Kasai, Gunho Lee, Hiroki Arimura, Setsuo Arikawa, and Kunsoo Park.
\newblock Linear-time longest-common-prefix computation in suffix arrays and
  its applications.
\newblock In {\em Combinatorial Pattern Matching, 12th Annual Symposium, {CPM}
  2001 Jerusalem, Israel, July 1-4, 2001 Proceedings}, pages 181--192, 2001.

\bibitem{PWM}
A.E. Kel, E.~Gössling, I.~Reuter, E.~Cheremushkin, O.V. Kel-Margoulis, and
  E.~Wingender.
\newblock {MATCHTM: a tool for searching transcription factor binding sites in
  DNA sequences}.
\newblock {\em Nucleic Acids Research}, 31(13):3576--3579, 07 2003.

\bibitem{DBLP:conf/stoc/KempaK19}
Dominik Kempa and Tomasz Kociumaka.
\newblock String synchronizing sets: sublinear-time {BWT} construction and
  optimal {LCE} data structure.
\newblock In Moses Charikar and Edith Cohen, editors, {\em Proceedings of the
  51st Annual {ACM} {SIGACT} Symposium on Theory of Computing, {STOC} 2019,
  Phoenix, AZ, USA, June 23-26, 2019}, pages 756--767. {ACM}, 2019.

\bibitem{DBLP:conf/soda/KempaK23}
Dominik Kempa and Tomasz Kociumaka.
\newblock Breaking the {O}(n)-barrier in the construction of compressed suffix
  arrays and suffix trees.
\newblock In Nikhil Bansal and Viswanath Nagarajan, editors, {\em Proceedings
  of the 2023 {ACM-SIAM} Symposium on Discrete Algorithms, {SODA} 2023,
  Florence, Italy, January 22-25, 2023}, pages 5122--5202. {SIAM}, 2023.

\bibitem{DBLP:conf/isaac/KociumakaPR16}
Tomasz Kociumaka, Solon~P. Pissis, and Jakub Radoszewski.
\newblock Pattern matching and consensus problems on weighted sequences and
  profiles.
\newblock In Seok{-}Hee Hong, editor, {\em 27th International Symposium on
  Algorithms and Computation, {ISAAC} 2016, December 12-14, 2016, Sydney,
  Australia}, volume~64 of {\em LIPIcs}, pages 46:1--46:12. Schloss Dagstuhl -
  Leibniz-Zentrum f{\"{u}}r Informatik, 2016.

\bibitem{DBLP:journals/mst/KociumakaPR19}
Tomasz Kociumaka, Solon~P. Pissis, and Jakub Radoszewski.
\newblock Pattern matching and consensus problems on weighted sequences and
  profiles.
\newblock {\em Theory Comput. Syst.}, 63(3):506--542, 2019.

\bibitem{DBLP:journals/bioinformatics/KorhonenMPRU09}
Janne~H. Korhonen, Petri Martinm{\"{a}}ki, Cinzia Pizzi, Pasi Rastas, and Esko
  Ukkonen.
\newblock {MOODS:} fast search for position weight matrix matches in {DNA}
  sequences.
\newblock {\em Bioinform.}, 25(23):3181--3182, 2009.

\bibitem{DBLP:journals/tcs/LandauV86}
Gad~M. Landau and Uzi Vishkin.
\newblock Efficient string matching with k mismatches.
\newblock {\em Theor. Comput. Sci.}, 43:239--249, 1986.

\bibitem{DBLP:conf/sdm/LiBKP14}
Yuxuan Li, James Bailey, Lars Kulik, and Jian Pei.
\newblock Efficient matching of substrings in uncertain sequences.
\newblock In Mohammed~Javeed Zaki, Zoran Obradovic, Pang{-}Ning Tan, Arindam
  Banerjee, Chandrika Kamath, and Srinivasan Parthasarathy, editors, {\em
  Proceedings of the 2014 {SIAM} International Conference on Data Mining,
  Philadelphia, Pennsylvania, USA, April 24-26, 2014}, pages 767--775. {SIAM},
  2014.

\bibitem{Logsdon2020}
Glennis~A. Logsdon, Mitchell~R. Vollger, and Evan~E. Eichler.
\newblock Long-read human genome sequencing and its applications.
\newblock {\em Nat. Rev. Genet.}, 21(10):597--614, 2020.

\bibitem{DBLP:conf/esa/LoukidesP21}
Grigorios Loukides and Solon~P. Pissis.
\newblock Bidirectional string anchors: {A} new string sampling mechanism.
\newblock In Petra Mutzel, Rasmus Pagh, and Grzegorz Herman, editors, {\em 29th
  Annual European Symposium on Algorithms, {ESA} 2021, September 6-8, 2021,
  Lisbon, Portugal (Virtual Conference)}, volume 204 of {\em LIPIcs}, pages
  64:1--64:21. Schloss Dagstuhl - Leibniz-Zentrum f{\"{u}}r Informatik, 2021.

\bibitem{DBLP:conf/latin/MakinenN06}
Veli M{\"{a}}kinen and Gonzalo Navarro.
\newblock Position-restricted substring searching.
\newblock In Jos{\'{e}}~R. Correa, Alejandro Hevia, and Marcos~A. Kiwi,
  editors, {\em {LATIN} 2006: Theoretical Informatics, 7th Latin American
  Symposium, Valdivia, Chile, March 20-24, 2006, Proceedings}, volume 3887 of
  {\em Lecture Notes in Computer Science}, pages 703--714. Springer, 2006.

\bibitem{DBLP:journals/siamcomp/ManberM93}
Udi Manber and Eugene~W. Myers.
\newblock Suffix arrays: {A} new method for on-line string searches.
\newblock {\em {SIAM} J. Comput.}, 22(5):935--948, 1993.

\bibitem{representativephrases1}
Olena Medelyan and Ian~H. Witten.
\newblock Thesaurus based automatic keyphrase indexing.
\newblock In Gary Marchionini, Michael~L. Nelson, and Catherine~C. Marshall,
  editors, {\em {ACM/IEEE} Joint Conference on Digital Libraries, {JCDL} 2006,
  Chapel Hill, NC, USA, June 11-15, 2006, Proceedings}, pages 296--297. {ACM},
  2006.

\bibitem{DBLP:journals/jacm/Morrison68}
Donald~R. Morrison.
\newblock {PATRICIA} - practical algorithm to retrieve information coded in
  alphanumeric.
\newblock {\em J. {ACM}}, 15(4):514--534, 1968.

\bibitem{DBLP:conf/edbt/0002RF14}
Ingo M{\"{u}}ller, Cornelius Ratsch, and Franz F{\"{a}}rber.
\newblock Adaptive string dictionary compression in in-memory column-store
  database systems.
\newblock In Sihem Amer{-}Yahia, Vassilis Christophides, Anastasios
  Kementsietsidis, Minos~N. Garofalakis, Stratos Idreos, and Vincent Leroy,
  editors, {\em Proceedings of the 17th International Conference on Extending
  Database Technology, {EDBT} 2014, Athens, Greece, March 24-28, 2014}, pages
  283--294. OpenProceedings.org, 2014.

\bibitem{rssi_data}
Claro Noda, Shashi Prabh, Mário Alves, Thiemo Voigt, and Carlo~Alberto Boano.
\newblock {CRAWDAD} cister/rssi.
\newblock \url{https://dx.doi.org/10.15783/C7WC75}, 2022.

\bibitem{DBLP:journals/tcbb/PizziRU11}
Cinzia Pizzi, Pasi Rastas, and Esko Ukkonen.
\newblock Finding significant matches of position weight matrices in linear
  time.
\newblock {\em {IEEE} {ACM} Trans. Comput. Biol. Bioinform.}, 8(1):69--79,
  2011.

\bibitem{DBLP:conf/sigmod/QiJSP10}
Yinian Qi, Rohit Jain, Sarvjeet Singh, and Sunil Prabhakar.
\newblock Threshold query optimization for uncertain data.
\newblock In Ahmed~K. Elmagarmid and Divyakant Agrawal, editors, {\em
  Proceedings of the {ACM} {SIGMOD} International Conference on Management of
  Data, {SIGMOD} 2010, Indianapolis, Indiana, USA, June 6-10, 2010}, pages
  315--326. {ACM}, 2010.

\bibitem{topk4}
Niranjan Rai and Xiang Lian.
\newblock Distributed probabilistic top-k dominating queries over uncertain
  databases.
\newblock {\em Knowl. Inf. Syst.}, 65(11):4939--4965, 2023.

\bibitem{DBLP:journals/bioinformatics/RobertsHHMY04}
Michael Roberts, Wayne Hayes, Brian~R. Hunt, Stephen~M. Mount, and James~A.
  Yorke.
\newblock Reducing storage requirements for biological sequence comparison.
\newblock {\em Bioinform.}, 20(18):3363--3369, 2004.

\bibitem{DBLP:journals/nar/Rodriguez-TomeSCF96}
Patricia Rodriguez{-}Tom{\'{e}}, Peter Stoehr, Graham Cameron, and Tomas~P.
  Flores.
\newblock The european bioinformatics institute {(EBI)} databases.
\newblock {\em Nucleic Acids Res.}, 24(1):6--12, 1996.

\bibitem{DBLP:conf/sigmod/SchleimerWA03}
Saul Schleimer, Daniel~Shawcross Wilkerson, and Alexander Aiken.
\newblock Winnowing: Local algorithms for document fingerprinting.
\newblock In Alon~Y. Halevy, Zachary~G. Ives, and AnHai Doan, editors, {\em
  Proceedings of the 2003 {ACM} {SIGMOD} International Conference on Management
  of Data, San Diego, California, USA, June 9-12, 2003}, pages 76--85. {ACM},
  2003.

\bibitem{petq}
Sarvjeet Singh, Chris Mayfield, Sunil Prabhakar, Rahul Shah, and Susanne~E.
  Hambrusch.
\newblock Indexing uncertain categorical data.
\newblock In Rada Chirkova, Asuman Dogac, M.~Tamer {\"{O}}zsu, and Timos~K.
  Sellis, editors, {\em Proceedings of the 23rd International Conference on
  Data Engineering, {ICDE} 2007, The Marmara Hotel, Istanbul, Turkey, April
  15-20, 2007}, pages 616--625. {IEEE} Computer Society, 2007.

\bibitem{DBLP:conf/icde/SinghMSPHNC08}
Sarvjeet Singh, Chris Mayfield, Rahul Shah, Sunil Prabhakar, Susanne~E.
  Hambrusch, Jennifer Neville, and Reynold Cheng.
\newblock Database support for probabilistic attributes and tuples.
\newblock In Gustavo Alonso, Jos{\'{e}}~A. Blakeley, and Arbee L.~P. Chen,
  editors, {\em Proceedings of the 24th International Conference on Data
  Engineering, {ICDE} 2008, April 7-12, 2008, Canc{\'{u}}n, Mexico}, pages
  1053--1061. {IEEE} Computer Society, 2008.

\bibitem{graph2}
Zitan Sun, Xin Huang, Jianliang Xu, and Francesco Bonchi.
\newblock Efficient probabilistic truss indexing on uncertain graphs.
\newblock In Jure Leskovec, Marko Grobelnik, Marc Najork, Jie Tang, and Leila
  Zia, editors, {\em {WWW} '21: The Web Conference 2021, Virtual Event /
  Ljubljana, Slovenia, April 19-23, 2021}, pages 354--366. {ACM} / {IW3C2},
  2021.

\bibitem{tao1}
Yufei Tao, Reynold Cheng, Xiaokui Xiao, Wang~Kay Ngai, Ben Kao, and Sunil
  Prabhakar.
\newblock Indexing multi-dimensional uncertain data with arbitrary probability
  density functions.
\newblock In Klemens B{\"{o}}hm, Christian~S. Jensen, Laura~M. Haas, Martin~L.
  Kersten, Per{-}{\AA}ke Larson, and Beng~Chin Ooi, editors, {\em Proceedings
  of the 31st International Conference on Very Large Data Bases, Trondheim,
  Norway, August 30 - September 2, 2005}, pages 922--933. {ACM}, 2005.

\bibitem{tao2}
Yufei Tao, Xiaokui Xiao, and Reynold Cheng.
\newblock Range search on multidimensional uncertain data.
\newblock {\em {ACM} Trans. Database Syst.}, 32(3):15, 2007.

\bibitem{PanGenomeConsortium18}
{The Computational Pan{-}Genomics Consortium}.
\newblock Computational pan-genomics: status, promises and challenges.
\newblock {\em Briefings in Bioinformatics}, 19(1):118--135, 2018.

\bibitem{SARS}
Gerry Tonkin-Hill, Inigo Martincorena, Roberto Amato, Andrew~RJ Lawson, Moritz
  Gerstung, Ian Johnston, David~K Jackson, Naomi Park, Stefanie~V Lensing,
  Michael~A Quail, Sónia Gonçalves, Cristina Ariani, Michael Spencer~Chapman,
  William~L Hamilton, Luke~W Meredith, Grant Hall, Aminu~S Jahun, Yasmin
  Chaudhry, Myra Hosmillo, Malte~L Pinckert, Iliana Georgana, Anna Yakovleva,
  Laura~G Caller, Sarah~L Caddy, Theresa Feltwell, Fahad~A Khokhar, Charlotte~J
  Houldcroft, Martin~D Curran, Surendra Parmar, The COVID-19 Genomics UK
  (COG-UK) Consortium, Alex Alderton, Rachel Nelson, Ewan~M Harrison, John
  Sillitoe, Stephen~D Bentley, Jeffrey~C Barrett, M~Estee Torok, Ian~G
  Goodfellow, Cordelia Langford, Dominic Kwiatkowski, and Wellcome Sanger
  Institute COVID-19~Surveillance Team.
\newblock Patterns of within-host genetic diversity in {SARS-CoV-2}.
\newblock {\em eLife}, 10:e66857, aug 2021.

\bibitem{TOIS17}
Kazutoshi Umemoto, Ruihua Song, Jian{-}Yun Nie, Xing Xie, Katsumi Tanaka, and
  Yong Rui.
\newblock Search by screenshots for universal article clipping in mobile apps.
\newblock {\em {ACM} Trans. Inf. Syst.}, 35(4):34:1--34:29, 2017.

\bibitem{error_RSSI}
Angelos Vlavianos, Lap~Kong Law, Ioannis Broustis, Srikanth~V. Krishnamurthy,
  and Michalis Faloutsos.
\newblock Assessing link quality in {IEEE} 802.11 wireless networks: Which is
  the right metric?
\newblock In {\em Proceedings of the {IEEE} 19th International Symposium on
  Personal, Indoor and Mobile Radio Communications, {PIMRC} 2008, 15-18
  September 2008, Cannes, French Riviera, France}, pages 1--6. {IEEE}, 2008.

\bibitem{DBLP:conf/sigmod/VogelsgesangHFK18}
Adrian Vogelsgesang, Michael Haubenschild, Jan Finis, Alfons Kemper, Viktor
  Leis, Tobias M{\"{u}}hlbauer, Thomas Neumann, and Manuel Then.
\newblock Get real: How benchmarks fail to represent the real world.
\newblock In Alexander B{\"{o}}hm and Tilmann Rabl, editors, {\em Proceedings
  of the 7th International Workshop on Testing Database Systems, DBTest@SIGMOD
  2018, Houston, TX, USA, June 15, 2018}, pages 1:1--1:6. {ACM}, 2018.

\bibitem{DBLP:conf/focs/Weiner73}
Peter Weiner.
\newblock Linear pattern matching algorithms.
\newblock In {\em 14th Annual Symposium on Switching and Automata Theory, Iowa
  City, Iowa, USA, October 15-17, 1973}, pages 1--11, 1973.

\bibitem{Wenger2019}
Aaron~M. Wenger et~al.
\newblock Accurate circular consensus long-read sequencing improves variant
  detection and assembly of a human genome.
\newblock {\em Nat. Biotechnol.}, 37:1155--1162, 2019.

\bibitem{graph1}
Bohua Yang, Dong Wen, Lu~Qin, Ying Zhang, Lijun Chang, and Rong{-}Hua Li.
\newblock Index-based optimal algorithm for computing k-cores in large
  uncertain graphs.
\newblock In {\em 35th {IEEE} International Conference on Data Engineering,
  {ICDE} 2019, Macao, China, April 8-11, 2019}, pages 64--75. {IEEE}, 2019.

\bibitem{topk1}
Ke~Yi, Feifei Li, George Kollios, and Divesh Srivastava.
\newblock Efficient processing of top-k queries in uncertain databases with
  x-relations.
\newblock {\em {IEEE} Trans. Knowl. Data Eng.}, 20(12):1669--1682, 2008.

\bibitem{topk3}
Liming Zhan, Ying Zhang, Wenjie Zhang, and Xuemin Lin.
\newblock Identifying top k dominating objects over uncertain data.
\newblock In Sourav~S. Bhowmick, Curtis~E. Dyreson, Christian~S. Jensen,
  Mong{-}Li Lee, Agus Muliantara, and Bernhard Thalheim, editors, {\em Database
  Systems for Advanced Applications - 19th International Conference, {DASFAA}
  2014, Bali, Indonesia, April 21-24, 2014. Proceedings, Part {I}}, volume 8421
  of {\em Lecture Notes in Computer Science}, pages 388--405. Springer, 2014.

\bibitem{10.1093/bioinformatics/btaa472}
Hongyu Zheng, Carl Kingsford, and Guillaume Marçais.
\newblock {Improved design and analysis of practical minimizers}.
\newblock {\em Bioinformatics}, 36(Supplement\_1):i119--i127, 07 2020.

\end{thebibliography}

\end{document}